\documentclass[oneside,english]{amsart}
\usepackage[T1]{fontenc}
\usepackage[latin9]{inputenc}
\usepackage{geometry}
\usepackage{amstext}
\usepackage{amsthm}
\usepackage{amssymb}
\usepackage{graphicx}
\PassOptionsToPackage{normalem}{ulem}
\usepackage{ulem}

\makeatletter
\numberwithin{equation}{section}
\numberwithin{figure}{section}
\theoremstyle{plain}
\newtheorem{thm}{\protect\theoremname}[section]
\theoremstyle{remark}
\newtheorem{rem}[thm]{\protect\remarkname}
\theoremstyle{plain}
\newtheorem{lem}[thm]{\protect\lemmaname}
\ifx\proof\undefined
\newenvironment{proof}[1][\protect\proofname]{\par
	\normalfont\topsep6\p@\@plus6\p@\relax
	\trivlist
	\itemindent\parindent
	\item[\hskip\labelsep\scshape #1]\ignorespaces
}{%
	\endtrivlist\@endpefalse
}
\providecommand{\proofname}{Proof}
\fi

\makeatother

\usepackage{babel}
\providecommand{\lemmaname}{Lemma}
\providecommand{\remarkname}{Remark}
\providecommand{\theoremname}{Theorem}
\newcommand{\dd}{\mathrm{d}}

\begin{document}
\title[Asymptotic recovery of the net magnetisation of a bounded sample]{Magnetisation moment of a bounded 3D sample: \\asymptotic recovery
from planar measurements on a large disk}
\author{Dmitry Ponomarev $^{1,2}$}

\begin{abstract}
Inverse magnetisation problem consists in inferring information about a magnetic source from measurements of its magnetic field.
Unlike a general magnetisation distribution, the total magnetisation (net moment) of the source is a quantity that theoretically can be uniquely determined from the field. At the same time, it is often the most useful quantity for practical applications (on large and small scales) such as detection of a magnetic anomaly in magnetic prospection problem or finding the overall strength and mean direction of the magnetisation distribution of a magnetised rock sample.
It is known that the net moment components can be explicitly estimated using the so-called Helbig's integrals which involve integration of the magnetic field data on the plane against simple polynomials. Evaluation of these integrals requires knowledge of the magnetic field data on a large region or the use of ad hoc methods to compensate for the lack thereof.
In this paper, we derive higher-order analogs of Helbig's integrals which permit estimation of total magnetisation components in terms of measurement data available on a smaller region. 
Motivated by a concrete experimental setup for analysing remanent magnetisation of rock samples with a scanning microscope, we also extend Helbig's integrals to the situation when knowledge of only one field component is necessary. Moreover, apart from derivation of these novel formulas, we rigorously prove their accuracy.
The presented approach, based on an appropriate splitting in the Fourier domain and estimates of oscillatory integrals (involving both small and large parameters), elucidates the derivation of asymptotic formulas for the net moment components to an arbitrary order, a possibility that was previously unclear.
The obtained results are illustrated numerically and their robustness with respect to the noise is discussed.

\end{abstract}

\maketitle

\section{Introduction\label{sec:intro}}

\footnotetext[1]{FACTAS team, Centre Inria d'Universit{\'e} C{\^o}te d'Azur, France}
\footnotetext[2]{Contact: dmitry.ponomarev@inria.fr} 

Constant advances in magnetometry allow measurements of magnetic fields
of very low intensities with high spatial resolution. In particular,
this opens new horizons in paleomagnetic contexts. Ancient rocks and
meteorites possess remanent magnetisation and thus might preserve
valuable records of a past magnetic activity on Earth and other planets,
asteroids and satellites. Extraction of this relict magnetic information
is a lucrative but challenging task. Deducing magnetisation of a geosample
hinges on effective processing of the measurements of the magnetic
field available in a nearest neighbourhood of the sample since the
informative part of the field further away is very weak and significantly
deteriorated by noise. In particular set-ups of scanning SQUID (Superconducting
Quantum Interference Device) magnetometer or QDM (Quantum Diamond
Microscope), measurements are available in a planar area above the
sample, in a close vicinity of it, and such measurements typically
feature only one component of the magnetic field. This is in contrast
with more common settings that deal with magnetic fields of higher
intensity and hence could, on a methodological level, rely on the
classical dipolar approximation of a sample valid in a far-away region. 

In the present work, we are concerned with recovery of the total 
magnetisation (the so-called net moment) of a sample rather than dealing with reconstruction of the magnetisation distribution. While both inverse problems are known to be ill-posed (as an inverse source problem for elliptic partial-differential equations) due to the lack of continuous dependence of their solution on the input data (magnetic field measurements), the problem of reconstruction of the magnetisation distribution 
additionally lacks uniqueness of the solution in view of
presence of invisible (or ``silent'') sources, i.e., magnetisations that do not produce magnetic field, see \cite{BarHarLimSafWei2013}.  However, as it was shown in \cite{BarChevHarLebLimMar2019} for planar (thin-plate) magnetisation distributions, and as follows from the general Helmholtz decomposition \cite{BarLebNem2023}, compactly supported invisible sources do not contribute to the net moment. This statement fixes the non-uniqueness issue and makes the problem of net moment recovery a feasible task.
In theory, this problem is even solvable in a closed form when measurements
are available on the entire plane above the sample. In reality, however,
the measurements are very limited and corrupted by the presence of
noise which dominates the signal in distant regions. Therefore, we
arrive at the problem of estimating the net moment of a sample from
a magnetic field component available on some portion of the plane
in proximity of the sample. We shall assume that this portion of the plane is sufficiently large, for otherwise the instability of the problem, due to the already mentioned inherent ill-posedness, will be even more severe. On the other hand, this assumption on the large size of the measurement area allows us to obtain explicit ready-to-use formulas.


We do not intend here to provide neither physical nor mathematical
description of the problem in any detailed fashion. Instead, we refer
the reader to the set of previous publications \cite{BarHarLimSafWei2013,LimWeiBarHarSaf2013,BarChevHarLebLimMar2019,BarGuiHarNorSaf2020,LimWei2016,FuLimVolTrub2020,WeiLimFonBau2007,LimWei2009}
and briefly introduce basic concepts that allow us to be more specific
in describing our main result and comparing it with relevant works.

We should mention that the present setting is somewhat classical for problems in magnetometry or gravimetry (i.e., detection of underground anomalous region from surveillance data), with a particularity here that only one component is available for direct measurements. The obtained results hence should be useful beyond the laboratory context of rock paleomagnetism, and the proposed methodology is clearly extendable to other particular set-ups.

We assume that the magnetic sample is described by a compactly supported
vector distribution
\[
\vec{\mathcal{M}}\left(\vec{x}\right)\equiv\left(\mathcal{M}_{1}\left(x_{1},x_{2},x_{3}\right),\mathcal{M}_{2}\left(x_{1},x_{2},x_{3}\right),\mathcal{M}_{3}\left(x_{1},x_{2},x_{3}\right)\right)^{T},\hspace{1em}\text{supp \ensuremath{\vec{\mathcal{M}}}}\subset Q,
\]
with some bounded set $Q\subset\mathbb{R}^{3}$. 

The relation between the unknown magnetisation distribution $\vec{\mathcal{M}}$ and the vertical component of the produced magnetic field is 
\begin{equation}
B_3\left(\vec{x}\right)=\dfrac{\partial}{\partial x_3}\iiint_{Q}\dfrac{1}{4\pi\left|\vec{x}-\vec{t}\right|}\nabla\cdot\vec{\mathcal{M}}\left(\vec{t}\right)\dd^{3}t,\hspace{1em}\vec{x}\in\mathbb{R}^{3}\backslash \overline{Q}.\label{eq:B3_divM}
\end{equation}
This latter quantity is experimentally measured on a portion of the horizontal plane at height $x_3=h$ for some constant $h>0$ such that this plane does not intersect $Q$, and it can be equivalently written as

{\small{}
\begin{eqnarray}
B_{3}\left(\mathbb{\mathbf{x}},h\right) & = & \iiint_{Q}\dfrac{3\left(h-t_{3}\right)\left[\mathcal{M}_{1}\left(\mathbf{t},t_{3}\right)\left(x_{1}-t_{1}\right)+\mathcal{M}_{2}\left(\mathbf{t},t_{3}\right)\left(x_{2}-t_{2}\right)\right]+\mathcal{M}_{3}\left(\mathbf{t},t_{3}\right)\left(2\left(h-t_{3}\right)^{2}-\left|\mathbf{x}-\mathbf{t}\right|^{2}\right)}{4\pi\left(\left|\mathbf{x}-\mathbf{t}\right|^{2}+\left(h-t_{3}\right)^{2}\right)^{5/2}}\dd^{3}t.\label{eq:B3_3D}
\end{eqnarray}
}Here and onwards, we employ bold symbols to denote $\mathbb{R}^{2}$ vectors,
e.g., $\mathbf{x}\equiv\left(x_{1},x_{2}\right)^{T}$ while using $\vec{x}\equiv\left(x_1,x_2,x_3\right)^{T}$ for $\mathbb{R}^3$ vectors. 
When $\vec{\mathcal{M}}$ is a distribution, the integral on the right-hand side of (\ref{eq:B3_3D}) should be understood as a sum of three terms, each is the duality pairing of a compactly supported scalar distribution $\mathcal{M}_{1}$, $\mathcal{M}_{2}$ or $\mathcal{M}_{3}$ with the corresponding smooth function on $\mathbb{R}^3$ (note that since the measurement plane does not intersect $Q$, the denominator is bounded away from zero and thus no singularities arise).

The geometry of the described setting is schematically shown in Figure \ref{fig:magn_geom}. This corresponds, up to a truncation of rectangular magnetic field map, to 
an experimental set-up of the Paleomagnetism lab at EAPS (Earth, Atmospheric and Planetary Sciences) department, MIT (Massachusetts Institute of Technology), involving a SQUID magnetometer, see \cite{WeiLimFonBau2007}. Moreover, with an extra preprocessing step of the field data, this also extends to the QDM magnetometer set-up used in the Paleomagnetism lab at Harvard University \cite{FuLimVolTrub2020}.

\begin{figure}
\includegraphics[scale=0.4]{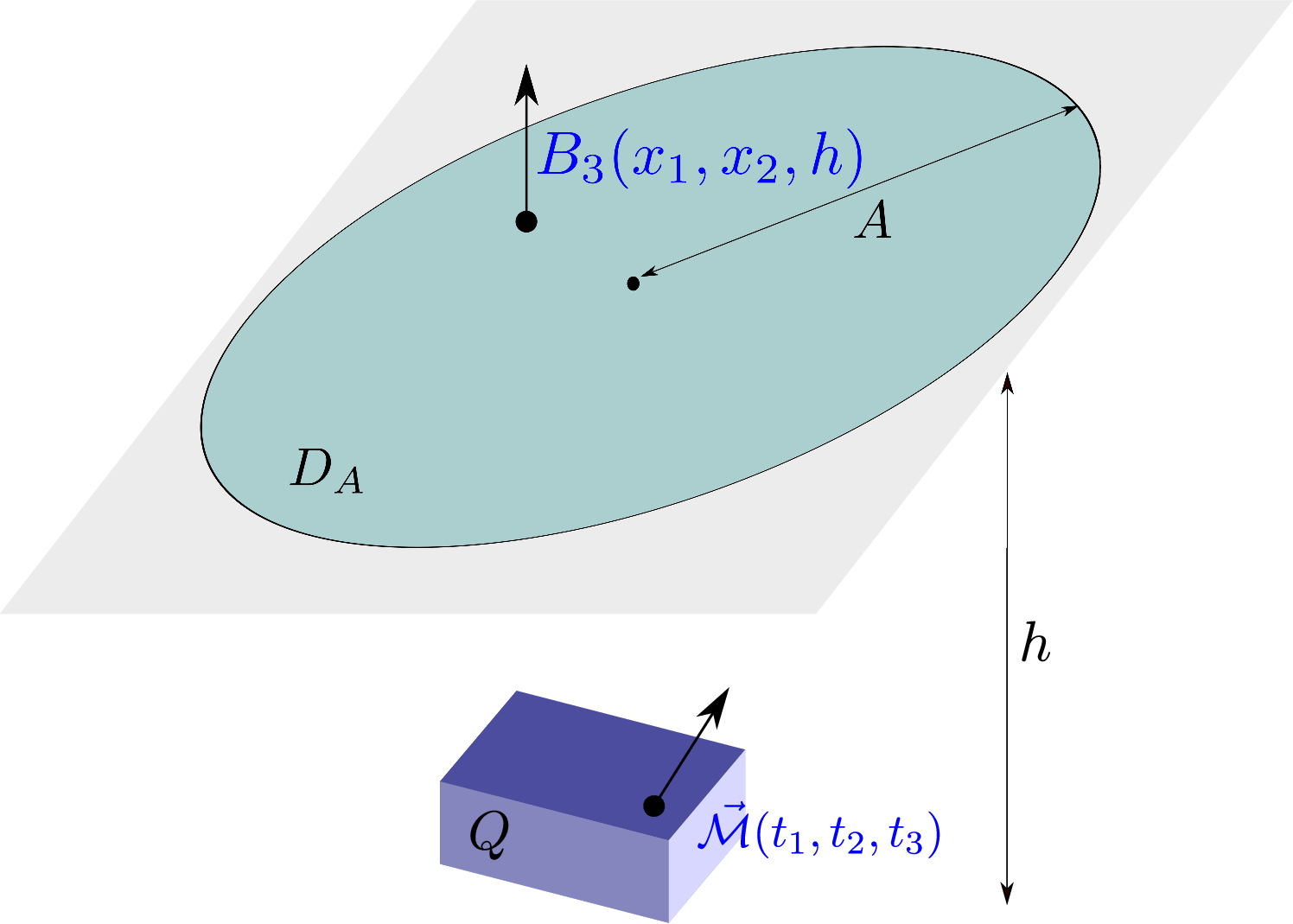}\caption{\label{fig:magn_geom} Schematic illustration of geometry of the problem arising from experimental set-up in the Paleomagnetism lab at EAPS, MIT (USA). $B_3$ is a vertical component of the magnetic field measured over the horizontal disk $D_A$ at some distance $h$ from a sample $Q$ with the magnetisation distribution $\vec{\mathcal{M}}$.}

\end{figure}

A quantity of basic physical interest is the net magnetisation moment:
\begin{equation}
\vec{m}\equiv\left(m_{1},m_{2},m_{2}\right)^{T}:=\iiint_{Q}\vec{\mathcal{M}}\left(\vec{x}\right)\dd^{3}x.\label{eq:m_def}
\end{equation}
This is a constant $\mathbb{R}^3$ vector equal to the total magnetisation of the sample which can be determined uniquely, unlike the complete magnetisation distribution $\vec{\mathcal{M}}$ which essentially enters the expression for  $B_3$ through $\nabla\cdot\vec{\mathcal{M}}$ and hence can be changed by any divergence-free vector field without altering values of $B_3$, see (\ref{eq:B3_divM}) and discussions in \cite{BarHarLimSafWei2013}. In what follows, we may interchangeably use slightly different names for $\vec{m}$: net moment, net magnetisation, zeroth-order algebraic
moment of a magnetisation distribution. When $\vec{\mathcal{M}}$ is not a regular function but a distribution (generalised function), the integral in (\ref{eq:m_def}) should be understood componentwise as the duality pairing of a compactly supported distribution with the constant function $1$ on $\mathbb{R}^3$: $m_j=\left\langle\mathcal{M}_j,1\right\rangle$,  $j\in\left\{1, 2, 3\right\}$. 

The present paper is dedicated to the explicit asymptotic estimation of the net magnetisation moment $\vec{m}$ from measurements of the magnetic field $B_3$ over a large planar area. In particular, as the measurement area, we have taken the horizontal disk $D_A$ of radius $A$ located at height $h$ above the magnetic sample. We formally derive and prove that each of the quantities $m_1$, $m_2$ and $m_3$ can be estimated by means of integration of the measured magnetic field data $B_3$ on $D_A$ against an appropriate function. We obtain several of such estimates depending on the asymptotic order (up to order 5 for $m_1$, $m_2$, and up to order 4 for $m_3$), though this process could be continued to obtain even higher order estimates. As it will be discussed, higher-order estimates do not make those of lower order redundant because of different level of sensitivity to noise which is unavoidable in any realistic setting.

It should be mentioned that even though recovery of the net moment
of a sample, our primary concern, is an important practical problem
on its own right, it can also serve, under appropriate assumptions,
as a preliminary step for the full magnetisation inversion (i.e., finding a magnetisation distribution $\mathcal{\vec{M}}$ that would be consistent with the measured data $B_3$ according to (\ref{eq:B3_3D})). Indeed, while it is unrealistic to retrieve three   generally independent magnetisation components (functions
or, more generally, distributions), $\mathcal{M}_{1}$, $\mathcal{M}_{2}$,
$\mathcal{M}_{3}$, from the partial knowledge of only one function $B_{3}$
(see (\ref{eq:B3_3D})) without additional assumptions, the problem simplifies significantly for
a class of samples which are unidirectionally magnetised (i.e., when magnetisation direction does not change throughout the sample but its magnitude does). Since quantifying the net moment of a sample implies a definite magnetisation direction, it is thus an essential element in this complete reconstruction procedure, see \cite{LimWeiBarHarSaf2013}. 

Note that the obtained estimates may be used either directly on the data when the original measurement area is sufficiently large compared to the localisation of the magnetisation in the sample and the lift-off distance (scanning height $h$), or after a preliminary application of the field extrapolation, see \cite{Pon2024}.

We stress that our analytical expressions for the asymptotic estimates are original, with only partial analogs in other works when translated into the current context.
According to \cite{Phil2005}, for extended finite magnetic sources, net moment estimates in terms of the measured field have first appeared in Helbig's work \cite{Hel1963} and later rediscovered by Clark and Schmidt, see, e.g., \cite{Schm-Cl1998} in context of magnetic prospection and detection of potential field anomalies, in particular. When truncated to a finite measurement area, these estimates correspond to our first-order asymptotic formulas for tangential components of net moment, but without a mention of the order of the asymptotic convergence, let alone a possibility of obtaining their more precise higher-order counterparts. Note that, in our context, since only $B_3$ component of the field is measured, those previous estimates do not provide means to estimate the normal net moment component $m_3$. Indeed, an estimate of $m_3$ was provided only in terms of either $B_1$ or $B_2$ components. When measurements of these two field components are not available, one can potentially argue that it is still possible to use the Helbig's formulas after reconstructing the missing components of the field from the knowledge of $B_3$ alone using, for example, \cite{Gun1975} and \cite[Ch. 12]{Bla1996}. However, this is possible only to accomplish in a stable way only if $B_3$ is known on the entire plane. In view of this, we believe that an explicit estimate of $m_3$ from knowledge of $B_3$ alone on a finite region for the first time was proposed only in \cite{Thesis, BarLebLimPon2017, BarChevLebHarLimPon2018}. It should be noted that results in \cite{BarChevLebHarLimPon2018} are also applicable for measurement regions different from disk, namely, for those of rectangular and diamond shapes.
We take advantage of this comparison of results to mention and to correct a couple of issues in \cite{BarLebLimPon2017}. First of all, there is a typo in the $m_3$ formula obvious when compared to (\ref{eq:fin_estim_m3_ord2}) or \cite[Thm 3.1.1]{Thesis}: in eq. \cite[Eq. (5)]{BarLebLimPon2017}, the factor $A$ should be in the numerator rather than the denominator. Moreover, the numerical illustration of the $m_3$ estimate, in case of noisy data, showed a divergent behaviour with the growth of the measurement disk, and this linearly growing trend was proposed to be removed by an appropriate postprocessing. It turns out that this undesirable behaviour was merely a consequence of the numerical implementation (Simpson's quadrature rule should have been avoided due to the low regularity of the noisy field) and not the asymptotic formula itself, as demonstrated in the present work; the same is also true for higher-order estimates of $m_3$. 
The drawback of finiteness of the measurement region for the application of Helbig's integral formulas has been recognised a while ago and its influence has been analysed, see, e.g., \cite{McKeFosHil2012}. One possibility to deal with this was to choose and correct a suitable integration window so that the consistency of Helbig's integrals is respected (since some integrals of the field data must be identically zero), see \cite{Phil2005}. Another approach consists in replacing the missing field data by the field of a fictitious dipole whose location and moment are estimated iteratively from appropriate integrals of the available data \cite{AndPed1979, CarPed2008}. A similar iterative strategy \cite{MedSil1995} is based on the multipolar expansion of order 2 (involving quadrupole moment tensor). In a recent work \cite{LimWei2023}, a multipolar approach of arbitrary order has been used to fit the measured data furnishing the net moment as first coefficients of the multipolar expansion. This shows good results but requires a potentially delicate parameter tuning (such as choice of the origin and the order of the expansion) as a part of the regularisation process.  
A totally different method which does not theoretically require a large measurement region is that based on a bounded extremal problem for finding best-possible linear estimators \cite{BarChevHarLebLimMar2019}. These estimators are auxiliary functions whose integrals against the measured data furnish the components of the net moment. Finding such estimators in a particular functional class requires solving an integro-differential equation which is a numerically laborous task. While formally applicable for any size and position of the measurement region with respect to the magnetic source, a practical sensitivity of that approach to perturbations of the measured data must increase enormously when the measurement area is not sufficiently large or well-located with respect to the source. Moreover, this approach is restricted to planar and regular (square-integrable) sources at a known height (depth). 

The present approach is applicable for one-component field measurements. It has high accuracy already for relatively small measurement area due to the higher-order asymptotic estimates (which, unlike Helbig's integrals, include also estimates for $m_3$ from $B_3$). The approach is easy to implement numerically due to the explicit (closed-form) nature of the estimates. It is applicable for bounded volumetric (and planar, in particular case) magnetic sources which are not necessarily regular (they may be even compactly supported distributions, not just square-integrable functions). Moreover, the approach does not require the knowledge of the depth of the source (as all the estimates are seen to be independent of the height parameter $h$  which must only satisfy a certain \textit{a priori} assumption). It is worth stressing that applicability of the approach to volumetric samples is not a purely academic generalisation of that for planar (thin-slab) magnetic distributions:  it is important for practical situations such as QDM microscopy (when the lift-off distance is too small, contributions from deeper sources are not negligible) but also in macroscopic contexts such as the already mentioned magnetic prospection problem.

The structure of the paper is as follows. Section \ref{sec:main} presents the main results of the paper formulated as Theorem \ref{thm:main} and discusses the limitations of their applicability. Section \ref{sec:proof}
has a twofold goal. First, it is meant to show how one idea based
on straightforward Fourier analysis can yield the simplest version of asymptotic
net moment estimates for both tangential and normal components. Second,
that section illustrates that, by means of a careful asymptotic analysis,
the explicit estimates can be not only proved rigorously but also
extended to higher orders. Hence, this material exactly constitutes
the proof of Theorem \ref{thm:main}. Then, in Section \ref{sec:numerics},
we illustrate the results numerically on the case where the magnetisation
has a singular support (a collection
of dipoles is modelled by magnetisation distribution that is a sum of Dirac delta functions) and deal with some practical aspects of the obtained estimates.
Finally, we conclude with Section \ref{sec:concl} summarising the
work, discussing the obtained results and outlining potential further research directions.

\section{Main results}\label{sec:main}
Let $\mathcal{E}^{\prime}\left(\mathbb{R}^{3}\right)$ be the space of compactly supported distributions on $\mathbb{R}^{3}$, i.e., the linear functionals on smooth functions $C^\infty\left(\mathbb{R}^3\right)$, see \cite[Sect. 6.1]{Strichartz}.

Denote $D_{A}:=\left\{ \mathbf{x}\in\mathbb{R}^{2}:\,\left|\mathbf{x}\right|<A\right\} $,
the disk of radius $A$ centreed at the origin $\mathbf{x}=\mathbf{0}$.

The main result of this work is summarised in the following theorem.
\begin{thm}
\label{thm:main} Assume that $\vec{\mathcal{M}}\in\left[\mathcal{E}^\prime\left(\mathbb{R}^3\right)\right]^3$ with $\text{supp }\vec{\mathcal{M}}\subset Q$ for a bounded set $Q\subset\mathbb{R}^{3}$. Suppose that the values of the vertical component of the magnetic field $B_3$ (related to $\vec{\mathcal{M}}$ by means of (\ref{eq:B3_3D})) are known on the horizontal disk $D_{A}\times\left\{x_3=h\right\}$ that does not intersect $Q$ and whose radius $A$ is
sufficiently large so that the following inequality holds:
\begin{equation}
\sup_{\vec{t}\in Q,\;\mathbf{x}\in\mathbb{R}^{2}\backslash D_{A}}\left|\frac{t_{1}^{2}+t_{2}^{2}+\left(h-t_{3}\right)^{2}}{x_{1}^{2}+x_{2}^{2}}-2\frac{x_{1}t_{1}+x_{2}t_{2}}{x_{1}^{2}+x_{2}^{2}}\right|<1.\label{eq:cond_asympt_thm}
\end{equation}
Then, the components of the net moment vector (\ref{eq:m_def}) can be asymptotically
estimated with different orders of accuracy as follows.\\
First-order estimates:
\begin{align}
m_{j} & =2\iint_{D_{A}}x_{j}B_{3}\left(\mathbf{x},h\right)\dd^{2}x+\mathcal{O}\left(\dfrac{1}{A}\right),\hspace{1em}\hspace{1em}j\in\left\{ 1,2\right\} .\label{eq:fin_estim_m1m2_ord1}
\end{align}
Second-order estimates:
\begin{align}
m_{j} & =2\iint_{D_{A}}\left(1+\dfrac{4x_{j}^{2}}{3A^{2}}\right)x_{j}B_{3}\left(\mathbf{x},h\right)\dd^{2}x+\mathcal{O}\left(\dfrac{1}{A^{2}}\right),\hspace{1em}\hspace{1em}j\in\left\{ 1,2\right\} ,\label{eq:fin_estim_m1m2_ord2}
\end{align}
\begin{equation}
m_{3}=2A\iint_{D_{A}}B_{3}\left(\mathbf{x},h\right)\dd^{2}x+\mathcal{O}\left(\dfrac{1}{A^{2}}\right).\label{eq:fin_estim_m3_ord2}
\end{equation}
Third-order estimates:
\begin{equation}
m_{j}=\frac{2}{5}\iint_{D_{A}}\left[5+24\left(\frac{x_{j}}{A}\right)^{4}\right]x_{j}B_{3}\left(\mathbf{x},h\right)\dd^{2}x+\mathcal{O}\left(\frac{1}{A^{3}}\right),\hspace{1em}\hspace{1em}j\in\left\{ 1,2\right\} ,\label{eq:fin_estim_m1m2_ord3}
\end{equation}
\begin{equation}
m_{3}=\frac{A}{4}\iint_{D_{A}}\left[5+40\left(\frac{x_{j}}{A}\right)^{4}-128\left(\frac{x_{j}}{A}\right)^{6}\right]B_{3}\left(\mathbf{x},h\right)\dd^{2}x+\mathcal{O}\left(\frac{1}{A^{3}}\right),\hspace{1em}\hspace{1em}j\in\left\{ 1,2\right\} .\label{eq:fin_estim_m3_ord3}
\end{equation}
Fourth-order estimates:
\begin{equation}
m_{j}=\frac{2}{105}\iint_{D_{A}}\left[105-2016\left(\frac{x_{j}}{A}\right)^{4}+19200\left(\frac{x_{j}}{A}\right)^{6}-22400\left(\frac{x_{j}}{A}\right)^{8}\right]x_{j}B_{3}\left(\mathbf{x},h\right)\dd^{2}x+\mathcal{O}\left(\frac{1}{A^{4}}\right),\hspace{1em}\hspace{1em}j\in\left\{ 1,2\right\} ,\label{eq:fin_estim_m1m2_ord4}
\end{equation}
\begin{equation}
m_{3}=\frac{A}{24}\iint_{D_{A}}\left[35+1792\left(\frac{x_{j}}{A}\right)^{6}-3200\left(\frac{x_{j}}{A}\right)^{8}\right]B_{3}\left(\mathbf{x},h\right)\dd^{2}x+\mathcal{O}\left(\frac{1}{A^{4}}\right),\hspace{1em}\hspace{1em}j\in\left\{ 1,2\right\} .\label{eq:fin_estim_m3_ord4}
\end{equation}
Fifth-order estimates:
\begin{equation}
m_{j}=\frac{2}{231}\iint_{D_{A}}\left[231-52800\left(\frac{x_{j}}{A}\right)^{6}+246400\left(\frac{x_{j}}{A}\right)^{8}-225792\left(\frac{x_{j}}{A}\right)^{10}\right]x_{j}B_{3}\left(\mathbf{x},h\right)\dd^{2}x+\mathcal{O}\left(\frac{1}{A^{5}}\right),\hspace{1em}\hspace{1em}j\in\left\{ 1,2\right\} .
\label{eq:fin_estim_m1m2_ord5}
\end{equation}
\end{thm}
\begin{rem}
\label{rem:about_units} Here and onwards (except numerics in Section
\ref{sec:numerics}), for the sake of simplicity, we have assumed
the system of physical units such that the constant of magnetic permeability
of vacuum $\mu_{0}$ is $1$. In general, the right-hand sides of
expression (\ref{eq:B3_3D}) should have the
factor $\mu_{0}$, and, in Si units, $\mu_{0}=4\pi\cdot10^{-7}$ N
/ A$^{2}$. Consequently, the right-hand sides of all the formulas
(\ref{eq:fin_estim_m1m2_ord1})--(\ref{eq:fin_estim_m1m2_ord5})
should, in principle, have the factor $1/\mu_{0}$.
\end{rem}
\begin{rem}\label{rem:assumpt}
It is not difficult to see that condition (\ref{eq:cond_asympt_thm}) can be replaced with
\begin{equation}\label{eq:cond_asympt_thm_alt}
\underset{{\vec{t}\in Q,\,\,\mathbf{x}\in\mathbb{R}^{2}\backslash D_{A}}}{\text{inf}}
\frac{\left(x_{1}-t_{1}\right)^{2}+\left(x_{2}-t_{2}\right)^{2}}{2\left(t_{1}^{2}+t_{2}^{2}\right)+\left(h-t_{3}\right)^{2}}>1.
\end{equation}
Indeed, (\ref{eq:cond_asympt_thm}) is obtained from (see (\ref{eq:cond_asympt}))
\[
-1<\frac{t_{1}^{2}+t_{2}^{2}+\left(h-t_{3}\right)^{2}}{x_{1}^{2}+x_{2}^{2}}-2\frac{x_{1}t_{1}+x_{2}t_{2}}{x_{1}^{2}+x_{2}^{2}}<1,\hspace{3em}\vec{t}\in Q,\,\,\mathbf{x}\in\mathbb{R}^{2}\backslash D_{A}.
\]
Here, the validity of the left inequality is trivial since 
\[
\left(x_{1}-t_{1}\right)^{2}+\left(x_{2}-t_{2}\right)^{2}+\left(h-t_{3}\right)^{2}>0, \hspace{3em}\vec{t}\in Q,\,\,\mathbf{x}\in\mathbb{R}^{2}\backslash D_{A},
\]
whereas the right inequality can be rewritten as
\[
x_{1}^{2}+x_{2}^{2}+2\left(x_{1}t_{1}+x_{2}t_{2}\right)+t_{1}^{2}+t_{2}^{2}>2\left(t_{1}^{2}+t_{2}^{2}\right)+\left(h-t_{3}\right)^{2},\hspace{3em}\vec{t}\in Q,\,\,\mathbf{x}\in\mathbb{R}^{2}\backslash D_{A}.
\]
By the symmetry of the area $\mathbb{R}^{2}\backslash D_{A}$, we can change the signs in front of $x_1$ and $x_2$. Recognising the complete square on the left-hand side and dividing over the positive expression from the right-hand side, we arrive exactly at the fraction appearing in (\ref{eq:cond_asympt_thm_alt}) and it only remains to take the infimum.

Finally, let us point out a simplification occurring in the common setting (e.g., when $Q$ is a minimal rectangular parallelepiped or a ball containing the magnetisation support located under the measurement device). Namely, if $Q$ is a bounded connected set whose lowest points lie at $x_3=0$ plane and whose horizontal projection $Q_{12}$ contains the centre $\mathbf{x}=\mathbf{0}$ of the measurement disk $D_A$, then inequality (\ref{eq:cond_asympt_thm_alt}) can be ensured by imposing a stricter but geometrically simpler condition
\begin{equation}\label{eq:cond_asympt_thm_alt_smpl}
\frac{2\,\left[\text{dist}\left(\partial D_{A},Q_{12}\right)\right]^{2}}{\left(\text{diam }Q_{12}\right)^{2}+2h^{2}}>1,
\end{equation}
where 
$\text{diam }Q_{12}$ 
denotes the diameter of $Q_{12}$ and $\text{dist}\left(\partial D_{A},Q_{12}\right)$ is the distance between the boundary $\partial D_{A}$ of the measurement disk and $Q_{12}$. Indeed, by the above assumption on the location of $Q$ and choice of the origin, we have $t_1^2+t_2^2\leq \left(\text{diam }Q_{12}\right)^2/4$, $\left(h-t_3\right)^2\leq h^{2}$ for $\vec{t}\in Q$, and hence the inequality
\begin{align*}
\underset{{\vec{t}\in Q,\,\,\mathbf{x}\in\mathbb{R}^{2}\backslash D_{A}}}{\text{inf}}
\frac{\left(x_{1}-t_{1}\right)^{2}+\left(x_{2}-t_{2}\right)^{2}}{2\left(t_{1}^{2}+t_{2}^{2}\right)+\left(h-t_{3}\right)^{2}}&\geq \frac{1}{\left(\text{diam }Q_{12}\right)^{2}/2+h^2}\,\,\underset{{\vec{t}\in Q,\,\,\mathbf{x}\in\mathbb{R}^{2}\backslash D_{A}}}{\text{inf}}
\left[\left(x_{1}-t_{1}\right)^{2}+\left(x_{2}-t_{2}\right)^{2}\right]\\
&=\frac{2\,\left[\text{dist}\left(\partial D_{A},Q_{12}\right)\right]^{2}}{\left(\text{diam }Q_{12}\right)^{2}+2h^{2}}
\end{align*}
means that (\ref{eq:cond_asympt_thm_alt_smpl}) would imply (\ref{eq:cond_asympt_thm_alt}).
\end{rem}

Note that condition (\ref{eq:cond_asympt_thm_alt_smpl}) specific to our asymptotic approach is actually quite natural and generally consistent with the geometric setting in which a good (reasonably stable) inversion can be expected since the Poisson transformation lying at the heart of the integral operator in (\ref{eq:B3_3D}) is known to rapidly spread away the source information, see also \cite{McKeFosHil2012}.

The derivation and the proof of the main estimates is lengthy and consists of a combination of simple techniques such as Fourier transforms, series expansions (at the origin and at infinity) and appropriate integral splitting. One of the main difficulties that some integrals involve both large and small parameters and hence obtaining uniform estimates by standard means of analysis of oscillatory and decaying integrals fails. Instead, we are able to benefit from circular geometry of the measurement area by identifying integral representations of known special (Bessel type) functions and rely on their differential and integral identities and asymptotics \cite[Sect. 10--11]{NIST} to conclude with the desired estimates.

\medskip

\section{Proof of Theorem \ref{thm:main}\label{sec:proof}}

\subsection{Some notation and preparatory transformations}

Before we proceed with deriving rigorously formulas (\ref{eq:fin_estim_m1m2_ord1})--(\ref{eq:fin_estim_m1m2_ord5})
and thus proving Theorem \ref{thm:main}, let us introduce some useful
notations. We shall use the following notational shortcut for the integral of
the magnetisation distribution against monomials 

\begin{equation}
\left\langle x_{1}^{j_{1}}x_{2}^{j_{2}}x_{3}^{j_{3}}M_{n}\right\rangle :=\iiint_{Q}x_{1}^{j_{1}}x_{2}^{j_{2}}x_{3}^{j_{3}}\mathcal{M}_{n}\left(\vec{x}\right)\dd^{3}x,\hspace{1em}\hspace{1em}n\in\left\{ 1,2,3\right\} ,\hspace{1em}j_{1},j_{2},j_{3}\in\mathbb{N}_{0},\label{eq:alg_moments_3D}
\end{equation}
where we denoted $\mathbb{N}_{0}:=\left\{ 0,1,2,\ldots\right\} $,
the set of natural numbers with zero. Following this convention,
for the sake of brevity, we may also write, for example,
\[
\left\langle \left(h-x_{3}\right)M_{n}\right\rangle :=hm_{n}-\left\langle x_{3}M_{n}\right\rangle =\iiint_{Q}\left(h-x_{3}\right)\mathcal{M}_{n}\left(\vec{x}\right)\dd^{3}x,\hspace{1em}\hspace{1em}n\in\left\{ 1,2,3\right\}.
\]
As before, we note that, when magnetisation $\vec{\mathcal{M}}$  is not a function but a compactly supported distribution,
the integrals above should be understood as the duality pairings between $\mathcal{M}_n\in\mathcal{E}^{\prime}\left(\mathbb{R}^3\right)$ and $x_{1}^{j_{1}}x_{2}^{j_{2}}x_{3}^{j_{3}}\in C^{\infty}\left(\mathbb{R}^3\right)$:
$\iiint_{Q}x_{1}^{j_{1}}x_{2}^{j_{2}}x_{3}^{j_{3}}\mathcal{M}_{n}\left(\vec{x}\right)\dd^{3}x\equiv \left\langle \mathcal{M}_{n}, x_{1}^{j_{1}}x_{2}^{j_{2}}x_{3}^{j_{3}}\right\rangle$ with $j_{1}$, $j_{2}$, $j_{3}\in\mathbb{N}_{0}$, $n\in\left\{ 1,2,3\right\}$. 

Let us $\widehat{\,\cdot\,}$ denote the two-dimensional Fourier transform
which, by our convention, is defined as 
\[
\hat{f}\left(\mathbf{k}\right)\equiv\mathcal{F}\left[f\right]\left(\mathbf{k}\right):=\iint_{\mathbb{R}^{2}}e^{2\pi i\mathbf{k}\cdot\mathbf{x}}f\left(\mathbf{x}\right)\dd^{2}x,
\]
where $i=\sqrt{-1}$ stands for the imaginary unit, and $\mathbf{k}\cdot\mathbf{x}=k_{1}x_{1}+k_{2}x_{2}$
is the Euclidean inner product. With this definition, the differentiation
and convolution properties of Fourier transform have the form
\begin{equation}
\mathcal{F}\left[\partial_{x_{j}}f\right]\left(\mathbf{k}\right)=-2\pi ik_{j}\hat{f}\left(\mathbf{k}\right),\hspace{1em}\mathcal{F}\left[x_{j}f\right]\left(\mathbf{k}\right)=\frac{1}{2\pi i}\partial_{k_{j}}\hat{f}\left(\mathbf{k}\right),\hspace{1em}\hspace{1em}j\in\left\{ 1,2\right\} ,\label{eq:FT_der}
\end{equation}
\begin{equation}
\mathcal{F}\left[f\star g\right]\left(\mathbf{k}\right):=\iint_{\mathbb{R}^{2}}e^{2\pi i\mathbf{k}\cdot\mathbf{x}}\iint_{\mathbb{R}^{2}}f\left(\mathbf{x}-\mathbf{t}\right)g\left(\mathbf{t}\right)\dd^2 t\, \dd^2 x=\hat{f}\left(\mathbf{k}\right)\hat{g}\left(\mathbf{k}\right).\label{eq:FT_conv}
\end{equation}
We also note that the Fourier transform of the two-dimensional Poisson
kernel is well-known (see, e.g., \cite[Sect. 4.2]{Strichartz}), that
is, for any $H>0$, we have
\begin{equation}
\mathcal{F}\left[\frac{H}{2\pi\left(\left|\mathbf{x}\right|^{2}+H^{2}\right)^{3/2}}\right]\left(\mathbf{k}\right)=e^{-2\pi H\left|\mathbf{k}\right|},\hspace{1em}\mathbf{k}\in\mathbb{R}^{2}.\label{eq:FT_Poisson}
\end{equation}

Let us now rewrite (\ref{eq:B3_3D}) as
\begin{align}
B_{3}\left(\mathbf{x},h\right)= & -\frac{1}{4\pi}\iiint_{Q}\left[\left(\mathcal{M}_{1}\left(\mathbf{t},t_{3}\right)\dfrac{\partial}{\partial x_{1}}+\mathcal{M}_{2}\left(\mathbf{t},t_{3}\right)\dfrac{\partial}{\partial x_{2}}\right)\frac{h-t_{3}}{\left(\left|\mathbf{x}-\mathbf{t}\right|^{2}+\left(h-t_{3}\right)^{2}\right)^{3/2}}\right.\label{eq:B3_alt}\\
 & \left.+\mathcal{M}_{3}\left(\mathbf{t},t_{3}\right)\left(\frac{\partial}{\partial x_{3}}\frac{x_{3}-t_{3}}{\left(\left|\mathbf{x}-\mathbf{t}\right|^{2}+\left(x_{3}-t_{3}\right)^{2}\right)^{3/2}}\right)\Biggr|_{x_{3}=h}\right]\dd^{3}t.\nonumber 
\end{align}
Taking Fourier transform of (\ref{eq:B3_alt}) in both $x_{1}$ and $x_{2}$
variables, we use (\ref{eq:FT_der}), (\ref{eq:FT_conv}) and employ
(\ref{eq:FT_Poisson}) twice: with $H:=h-t_{3}$ in the first line
of (\ref{eq:B3_alt}), and with $H:=x_{3}-t_{3}$ in the second one.
We thus arrive at
\begin{equation}
\hat{B}_{3}\left(\mathbf{k},h\right)=\pi\int_{Q_{3}}e^{-2\pi\left(h-t_{3}\right)\left|\mathbf{k}\right|}\left(ik_{1}\hat{\mathcal{M}}_{1}\left(\mathbf{k},t_{3}\right)+ik_{2}\hat{\mathcal{M}}_{2}\left(\mathbf{k},t_{3}\right)+\left|\mathbf{k}\right|\hat{\mathcal{M}}_{3}\left(\mathbf{k},t_{3}\right)\right)\dd t_{3},\label{eq:B3_hat}
\end{equation}
where $Q_{3}$ denotes the vertical projection of the set $Q$.

We note that even though (\ref{eq:fin_estim_m1m2_ord1})--(\ref{eq:fin_estim_m1m2_ord2}),
(\ref{eq:fin_estim_m1m2_ord3}), (\ref{eq:fin_estim_m1m2_ord4}),
(\ref{eq:fin_estim_m1m2_ord5}) give estimates for tangential net
moment components $m_{1}$, $m_{2}$, we shall restrict ourselves
to dealing only with $m_{1}$. The situation with $m_{2}$ is completely
analogous.

First, we are going to illustrate our strategy of the derivation of
asymptotic estimates of the net magnetisation moment. Here, we shall
be only concerned with the low-order formulas of Theorem \ref{thm:main}
and we shall omit a rigorous justification step. Then, we shall proceed
with formal justification and extension of the result to higher orders.

\subsection{Illustration of the basic idea of the derivation of the net moment
estimates\label{subsec:proof_illustr}}

We are going to focus on deriving (\ref{eq:fin_estim_m1m2_ord1}).

We take $k_{2}=0$ in expression (\ref{eq:B3_hat}) to obtain
\begin{equation}
\hat{B}_{3}\left(k_{1},0,h\right)=\pi\int_{Q_{3}}e^{-2\pi\left(h-t_{3}\right)\left|k_{1}\right|}\left(ik_{1}\hat{\mathcal{M}}_{1}\left(k_{1},0,t_{3}\right)+\left|k_{1}\right|\hat{\mathcal{M}}_{3}\left(k_{1},0,t_{3}\right)\right)\dd t_{3}.\label{eq:B3_hat_k1}
\end{equation}

Since the magnetisation distribution $\vec{\mathcal{M}}$ is compactly
supported, the Fourier transforms $\hat{\mathcal{M}}_{j}\left(\mathbf{k},t_{3}\right)$,
$j\in\left\{ 1,2,3\right\} $, are entire functions in $k_{1}$, $k_{2}\in\mathbb{C}$
for each $t_{3}\in Q_{3}$, according to the Paley-Wiener theory (see,
e.g., \cite[Thm 4.1]{SteWei2016}). In particular, power-series expansion
of $\hat{\mathcal{M}}_{j}\left(k_{1},0,t_{3}\right)$, $j\in\left\{ 1,2,3\right\} $,
about the origin $k_{1}=0$ of the complex plane $\left(\text{Re }k_{1},\,\text{Im }k_{1}\right)$
gives
\begin{align*}
\hat{\mathcal{M}}_{j}\left(k_{1},0,t_{3}\right)= & \hat{\mathcal{M}}_{j}\left(\mathbf{0},t_{3}\right)+\partial_{k_{1}}\hat{\mathcal{M}}_{j}\left(\mathbf{0},t_{3}\right)k_{1}+\dfrac{1}{2}\partial_{k_{1}}^{2}\hat{\mathcal{M}}_{j}\left(\mathbf{0},t_{3}\right)k_{1}^{2}+\mathcal{O}\left(\left|k_{1}\right|^{3}\right),\hspace{1em}j\in\left\{ 1,2,3\right\} .
\end{align*}
Combining this expansion with the straightforward identities
\[
m_{j}=\int_{Q_{3}}\hat{\mathcal{M}}_{j}\left(\mathbf{0},t_{3}\right)dt_{3},\hspace{1em}j\in\left\{ 1,2,3\right\} ,
\]
\[
2\pi i\left\langle M_{j}x_{1}\right\rangle =\int_{Q_{3}}\partial_{k_{1}}\hat{\mathcal{M}}_{j}\left(\mathbf{0},t_{3}\right)dt_{3},\hspace{1em}j\in\left\{ 1,2,3\right\} ,
\]
and the Taylor expansion in $\left|k_{1}\right|$ of the exponential
factor in (\ref{eq:B3_hat_k1})
\[
e^{-2\pi\left(h-t_{3}\right)\left|k_{1}\right|}=1-2\pi\left(h-t_{3}\right)\left|k_{1}\right|+2\pi^{2}\left(h-t_{3}\right)^{2}\left|k_{1}\right|^{2}+\mathcal{O}\left(\left|k_{1}\right|^{3}\right),
\]
we obtain
\[
\hat{B}_{3}\left(k_{1},0,h\right)=\text{Re }\hat{B}_{3}\left(k_{1},0,h\right)+i\,\text{Im }\hat{B}_{3}\left(k_{1},0,h\right),\hspace{1em}k_{1}\in\mathbb{R},
\]
\begin{align}
\text{Re }\hat{B}_{3}\left(k_{1},0,h\right)= & \pi m_{3}\left|k_{1}\right|-2\pi^{2}\left(\left\langle x_{1}M_{1}\right\rangle +\left\langle \left(h-x_{3}\right)M_{3}\right\rangle \right)\left|k_{1}\right|^{2}\nonumber \\
 & +2\pi^{3}\left(2\left\langle \left(h-x_{3}\right)x_{1}M_{1}\right\rangle +\left\langle \left(h-x_{3}\right)^{2}M_{3}\right\rangle -\left\langle x_{1}^{2}M_{3}\right\rangle \right)\left|k_{1}\right|^{3}+\mathcal{O}\left(\left|k_{1}\right|^{4}\right),\label{eq:Re_B3_simpl}
\end{align}
\begin{align}
\text{Im }\hat{B}_{3}\left(k_{1},0,h\right)= & \pi m_{1}k_{1}-2\pi^{2}\left(\left\langle \left(h-x_{3}\right)M_{1}\right\rangle -\left\langle x_{1}M_{3}\right\rangle \right)k_{1}\left|k_{1}\right|\nonumber \\
 & -2\pi^{3}\left(\left\langle x_{1}^{2}M_{1}\right\rangle -\left\langle \left(h-x_{3}\right)^{2}M_{1}\right\rangle +2\left\langle \left(h-x_{3}\right)x_{1}M_{3}\right\rangle \right)k_{1}^{3}+\mathcal{O}\left(\left|k_{1}\right|^{4}\right).\label{eq:Im_B3_simpl}
\end{align}
Here, the remainder terms are uniformly small for all $t_{3}\in Q_{3}$
due to the boundedness of the set $Q_{3}$.

On the other hand, we can write
\begin{align}
\hat{B}_{3}\left(k_{1},0,h\right) & =\iint_{D_{A}}e^{2\pi ik_{1}x_{1}}B_{3}\left(\mathbf{x},h\right)\dd^{2}x+\iint_{\mathbb{R}^{2}\backslash D_{A}}e^{2\pi ik_{1}x_{1}}B_{3}\left(\mathbf{x},h\right)\dd^{2}x.\label{eq:B3_decomp}
\end{align}

We note that in the first term on the right-hand side of (\ref{eq:B3_decomp}),
the integration range is finite and hence an expansion in powers of
$k_{1}$ simply follows from that of the exponential factor:
\begin{align}
\iint_{D_{A}}e^{2\pi ik_{1}x_{1}}B_{3}\left(\mathbf{x},h\right)\dd^{2}x= & \iint_{D_{A}}B_{3}\left(\mathbf{x},h\right)\dd^{2}x+2\pi ik_{1}\iint_{D_{A}}x_{1}B_{3}\left(\mathbf{x},h\right)\dd^{2}x-2\pi^{2}k_{1}^{2}\iint_{D_{A}}x_{1}^{2}B_{3}\left(\mathbf{x},h\right)\dd^{2}x\label{eq:B3_DA_exp_simpl}\\
 & -\frac{4\pi^{3}}{3}ik_{1}^{3}\iint_{D_{A}}x_{1}^{3}B_{3}\left(\mathbf{x},h\right)\dd^{2}x+\mathcal{O}\left(\left|k_{1}\right|^{4}\right),\nonumber 
\end{align}
and hence
\begin{equation}
\iint_{D_{A}}\cos\left(2\pi k_{1}x_{1}\right)B_{3}\left(\mathbf{x},h\right)\dd^{2}x=\iint_{D_{A}}B_{3}\left(\mathbf{x},h\right)\dd^{2}x-2\pi^{2}k_{1}^{2}\iint_{D_{A}}x_{1}^{2}B_{3}\left(\mathbf{x},h\right)\dd^{2}x+\mathcal{O}\left(\left|k_{1}\right|^{4}\right),\label{eq:B3_DA_exp_cos_simpl}
\end{equation}
\begin{equation}
\iint_{D_{A}}\sin\left(2\pi k_{1}x_{1}\right)B_{3}\left(\mathbf{x},h\right)\dd^{2}x=2\pi k_{1}\iint_{D_{A}}x_{1}B_{3}\left(\mathbf{x},h\right)\dd^{2}x-\frac{4\pi^{3}}{3}k_{1}^{3}\iint_{D_{A}}x_{1}^{3}B_{3}\left(\mathbf{x},h\right)\dd^{2}x+\mathcal{O}\left(\left|k_{1}\right|^{5}\right).\label{eq:B3_DA_exp_sin_simpl}
\end{equation}

Producing an expansion in powers of $k_{1}$ of the second term in
(\ref{eq:B3_decomp}) is much less straightforward and requires a
preliminary simplification. More precisely, we expand the field $B_{3}\left(\mathbf{x},h\right)$
for large $\left|\mathbf{x}\right|$ and, assuming largeness of the
region $D_{A}$, retain only first few terms of this expansion. Namely,
from (\ref{eq:B3_3D}), we have 
\begin{equation}
B_{3}\left(\mathbf{x},h\right)=-\frac{m_{3}}{4\pi\left|\mathbf{x}\right|^{3}}+\frac{3}{4\pi}\frac{\left(\left\langle \left(h-x_{3}\right)M_{1}\right\rangle -\left\langle x_{1}M_{3}\right\rangle \right)x_{1}+\left(\left\langle \left(h-x_{3}\right)M_{2}\right\rangle -\left\langle x_{2}M_{3}\right\rangle \right)x_{2}}{\left|\mathbf{x}\right|^{5}}+\mathcal{O}\left(\frac{1}{\left|\mathbf{x}\right|^{5}}\right).\label{eq:B3_asympt_simpl}
\end{equation}
Consequently, passing to the polar coordinates using $x_{1}=r\cos\theta$,
$x_{2}=r\sin\theta$, $\dd^{2}x=r\dd r\dd\theta$, we can write 
\begin{align*}
\iint_{\mathbb{R}^{2}\backslash D_{A}}e^{2\pi ik_{1}x_{1}}B_{3}\left(\mathbf{x},h\right)\dd^{2}x= & -\frac{m_{3}}{4\pi}\int_{A}^{\infty}\int_{0}^{2\pi}e^{2\pi ik_{1}r\cos\theta}\dd\theta\frac{\dd r}{r^{2}}\\
 & +\frac{3}{4\pi}\left(\left\langle \left(h-x_{3}\right)M_{1}\right\rangle -\left\langle x_{1}M_{3}\right\rangle \right)\int_{A}^{\infty}\int_{0}^{2\pi}e^{2\pi ik_{1}r\cos\theta}\cos\theta \dd\theta\frac{\dd r}{r^{3}}\\
 & +\frac{3}{4\pi}\left(\left\langle \left(h-x_{3}\right)M_{2}\right\rangle -\left\langle x_{2}M_{3}\right\rangle \right)\int_{A}^{\infty}\int_{0}^{2\pi}e^{2\pi ik_{1}r\cos\theta}\sin\theta \dd\theta\frac{\dd r}{r^{3}}\\
 & +\mathcal{R}_{A,k_{1}},
\end{align*}
where the residue term $\mathcal{R}_{A,k_{1}}$ is expected to be
$\mathcal{O}\left(1/A^{3}\right)$ for sufficiently small values of
$\left|k_{1}\right|$.

Furthermore, taking real and imaginary parts of both sides, we obtain,
respectively, 
\begin{equation}
\iint_{\mathbb{R}^{2}\backslash D_{A}}\cos\left(2\pi k_{1}x_{1}\right)B_{3}\left(\mathbf{x},h\right)\dd^{2}x=-\frac{m_{3}}{4\pi}\int_{A}^{\infty}\int_{0}^{2\pi}\cos\left(2\pi k_{1}r\cos\theta\right)\dd\theta\frac{\dd r}{r^{2}}+\text{Re }\mathcal{R}_{A,k_{1}},\label{eq:B3_cos_ext_exp_simpl}
\end{equation}
\begin{equation}
\iint_{\mathbb{R}^{2}\backslash D_{A}}\sin\left(2\pi k_{1}x_{1}\right)B_{3}\left(\mathbf{x},h\right)\dd^{2}x=\frac{3}{4\pi}\left(\left\langle \left(h-x_{3}\right)M_{1}\right\rangle -\left\langle x_{1}M_{3}\right\rangle \right)\int_{A}^{\infty}\int_{0}^{2\pi}\sin\left(2\pi k_{1}r\cos\theta\right)\cos\theta \dd\theta\frac{\dd r}{r^{3}}+\text{Im }\mathcal{R}_{A,k_{1}},\label{eq:B3_sin_ext_exp_simpl}
\end{equation}
where we took into account that 
\[
\int_{0}^{2\pi}\cos\left(2\pi k_{1}r\cos\theta\right)\sin\theta \dd\theta=\int_{0}^{2\pi}\cos\left(2\pi k_{1}r\cos\theta\right)\cos\theta \dd\theta=0,
\]
\[
\int_{0}^{2\pi}\sin\left(2\pi k_{1}r\cos\theta\right)\dd\theta=\int_{0}^{2\pi}\sin\left(2\pi k_{1}r\cos\theta\right)\sin\theta \dd\theta=0,
\]
according to the results of Lemma \ref{lem:trig_ints}.

Next, as it turns out (see Subsection \ref{subsec:proof_rigor} for
more details), the integrals on the right-hand sides of (\ref{eq:B3_cos_ext_exp_simpl})--(\ref{eq:B3_sin_ext_exp_simpl})
can be evaluated explicitly in terms of some cylindrical functions.
Known representations of these special functions lead to the desired
asymptotic expansions in powers of $k_{1}$. In particular, for small
$\left|k_{1}\right|$, we can deduce that
\begin{align}
\int_{A}^{\infty}\int_{0}^{2\pi}\cos\left(2\pi k_{1}r\cos\theta\right)\dd\theta\frac{\dd r}{r^{2}} & =2\pi\left|k_{1}\right|\int_{2\pi\left|k_{1}\right|A}^{\infty}\int_{0}^{2\pi}\cos\left(r\cos\theta\right)\dd\theta\frac{\dd r}{r^{2}}\label{eq:cos_int_exp_simpl}\\
 & =\frac{2\pi}{A}-4\pi^{2}\left|k_{1}\right|+2\pi^{3}A\left|k_{1}\right|^{2}+\mathcal{\mathcal{O}}\left(A^{2}\left|k_{1}\right|^{3}\right),\nonumber 
\end{align}
\begin{align}
\int_{A}^{\infty}\int_{0}^{2\pi}\sin\left(2\pi k_{1}r\cos\theta\right)\cos\theta \dd\theta\frac{\dd r}{r^{3}} & =4\pi^{2}k_{1}\left|k_{1}\right|\int_{2\pi\left|k_{1}\right|A}^{\infty}\int_{0}^{2\pi}\sin\left(r\cos\theta\right)\cos\theta \dd\theta\frac{\dd r}{r^{3}}\label{eq:sin_int_exp_simpl}\\
 & =\frac{2\pi^{2}}{A}k_{1}-\frac{8\pi^{3}}{3}k_{1}\left|k_{1}\right|+\pi^{4}Ak_{1}^{3}+\mathcal{O}\left(A^{3}\left|k_{1}\right|^{5}\right),\nonumber 
\end{align}
where the notation $\mathcal{\mathcal{O}}\left(A^{2}\left|k_{1}\right|^{3}\right)$
also hides the terms of powers of $\left|k_{1}\right|$ higher than
$3$ regardless of the presence of the $A$ factors such as $\mathcal{\mathcal{O}}\left(A^{3}\left|k_{1}\right|^{4}\right)$.

Therefore, by taking the real part of (\ref{eq:B3_decomp}) and using
(\ref{eq:B3_DA_exp_cos_simpl}), (\ref{eq:B3_cos_ext_exp_simpl})
and (\ref{eq:cos_int_exp_simpl}), we obtain
\begin{align*}
\text{Re }\hat{B}_{3}\left(k_{1},0,h\right)= & \iint_{D_{A}}B_{3}\left(\mathbf{x},h\right)\dd^{2}x-2\pi^{2}k_{1}^{2}\iint_{D_{A}}x_{1}^{2}B_{3}\left(\mathbf{x},h\right)\dd^{2}x\\
 & -\frac{m_{3}}{4\pi}\left(\frac{2\pi}{A}-4\pi^{2}\left|k_{1}\right|+2\pi^{3}A\left|k_{1}\right|^{2}\right)+\text{Re }\mathcal{R}_{A,k_{1}}+\mathcal{O}\left(A^{2}\left|k_{1}\right|^{3}\right).
\end{align*}
Comparing this with (\ref{eq:Re_B3_simpl}) and, in particular, evaluating
both expressions at $k_{1}=0$, we arrive at the following identity:
\begin{equation}
\iint_{D_{A}}B_{3}\left(\mathbf{x},h\right)\dd^{2}x-\frac{m_{3}}{2A}+\mathcal{O}\left(\frac{1}{A^{3}}\right)=0,\label{eq:m3_estim_simpl}
\end{equation}
where we took into account that $\left.\text{Re }\mathcal{R}_{A,k_{1}}\right|_{k_{1}=0}=\mathcal{O}\left(1/A^{3}\right)$.

Similarly, taking the imaginary part of (\ref{eq:B3_decomp}), we
combine (\ref{eq:B3_DA_exp_sin_simpl}), (\ref{eq:B3_sin_ext_exp_simpl})
and (\ref{eq:sin_int_exp_simpl}) to deduce that
\begin{align*}
\text{Im }\hat{B}_{3}\left(k_{1},0,h\right)= & 2\pi k_{1}\iint_{D_{A}}x_{1}B_{3}\left(\mathbf{x},h\right)\dd^{2}x-\frac{4\pi^{3}}{3}k_{1}^{3}\iint_{D_{A}}x_{1}^{3}B_{3}\left(\mathbf{x},h\right)\dd^{2}x\\
 & +\frac{3}{4\pi}\left(\left\langle \left(h-x_{3}\right)M_{1}\right\rangle -\left\langle x_{1}M_{3}\right\rangle \right)\left(\frac{2\pi^{2}}{A}k_{1}-\frac{8\pi^{3}}{3}k_{1}\left|k_{1}\right|+\pi^{4}Ak_{1}^{3}\right)+\text{Im }\mathcal{R}_{A,k_{1}}+\mathcal{O}\left(A^{3}\left|k_{1}\right|^{5}\right).
\end{align*}
Comparison of this expression with (\ref{eq:Im_B3_simpl}) and matching
the coefficients of the $k_{1}$ terms yields
\begin{equation}
2\pi\iint_{D_{A}}x_{1}B_{3}\left(\mathbf{x},h\right)\dd^{2}x+\frac{3\pi}{2A}\left(\left\langle \left(h-x_{3}\right)M_{1}\right\rangle -\left\langle x_{1}M_{3}\right\rangle \right)+\mathcal{O}\left(\frac{1}{A^{3}}\right)=\pi m_{1},\label{eq:m1_estim_simpl}
\end{equation}
where we assumed that, for sufficiently small $\left|k_{1}\right|$,
we have $\text{Im }\mathcal{R}_{A,k_{1}}=\mathcal{O}\left(1/A^{3}\right)$.

While (\ref{eq:m3_estim_simpl}), (\ref{eq:m1_estim_simpl}) imply
estimates (\ref{eq:fin_estim_m3_ord2}), (\ref{eq:fin_estim_m1m2_ord1}),
respectively, the derivation given above was not rigorous and required
additional assumptions on the residue term $\mathcal{R}_{A,k_{1}}$
which was reasonably deemed to be sufficiently small for large $A$
but was not estimated uniformly in $k_{1}$. We shall now proceed
with rigorous analysis which will also make it possible to derive
higher-order analogs of estimates (\ref{eq:fin_estim_m1m2_ord1}),
(\ref{eq:fin_estim_m3_ord2}).

\subsection{Rigorous analysis and higher-order asymptotic estimates\label{subsec:proof_rigor}}

Let us start by improving estimate (\ref{eq:B3_asympt_simpl}). To
this effect, we use the following elementary Taylor expansions, convergent
for $\left|z\right|<1$, 
\[
\frac{1}{\left(1+z\right)^{3/2}}=1-\frac{3}{2}z+\frac{15}{8}z^{2}-\frac{35}{16}z^{3}+\mathcal{O}\left(z^{4}\right),
\]
\[
\frac{1}{\left(1+z\right)^{5/2}}=1-\frac{5}{2}z+\frac{35}{8}z^{2}+\mathcal{O}\left(z^{3}\right),
\]
to obtain, for $t_{1}$, $t_{2}$, $t_{3}$, $h\in\mathbb{R}$, $t_3 \neq h$,
\begin{align}
\frac{1}{\left[\left(x_{1}-t_{1}\right)^{2}+\left(x_{2}-t_{2}\right)^{2}+\left(h-t_{3}\right)^{2}\right]^{3/2}}= & \frac{1}{\left(x_{1}^{2}+x_{2}^{2}\right)^{3/2}}\left[1-2\frac{x_{1}t_{1}+x_{2}t_{2}}{x_{1}^{2}+x_{2}^{2}}+\frac{t_{1}^{2}+t_{2}^{2}+\left(h-t_{3}\right)^{2}}{x_{1}^{2}+x_{2}^{2}}\right]^{-3/2}\label{eq:expans_alg32}\\
= & \frac{1}{\left(x_{1}^{2}+x_{2}^{2}\right)^{3/2}}\left[1+3\frac{x_{1}t_{1}+x_{2}t_{2}}{x_{1}^{2}+x_{2}^{2}}-\frac{3}{2}\frac{t_{1}^{2}+t_{2}^{2}+\left(h-t_{3}\right)^{2}}{x_{1}^{2}+x_{2}^{2}}+\frac{15}{2}\frac{\left(x_{1}t_{1}+x_{2}t_{2}\right)^{2}}{\left(x_{1}^{2}+x_{2}^{2}\right)^{2}}\right.\nonumber \\
 & \left.-\frac{15}{2}\frac{\left(x_{1}t_{1}+x_{2}t_{2}\right)\left(t_{1}^{2}+t_{2}^{2}+\left(h-t_{3}\right)^{2}\right)}{\left(x_{1}^{2}+x_{2}^{2}\right)^{2}}+\frac{35}{2}\frac{\left(x_{1}t_{1}+x_{2}t_{2}\right)^{3}}{\left(x_{1}^{2}+x_{2}^{2}\right)^{3}}\right]+\mathcal{O}\left(\frac{1}{\left|\mathbf{x}\right|^{7}}\right),\nonumber 
\end{align}
\begin{align}
\frac{1}{\left[\left(x_{1}-t_{1}\right)^{2}+\left(x_{2}-t_{2}\right)^{2}+\left(h-t_{3}\right)^{2}\right]^{5/2}}= & \frac{1}{\left(x_{1}^{2}+x_{2}^{2}\right)^{5/2}}\left[1+5\frac{x_{1}t_{1}+x_{2}t_{2}}{x_{1}^{2}+x_{2}^{2}}-\frac{5}{2}\frac{t_{1}^{2}+t_{2}^{2}+\left(h-t_{3}\right)^{2}}{x_{1}^{2}+x_{2}^{2}}\right.\label{eq:expans_alg52}\\
 & \left.+\frac{35}{2}\frac{\left(x_{1}t_{1}+x_{2}t_{2}\right)^{2}}{\left(x_{1}^{2}+x_{2}^{2}\right)^{2}}\right]+\mathcal{O}\left(\frac{1}{\left|\mathbf{x}\right|^{8}}\right),\nonumber 
\end{align}
where $x_{1}$, $x_{2}\in\mathbb{R}$, $\left|x_1\right|+\left|x_2\right|>0$, are of a sufficiently large magnitude so that
\begin{equation}
\left|\frac{t_{1}^{2}+t_{2}^{2}+\left(h-t_{3}\right)^{2}}{x_{1}^{2}+x_{2}^{2}}-2\frac{x_{1}t_{1}+x_{2}t_{2}}{x_{1}^{2}+x_{2}^{2}}\right|<1.\label{eq:cond_asympt}
\end{equation}
Expansions (\ref{eq:expans_alg32})--(\ref{eq:expans_alg52}) imply
that, for $\left|\mathbf{x}\right|\gg1$, (\ref{eq:B3_3D}) can be
written as 

\begin{equation}
B_{3}\left(\mathbf{x},h\right)=B_{3}^{\text{asympt}}\left(\mathbf{x},h\right)+\mathcal{O}\left(\frac{1}{\left|\mathbf{x}\right|^{7}}\right),\label{eq:B3_asympt}
\end{equation}
with
\begin{align}
B_{3}^{\text{asympt}}\left(\mathbf{x},h\right):= & \frac{c_{0,0}}{\left|\mathbf{x}\right|^{3}}+\frac{c_{1,0}x_{1}+c_{0,1}x_{2}}{\left|\mathbf{x}\right|^{5}}+\frac{c_{2,0}x_{1}^{2}+c_{0,2}x_{2}^{2}+c_{1,1}x_{1}x_{2}}{\left|\mathbf{x}\right|^{7}}\label{eq:B3_asympt_def}\\
 & +\frac{c_{3,0}x_{1}^{3}+c_{0,3}x_{2}^{3}+c_{2,1}x_{1}^{2}x_{2}+c_{1,2}x_{1}x_{2}^{2}}{\left|\mathbf{x}\right|^{9}},\nonumber 
\end{align}

\begin{equation}
c_{0,0}:=-\frac{m_{3}}{4\pi},\label{eq:a0_def}
\end{equation}
\begin{equation}
c_{1,0}:=\frac{3}{4\pi}\left[\left\langle \left(h-x_{3}\right)M_{1}\right\rangle -\left\langle x_{1}M_{3}\right\rangle \right],\hspace{1em}\hspace{1em}c_{0,1}:=\frac{3}{4\pi}\left[\left\langle \left(h-x_{3}\right)M_{2}\right\rangle -\left\langle x_{2}M_{3}\right\rangle \right],\label{eq:a1_def}
\end{equation}
\begin{align}
c_{2,0}:= & \frac{3}{8\pi}\left[8\left\langle \left(h-x_{3}\right)x_{1}M_{1}\right\rangle -4\left\langle x_{1}^{2}M_{3}\right\rangle +\left\langle x_{2}^{2}M_{3}\right\rangle\right.\label{eq:a31_def}\\
 & \left.-2\left\langle \left(h-x_{3}\right)x_{2}M_{2}\right\rangle+3\left\langle \left(h-x_{3}\right)^{2}M_{3}\right\rangle \right],\nonumber 
\end{align}
\begin{align}
c_{0,2}:= & \frac{3}{8\pi}\left[8\left\langle \left(h-x_{3}\right)x_{2}M_{2}\right\rangle -4\left\langle x_{2}^{2}M_{3}\right\rangle +\left\langle x_{1}^{2}M_{3}\right\rangle\right.\label{eq:a32_def}\\
 & \left.-2\left\langle \left(h-x_{3}\right)x_{1}M_{1}\right\rangle+3\left\langle \left(h-x_{3}\right)^{2}M_{3}\right\rangle \right],\nonumber 
\end{align}
\begin{equation}
c_{1,1}:=\frac{15}{4\pi}\left[\left\langle \left(h-x_{3}\right)x_{2}M_{1}\right\rangle +\left\langle \left(h-x_{3}\right)x_{1}M_{2}\right\rangle -\left\langle x_{1}x_{2}M_{3}\right\rangle \right],\label{eq:a33_def}
\end{equation}
\begin{align}
c_{3,0}:= & \frac{5}{8\pi}\left[12\left\langle \left(h-x_{3}\right)x_{1}^{2}M_{1}\right\rangle -4\left\langle x_{1}^{3}M_{3}\right\rangle -3\left\langle \left(h-x_{3}\right)x_{2}^{2}M_{1}\right\rangle -3\left\langle \left(h-x_{3}\right)^{3}M_{1}\right\rangle \right.\label{eq:a51_def}\\
 & \left.-6\left\langle \left(h-x_{3}\right)x_{1}x_{2}M_{2}\right\rangle +3\left\langle x_{1}x_{2}^{2}M_{3}\right\rangle +9\left\langle \left(h-x_{3}\right)^{2}x_{1}M_{3}\right\rangle \right],\nonumber 
\end{align}
\begin{align}
c_{0,3}:= & \frac{5}{8\pi}\left[12\left\langle \left(h-x_{3}\right)x_{2}^{2}M_{2}\right\rangle -4\left\langle x_{2}^{3}M_{3}\right\rangle -3\left\langle \left(h-x_{3}\right)x_{1}^{2}M_{2}\right\rangle -3\left\langle \left(h-x_{3}\right)^{3}M_{2}\right\rangle \right.\label{eq:a52_def}\\
 & \left.-6\left\langle \left(h-x_{3}\right)x_{1}x_{2}M_{1}\right\rangle +3\left\langle x_{1}^{2}x_{2}M_{3}\right\rangle +9\left\langle \left(h-x_{3}\right)^{2}x_{2}M_{3}\right\rangle \right],\nonumber 
\end{align}
\begin{align}
c_{2,1}:= & \frac{15}{8\pi}\left[6\left\langle \left(h-x_{3}\right)x_{1}^{2}M_{2}\right\rangle +12\left\langle \left(h-x_{3}\right)x_{1}x_{2}M_{1}\right\rangle -6\left\langle x_{1}^{2}x_{2}M_{3}\right\rangle \right.\label{eq:a53_def}\\
 & \left.-3\left\langle \left(h-x_{3}\right)x_{2}^{2}M_{2}\right\rangle -\left\langle \left(h-x_{3}\right)^{3}M_{2}\right\rangle +\left\langle x_{2}^{3}M_{3}\right\rangle +3\left\langle \left(h-x_{3}\right)^{2}x_{2}M_{3}\right\rangle \right],\nonumber 
\end{align}
\begin{align}
c_{1,2}:= & \frac{15}{8\pi}\left[6\left\langle \left(h-x_{3}\right)x_{2}^{2}M_{1}\right\rangle +12\left\langle \left(h-x_{3}\right)x_{1}x_{2}M_{2}\right\rangle -6\left\langle x_{1}x_{2}^{2}M_{3}\right\rangle \right.\label{eq:a54_def}\\
 & \left.-3\left\langle \left(h-x_{3}\right)x_{1}^{2}M_{1}\right\rangle -\left\langle \left(h-x_{3}\right)^{3}M_{1}\right\rangle +\left\langle x_{1}^{3}M_{3}\right\rangle +3\left\langle \left(h-x_{3}\right)^{2}x_{1}M_{3}\right\rangle \right],\nonumber 
\end{align}
and condition (\ref{eq:cond_asympt}), for our particular context,
rewrites as (\ref{eq:cond_asympt_thm}). 

We are going to pursue the idea outlined in the previous subsection.
Namely, comparing a series expansion of (\ref{eq:B3_hat_k1}) about
$k_{1}=0$ with that of (\ref{eq:B3_decomp}), we shall deduce a set
of identities which relate magnetisation moments to the integrals
of the measured data $B_{3}\left(\mathbf{x},h\right)$ on $D_{A}$.
More precisely, using compactness of the support of the magnetisation
$\vec{\mathcal{M}}$, it follows from (\ref{eq:B3_hat_k1}) that $\text{Re }\hat{B}_{3}\left(k_{1},0,h\right)$
is a convergent series in powers of $\left|k_{1}\right|$ whereas
$\text{Im }\hat{B}_{3}\left(k_{1},0,h\right)$ is a power series in
$\left|k_{1}\right|$ multiplied by $k_{1}$. To facilitate the situation
of matching the coefficients of different representations, we shall
focus on the region $k_{1}>0$. Consequently, in what follows, evaluation
at $k_{1}=0^{+}$ will mean the limiting value at $k_{1}=0$ taken
from the positive semiaxis ($k_{1}>0$).

\subsubsection{\label{subsec:tang_estim} Tangential components of the net moment}

Expanding the integrand in (\ref{eq:B3_hat_k1}) in power series in
a positive neighborhood of $k_{1}=0$ and taking the imaginary part,
we deduce: 

\begin{align}
\left.\partial_{k_{1}}\left[\text{Im }\hat{B}_{3}\left(k_{1},0,h\right)\right]\right|_{k_{1}=0^{+}} & =\pi m_{1}=:d_{1,}\label{eq:d1_def}
\end{align}
\begin{align}
\frac{1}{6}\left.\partial_{k_{1}}^{3}\left[\text{Im }\hat{B}_{3}\left(k_{1},0,h\right)\right]\right|_{k_{1}=0^{+}} & =2\pi^{3}\left[\left\langle \left(h-x_{3}\right)^{2}M_{1}\right\rangle -\left\langle x_{1}^{2}M_{1}\right\rangle -2\left\langle \left(h-x_{3}\right)x_{1}M_{3}\right\rangle \right]\label{eq:d3_def}\\
 & =:d_{3},\nonumber 
\end{align}
\begin{align}
\frac{1}{5!}\left.\partial_{k_{1}}^{5}\left[\text{Im }\hat{B}_{3}\left(k_{1},0,h\right)\right]\right|_{k_{1}=0^{+}} & =\frac{2\pi^{5}}{3}\left[\left\langle \left(h-x_{3}\right)^{4}M_{1}\right\rangle -6\left\langle \left(h-x_{3}\right)^{2}x_{1}^{2}M_{1}\right\rangle +\left\langle x_{1}^{4}M_{1}\right\rangle \right.\label{eq:d5_def}\\
 & \left.-4\left\langle \left(h-x_{3}\right)^{3}x_{1}M_{3}\right\rangle +4\left\langle \left(h-x_{3}\right)x_{1}^{3}M_{3}\right\rangle \right]\nonumber \\
 & =:d_{5},\nonumber 
\end{align}
\begin{align}
\frac{1}{7!}\left.\partial_{k_{1}}^{7}\left[\text{Im }\hat{B}_{3}\left(k_{1},0,h\right)\right]\right|_{k_{1}=0^{+}}= & \frac{4\pi^{7}}{45}\left[\left\langle \left(h-x_{3}\right)^{6}M_{1}\right\rangle -15\left\langle \left(h-x_{3}\right)^{4}x_{1}^{2}M_{1}\right\rangle +15\left\langle \left(h-x_{3}\right)^{2}x_{1}^{4}M_{1}\right\rangle -\left\langle x_{1}^{6}M_{1}\right\rangle \right.\label{eq:d7_def}\\
 & \left.-6\left\langle \left(h-x_{3}\right)^{5}x_{1}M_{3}\right\rangle +20\left\langle \left(h-x_{3}\right)^{3}x_{1}^{3}M_{3}\right\rangle -6\left\langle \left(h-x_{3}\right)x_{1}^{5}M_{3}\right\rangle \right]\nonumber \\
 & =:d_{7},\nonumber 
\end{align}
\begin{align}
\frac{1}{9!}\left.\partial_{k_{1}}^{9}\left[\text{Im }\hat{B}_{3}\left(k_{1},0,h\right)\right]\right|_{k_{1}=0^{+}}= & \frac{2\pi^{9}}{315}\left[\left\langle \left(h-x_{3}\right)^{8}M_{1}\right\rangle -28\left\langle \left(h-x_{3}\right)^{6}x_{1}^{2}M_{1}\right\rangle +70\left\langle \left(h-x_{3}\right)^{4}x_{1}^{4}M_{1}\right\rangle \right.\label{eq:d9_def}\\
 & -28\left\langle \left(h-x_{3}\right)^{2}x_{1}^{6}M_{1}\right\rangle +\left\langle x_{1}^{8}M_{1}\right\rangle -8\left\langle \left(h-x_{3}\right)^{7}x_{1}M_{3}\right\rangle \nonumber \\
 & \left.+56\left\langle \left(h-x_{3}\right)^{5}x_{1}^{3}M_{3}\right\rangle -56\left\langle \left(h-x_{3}\right)^{3}x_{1}^{5}M_{3}\right\rangle +8\left\langle \left(h-x_{3}\right)x_{1}^{7}M_{3}\right\rangle \right]\nonumber \\
 & =:d_{9},\nonumber 
\end{align}
\begin{align}
\frac{1}{11!}\left.\partial_{k_{1}}^{11}\left[\text{Im }\hat{B}_{3}\left(k_{1},0,h\right)\right]\right|_{k_{1}=0^{+}}= & \frac{4\pi^{11}}{14175}\left[\left\langle \left(h-x_{3}\right)^{10}M_{1}\right\rangle -45\left\langle \left(h-x_{3}\right)^{8}x_{1}^{2}M_{1}\right\rangle +210\left\langle \left(h-x_{3}\right)^{6}x_{1}^{4}M_{1}\right\rangle \right.\label{eq:d11_def}\\
 & -210\left\langle \left(h-x_{3}\right)^{4}x_{1}^{6}M_{1}\right\rangle +45\left\langle \left(h-x_{3}\right)^{2}x_{1}^{8}M_{1}\right\rangle -\left\langle x_{1}^{10}M_{1}\right\rangle -10\left\langle \left(h-x_{3}\right)^{9}x_{1}M_{3}\right\rangle \nonumber \\
 & +120\left\langle \left(h-x_{3}\right)^{7}x_{1}^{3}M_{3}\right\rangle -252\left\langle \left(h-x_{3}\right)^{5}x_{1}^{5}M_{3}\right\rangle +120\left\langle \left(h-x_{3}\right)^{3}x_{1}^{7}M_{3}\right\rangle \nonumber \\
 & \left.-10\left\langle \left(h-x_{3}\right)x_{1}^{9}M_{3}\right\rangle \right]\nonumber \\
 & =:d_{11}.\nonumber 
\end{align}

On the other hand, from (\ref{eq:B3_decomp}), we have, for $n\in\mathbb{N}_{0}$,
\begin{align}
\left.\partial_{k_{1}}^{2n+1}\left[\text{Im }\hat{B}_{3}\left(k_{1},0,h\right)\right]\right|_{k_{1}=0^{+}}= & \left(-1\right)^{n}\left(2\pi\right)^{2n+1}\iint_{D_{A}}x_{1}^{2n+1}B_{3}\left(\mathbf{x},h\right)\dd^{2}x\label{eq:ImB3_int}\\
 & +\left.\partial_{k_{1}}^{2n+1}\left(\iint_{\mathbb{R}^{2}\backslash D_{A}}\sin\left(2\pi k_{1}x_{1}\right)B_{3}^{\text{asympt}}\left(\mathbf{x},h\right)\dd^{2}x\right)\right|_{k_{1}=0^{+}}\nonumber \\
 & +\left.\partial_{k_{1}}^{2n+1}\left(\iint_{\mathbb{R}^{2}\backslash D_{A}}\sin\left(2\pi k_{1}x_{1}\right)\left[B_{3}\left(\mathbf{x},h\right)-B_{3}^{\text{asympt}}\left(\mathbf{x},h\right)\right]\dd^{2}x\right)\right|_{k_{1}=0^{+}}.\nonumber 
\end{align}

We shall proceed in 3 steps. First, we evaluate the integral 
\begin{equation}
\mathcal{I}^{\sin}\equiv\mathcal{I}^{\sin}\left(k_{1},A\right):=\iint_{\mathbb{R}^{2}\backslash D_{A}}\sin\left(2\pi k_{1}x_{1}\right)B_{3}^{\text{asympt}}\left(\mathbf{x},h\right)\dd^{2}x,\label{eq:I_sin}
\end{equation}
and compute its derivatives appearing on the second line of (\ref{eq:ImB3_int}),
hence obtaining a set of valuable identities. Second, we estimate
the derivatives of the remainder
\begin{equation}
\mathcal{R}_{2n+1}^{\sin}\equiv\mathcal{R}_{2n+1}^{\sin}\left(A\right):=\left.\partial_{k_{1}}^{2n+1}\left(\iint_{\mathbb{R}^{2}\backslash D_{A}}\sin\left(2\pi k_{1}x_{1}\right)\left[B_{3}\left(\mathbf{x},h\right)-B_{3}^{\text{asympt}}\left(\mathbf{x},h\right)\right]\dd^{2}x\right)\right|_{k_{1}=0^{+}},\hspace{1em}n\in\mathbb{N}_{0},\label{eq:Rn_sin}
\end{equation}
in order to show that the contribution of the term in the third line
of (\ref{eq:ImB3_int}) is not significant for large $A$ (for the
chosen order of the asymptotic expansion). Finally, at the last step,
we combine the obtained identities and derive asymptotic formulas
(\ref{eq:fin_estim_m1m2_ord1})--(\ref{eq:fin_estim_m1m2_ord2}),
(\ref{eq:fin_estim_m1m2_ord3}), (\ref{eq:fin_estim_m1m2_ord4}),
(\ref{eq:fin_estim_m1m2_ord5}) in a rigorously justified fashion.

\medskip

\subparagraph*{\uline{\label{par:tang_estim_step1} Step 1: Derivation of the
set of identities}}

\medskip

Using (\ref{eq:B3_asympt_def}) and passing to the polar coordinates
using $x=r\cos\theta$, $y=r\sin\theta$, $\dd^{2}x=r\dd r\dd\theta$, we
obtain from (\ref{eq:I_sin})
\begin{equation}
\mathcal{I}^{\sin}=c_{1,0}\mathcal{I}_{1}^{\sin}+
+c_{3,0}\mathcal{I}_{2}^{\sin}+c_{1,2}\mathcal{I}_{3}^{\sin},\label{eq:I_sin_decomp}
\end{equation}
where 
\begin{equation}
\mathcal{I}_{1}^{\sin}:=\int_{A}^{\infty}\int_{0}^{2\pi}\sin\left(2\pi k_{1}r\cos\theta\right)\cos\theta \dd\theta\frac{\dd r}{r^{3}},\label{eq:I_sin1_def}
\end{equation}
\begin{equation}
\mathcal{I}_{2}^{\sin}:=\int_{A}^{\infty}\int_{0}^{2\pi}\sin\left(2\pi k_{1}r\cos\theta\right)\cos^{3}\theta \dd\theta\frac{\dd r}{r^{5}},\label{eq:I_sin3_def}
\end{equation}
\begin{equation}
\mathcal{I}_{3}^{\sin}:=\int_{A}^{\infty}\int_{0}^{2\pi}\sin\left(2\pi k_{1}r\cos\theta\right)\cos\theta\sin^{2}\theta \dd\theta\frac{\dd r}{r^{5}}.\label{eq:I_sin4_def}
\end{equation}
Here, we used results (\ref{eq:trig_int_sin1})--(\ref{eq:trig_int_sin2})
of Lemma \ref{lem:trig_ints} multiple times to deduce vanishing of
the integrals associated with the terms which involve $c_{0,0}$, $c_{0,1}$, 
$c_{2,0}$, $c_{0,2}$, $c_{1,1}$,  $c_{0,3}$, $c_{2,1}$.

We now employ the integral representation of Bessel functions, given
in (\ref{eq:BessJ1_int_repr}), to rewrite, for $k_{1}>0$,
\begin{equation}
\mathcal{I}_{1}^{\sin}=2\pi\int_{A}^{\infty}J_{1}\left(2\pi k_{1}r\right)\frac{\dd r}{r^{3}}=\left(2\pi\right)^{3}k_{1}^{2}\int_{2\pi k_{1}A}^{\infty}\frac{J_{1}\left(x\right)}{x^{3}}\dd x,\label{eq:I_sin1}
\end{equation}
\begin{align}
\mathcal{I}_{2}^{\sin} & =-2\pi\int_{A}^{\infty}J_{1}^{\prime\prime}\left(2\pi k_{1}r\right)\frac{\dd r}{r^{5}}=-\left(2\pi\right)^{5}k_{1}^{4}\int_{2\pi k_{1}A}^{\infty}\frac{J_{1}^{\prime\prime}\left(x\right)}{x^{5}}\dd x\label{eq:I_sin3}\\
 & =\left(2\pi\right)^{5}k_{1}^{4}\left.\left(\frac{5J_{1}\left(\rho\right)}{\rho^{6}}+\frac{J_{1}^{\prime}\left(\rho\right)}{\rho^{5}}-30\int_{\rho}^{\infty}\frac{J_{1}\left(x\right)}{x^{7}}\dd x\right)\right|_{\rho=2\pi k_{1}A},\nonumber 
\end{align}
\begin{align}
\mathcal{I}_{3}^{\sin} & =2\pi\int_{A}^{\infty}\left[J_{1}\left(2\pi k_{1}r\right)+J_{1}^{\prime\prime}\left(2\pi k_{1}r\right)\right]\frac{\dd r}{r^{5}}=\left(2\pi\right)^{5}k_{1}^{4}\int_{2\pi k_{1}A}^{\infty}\frac{J_{1}\left(x\right)+J_{1}^{\prime\prime}\left(x\right)}{x^{5}}\dd x,\label{eq:I_sin4}\\
 & =-\left(2\pi\right)^{5}k_{1}^{4}\left.\left(\frac{5J_{1}\left(\rho\right)}{\rho^{6}}+\frac{J_{1}^{\prime}\left(\rho\right)}{\rho^{5}}-\int_{\rho}^{\infty}\frac{J_{1}\left(x\right)}{x^{5}}\dd x-30\int_{\rho}^{\infty}\frac{J_{1}\left(x\right)}{x^{7}}\dd x\right)\right|_{\rho=2\pi k_{1}A}.\nonumber 
\end{align}
Note that, in (\ref{eq:I_sin3})--(\ref{eq:I_sin4}), we employed
integration by parts twice using the asymptotic behaviour of $J_{1}$
given in (\ref{eq:BessJ_asympt}).

Using the results of Lemmas \ref{lem:int_J1_x3}--\ref{lem:int_J1_var},
we have
\begin{align}
\mathcal{I}_{1}^{\sin} & =\frac{\left(2\pi\right)^{2}k_{1}}{A}\left.\left(\rho\int_{\rho}^{\infty}\frac{J_{1}\left(x\right)}{x^{3}}\dd x\right)\right|_{\rho=2\pi k_{1}A}\label{eq:I_sin1_fin}\\
 & =\frac{\left(2\pi\right)^{2}k_{1}}{3A}\left.\left[J_{0}\left(\rho\right)+\frac{J_{1}\left(\rho\right)}{\rho}-\rho-\rho J_{1}\left(\rho\right)+\rho^{2}J_{0}\left(\rho\right)-\frac{\pi}{2}\rho^{2}J_{0}\left(\rho\right)H_{1}\left(\rho\right)+\frac{\pi}{2}\rho^{2}J_{1}\left(\rho\right)H_{0}\left(\rho\right)\right]\right|_{\rho=2\pi k_{1}A},\nonumber 
\end{align}
\begin{align}
\mathcal{I}_{2}^{\sin} & =\frac{\left(2\pi\right)^{2}k_{1}}{A^{3}}\left.\left(\frac{5J_{1}\left(\rho\right)}{\rho^{3}}+\frac{J_{1}^{\prime}\left(\rho\right)}{\rho^{2}}-30\rho^{3}\int_{\rho}^{\infty}\frac{J_{1}\left(x\right)}{x^{7}}\dd x\right)\right|_{\rho=2\pi k_{1}A}\label{eq:I_sin3_fin}\\
 & =\frac{\left(2\pi\right)^{2}k_{1}}{105A^{3}}\left[-\frac{15J_{1}\left(\rho\right)}{\rho^{3}}+\frac{15J_{1}^{\prime}\left(\rho\right)}{\rho^{2}}+\frac{24J_{1}\left(\rho\right)}{\rho}+6J_{1}^{\prime}\left(\rho\right)-2\rho^{2}J_{0}\left(\rho\right)-2\rho J_{1}\left(\rho\right)\right.\nonumber \\
 & \left.+2\rho^{3}+2\rho^{3}J_{1}\left(\rho\right)-2\rho^{4}J_{0}\left(\rho\right)+\pi\rho^{4}J_{0}\left(\rho\right)H_{1}\left(\rho\right)-\pi\rho^{4}J_{1}\left(\rho\right)H_{0}\left(\rho\right)\biggr]\right|_{\rho=2\pi k_{1}A},\nonumber 
\end{align}
\begin{align}
\mathcal{I}_{3}^{\sin} & =-\frac{\left(2\pi\right)^{2}k_{1}}{A^{3}}\left.\left(\frac{5J_{1}\left(\rho\right)}{\rho^{3}}+\frac{J_{1}^{\prime}\left(\rho\right)}{\rho^{2}}-\rho^{3}\int_{\rho}^{\infty}\frac{J_{1}\left(x\right)}{x^{5}}\dd x-30\rho^{3}\int_{\rho}^{\infty}\frac{J_{1}\left(x\right)}{x^{7}}\dd x\right)\right|_{\rho=2\pi k_{1}A}\label{eq:I_sin4_fin}\\
 & =\frac{\left(2\pi\right)^{2}k_{1}}{105A^{3}}\left[\frac{15J_{1}\left(\rho\right)}{\rho^{3}}-\frac{15J_{1}^{\prime}\left(\rho\right)}{\rho^{2}}+\frac{4J_{1}\left(\rho\right)}{\rho}+J_{1}^{\prime}\left(\rho\right)-\frac{\rho^{2}J_{0}\left(\rho\right)}{3}-\frac{\rho J_{1}\left(\rho\right)}{3}\right.\nonumber \\
 & \left.+\frac{\rho^{3}}{3}+\frac{\rho^{3}J_{1}\left(\rho\right)}{3}-\frac{\rho^{4}J_{0}\left(\rho\right)}{3}+\frac{\pi\rho^{4}}{6}J_{0}\left(\rho\right)H_{1}\left(\rho\right)-\frac{\pi\rho^{4}}{6}J_{1}\left(\rho\right)H_{0}\left(\rho\right)\biggr]\right|_{\rho=2\pi k_{1}A}.\nonumber 
\end{align}
Therefore, (\ref{eq:I_sin_decomp}) together with (\ref{eq:I_sin1_fin})--(\ref{eq:I_sin4_fin})
furnishes an explicit form of (\ref{eq:I_sin}). In particular, using
(\ref{eq:BessJ_series}), (\ref{eq:StruvH}), we can compute
\begin{equation}
\left.\partial_{k_{1}}\mathcal{I}^{\sin}\right|_{k_{1}=0^{+}}=\frac{1}{2}\left(2\pi\right)^{2}\left(\frac{c_{1,0}}{A}+\frac{3c_{3,0}+c_{1,2}}{12A^{3}}\right),\label{eq:I_sin_d1}
\end{equation}
\begin{equation}
\left.\partial_{k_{1}}^{3}\mathcal{I}^{\sin}\right|_{k_{1}=0^{+}}=\frac{3}{8}\left(2\pi\right)^{4}A^{2}\left(\frac{c_{1,0}}{A}-\frac{5c_{3,0}+c_{1,2}}{6A^{3}}\right),\label{eq:I_sin_d3}
\end{equation}
\begin{equation}
\left.\partial_{k_{1}}^{5}\mathcal{I}^{\sin}\right|_{k_{1}=0^{+}}=-\frac{5}{16}\left(2\pi\right)^{6}A^{4}\left(\frac{c_{1,0}}{3A}+\frac{7c_{3,0}+c_{1,2}}{8A^{3}}\right),\label{eq:I_sin_d5}
\end{equation}
\begin{equation}
\left.\partial_{k_{1}}^{7}\mathcal{I}^{\sin}\right|_{k_{1}=0^{+}}=\frac{7}{128}\left(2\pi\right)^{8}A^{6}\left(\frac{c_{1,0}}{A}+\frac{9c_{3,0}+c_{1,2}}{6A^{3}}\right),\label{eq:I_sin_d7}
\end{equation}
\begin{equation}
\left.\partial_{k_{1}}^{9}\mathcal{I}^{\sin}\right|_{k_{1}=0^{+}}=-\frac{21}{256}\left(2\pi\right)^{10}A^{8}\left(\frac{3c_{1,0}}{7A}+\frac{11c_{3,0}+c_{1,2}}{20A^{3}}\right),\label{eq:I_sin_d9}
\end{equation}
\begin{equation}
\left.\partial_{k_{1}}^{11}\mathcal{I}^{\sin}\right|_{k_{1}=0^{+}}=\frac{33}{14336}\left(2\pi\right)^{12}A^{10}\left(\frac{98c_{1,0}}{9A}+\frac{13c_{3,0}+c_{1,2}}{A^{3}}\right).\label{eq:I_sin_d11}
\end{equation}
Taking into account (\ref{eq:I_sin}), (\ref{eq:Rn_sin}), we use
(\ref{eq:d1_def})--(\ref{eq:d11_def}) and (\ref{eq:I_sin_d1})--(\ref{eq:I_sin_d11})
in (\ref{eq:ImB3_int}) with $n=0,\ldots,5$, respectively, and thus
arrive at the following set of identities:
\begin{equation}
2\pi\iint_{D_{A}}x_{1}B_{3}\left(\mathbf{x},h\right)\dd^{2}x+2\pi^{2}\left(\frac{c_{1,0}}{A}+\frac{3c_{3,0}+c_{1,2}}{12A^{3}}\right)+\mathcal{R}_{1}^{\sin}=d_{1},\label{eq:d1_ident}
\end{equation}
\begin{equation}
-\frac{4\pi^{3}}{3}\iint_{D_{A}}x_{1}^{3}B_{3}\left(\mathbf{x},h\right)\dd^{2}x+\pi^{4}A^{2}\left(\frac{c_{1,0}}{A}-\frac{5c_{3,0}+c_{1,2}}{6A^{3}}\right)+\frac{1}{6}\mathcal{R}_{3}^{\sin}=d_{3},\label{eq:d3_ident}
\end{equation}
\begin{equation}
\frac{4\pi^{5}}{15}\iint_{D_{A}}x_{1}^{5}B_{3}\left(\mathbf{x},h\right)\dd^{2}x-\frac{\pi^{6}A^{4}}{6}\left(\frac{c_{1,0}}{3A}+\frac{7c_{3,0}+c_{1,2}}{8A^{3}}\right)+\frac{1}{5!}\mathcal{R}_{5}^{\sin}=d_{5},\label{eq:d5_ident}
\end{equation}
\begin{equation}
-\frac{\left(2\pi\right)^{7}}{5040}\iint_{D_{A}}x_{1}^{7}B_{3}\left(\mathbf{x},h\right)\dd^{2}x+\frac{\left(2\pi\right)^{8}A^{6}}{92160}\left(\frac{c_{1,0}}{A}+\frac{9c_{3,0}+c_{1,2}}{6A^{3}}\right)+\frac{1}{7!}\mathcal{R}_{7}^{\sin}=d_{7},\label{eq:d7_ident}
\end{equation}
\begin{equation}
\frac{\left(2\pi\right)^{9}}{362880}\iint_{D_{A}}x_{1}^{9}B_{3}\left(\mathbf{x},h\right)\dd^{2}x-\frac{7\left(2\pi\right)^{10}A^{8}}{10321920}\left(\frac{c_{1,0}}{7A}+\frac{11c_{3,0}+c_{1,2}}{60A^{3}}\right)+\frac{1}{9!}\mathcal{R}_{9}^{\sin}=d_{9},\label{eq:d9_ident}
\end{equation}
\begin{equation}
-\frac{\left(2\pi\right)^{11}}{39916800}\iint_{D_{A}}x_{1}^{11}B_{3}\left(\mathbf{x},h\right)\dd^{2}x+\frac{\left(2\pi\right)^{12}A^{10}}{17340825600}\left(\frac{98c_{1,0}}{9A}+\frac{13c_{3,0}+c_{1,2}}{A^{3}}\right)+\frac{1}{11!}\mathcal{R}_{11}^{\sin}=d_{11}.\label{eq:d11_ident}
\end{equation}

\medskip

\subparagraph*{\uline{\label{par:tang_estim_step2} Step 2: Analysis of the remainder
terms \mbox{$\mathcal{R}_{2n+1}^{\sin}$} for \mbox{$0\protect\leq n\protect\leq5$}}}

\medskip

We shall now show that the remainder terms $\mathcal{R}_{2n+1}^{\sin}$
with $n=0,\ldots,5$ given by (\ref{eq:Rn_sin}) (with $k_{1}>0$,
as assumed before), can be estimated, for $A\gg1$, as follows
\begin{equation}
\mathcal{R}_{2n+1}^{\sin}=\mathcal{O}\left(\frac{1}{A^{5-2n}}\right),\hspace{1em}\hspace{1em}0\leq n\leq5.\label{eq:Rn_sin_estim}
\end{equation}

Proceeding with higher-order terms in the expansions in (\ref{eq:expans_alg32})--(\ref{eq:expans_alg52}),
and hence also in (\ref{eq:B3_asympt}), we notice that we can write,
for $N\geq4$,
\begin{equation}
B_{3}\left(\mathbf{x},h\right)-B_{3}^{\text{asympt}}\left(\mathbf{x},h\right)=\sum_{q=4}^{N}\mathop{\sum\sum}_{\underset{l_{1}+l_{2}=q}{l_{1},\,l_{2}\geq0,}}c_{l_{1},l_{2}}\frac{x_{1}^{l_{1}}x_{2}^{l_{2}}}{\left|\mathbf{x}\right|^{2q+3}}+\mathcal{O}\left(\frac{1}{\left|\mathbf{x}\right|^{N+4}}\right)=:L_{N}\left(\mathbf{x}\right)+\mathcal{O}\left(\frac{1}{\left|\mathbf{x}\right|^{N+4}}\right),\label{eq:LN_def}
\end{equation}
with some constants $c_{l_{1},l_{2}}\in\mathbb{R}$ for $l_{1}$,
$l_{2}\in\mathbb{N}_{0}$. Consequently, we consider, for $n=0,\ldots,5$,
\begin{align}
\mathcal{R}_{2n+1}^{\sin}= & \left.\partial_{k_{1}}^{2n+1}\left(\iint_{\mathbb{R}^{2}\backslash D_{A}}\sin\left(2\pi k_{1}x_{1}\right)L_{N}\left(\mathbf{x}\right)\dd^{2}x\right)\right|_{k_{1}=0^{+}}\label{eq:Rn_sin_decomp}\\
 & +\left.\partial_{k_{1}}^{2n+1}\left(\iint_{\mathbb{R}^{2}\backslash D_{A}}\sin\left(2\pi k_{1}x_{1}\right)\left[B_{3}\left(\mathbf{x},h\right)-B_{3}^{\text{asympt}}\left(\mathbf{x},h\right)-L_{N}\left(\mathbf{x}\right)\right]\dd^{2}x\right)\right|_{k_{1}=0^{+}}.\nonumber 
\end{align}

First of all, we deal with the term on the second line. The integrand
is regular and, for $N>9$, it decays at infinity sufficiently fast
so that the differential operator $\partial_{k_{1}}^{2n+1}$ with
$0\leq n\leq5$ can be passed under the integral sign. We can thus
estimate
\[
\left|\partial_{k_{1}}^{2n+1}\left(\iint_{\mathbb{R}^{2}\backslash D_{A}}\sin\left(2\pi k_{1}x_{1}\right)\left[B_{3}\left(\mathbf{x},h\right)-B_{3}^{\text{asympt}}\left(\mathbf{x},h\right)-L_{N}\left(\mathbf{x}\right)\right]\dd^{2}x\right)\Bigr|_{k_{1}=0^{+}}\right|
\]
\begin{align*}
= & \left|\left(\iint_{\mathbb{R}^{2}\backslash D_{A}}\frac{\left(2\pi x_{1}\right)^{2n+1}}{\left|\mathbf{x}\right|^{N+4}}\left|\mathbf{x}\right|^{N+4}\left[B_{3}\left(\mathbf{x},h\right)-B_{3}^{\text{asympt}}\left(\mathbf{x},h\right)-L_{N}\left(\mathbf{x}\right)\right]\dd^{2}x\right)\right|\\
\leq & \left(2\pi\right)^{2\left(n+1\right)}C_{N}\int_{A}^{\infty}\frac{dr}{r^{N-2n+1}}=\frac{\left(2\pi\right)^{2\left(n+1\right)}C_{N}}{\left(N-2n\right)}\frac{1}{A^{N-2n}}
\end{align*}
for some constant $C_{N}>0$ such that
\[
\left|\mathbf{x}\right|^{N+4}\left|B_{3}\left(\mathbf{x},h\right)-B_{3}^{\text{asympt}}\left(\mathbf{x},h\right)-L_{N}\left(\mathbf{x}\right)\right|\leq C_{N},\hspace{1em}\left|\mathbf{x}\right|\geq A,
\]
and such a bound is possible due to the remainder estimate $\mathcal{O}\left(1/\left|\mathbf{x}\right|^{N+4}\right)$
in (\ref{eq:LN_def}). Here, in the third line of the estimates, we
used the fact that the integral in $r$ converges for $N>2n-1$, $0\leq n\leq5$,
which is true for $N>9$. For such $N$, the obtained estimate of
order $\mathcal{O}\left(1/A^{N-2n}\right)$ is clearly even better
than was aimed for (recall (\ref{eq:Rn_sin_estim})).

We now fix $N=10$ and proceed with estimating the term in the first
line of (\ref{eq:Rn_sin_decomp}). Upon substitution of (\ref{eq:LN_def})
in (\ref{eq:Rn_sin}) and use of polar coordinates (with $x_{1}=r\cos\theta$,
$x_{2}=r\sin\theta$, as before), let us observe that, due to Lemma
\ref{lem:trig_ints} (namely, identities (\ref{eq:trig_int_sin1})--(\ref{eq:trig_int_sin2})),
the only non-vanishing terms stemming from the $L_{N}$ part are those
proportional to 
\begin{equation}
\int_{A}^{\infty}\int_{0}^{2\pi}\sin\left(2\pi k_{1}r\cos\theta\right)\cos^{2\left(p-l\right)+1}\theta\,\,\sin^{2l}\theta \dd\theta\frac{\dd r}{r^{2p+3}},\hspace{1em}\hspace{1em}0\leq l\leq p,\hspace{1em}p\geq2.\label{eq:sin_terms_nonvan}
\end{equation}
Since we can write
\[
\sin^{2l}\theta=\left(1-\cos^{2}\theta\right)^{l}=\sum_{j=0}^{l}\left(\begin{array}{c}
l\\
j
\end{array}\right)\left(-1\right)^{l-j}\cos^{2\left(l-j\right)}\theta,
\]
with $\left(\begin{array}{c}
l\\
j
\end{array}\right)$ denoting a binomial coefficient, we deduce that, to estimate $\mathcal{R}_{2n+1}^{\sin}$,
it suffices only to consider the quantities
\begin{align}
\mathcal{S}_{p,j}:= & \int_{A}^{\infty}\int_{0}^{2\pi}\sin\left(2\pi k_{1}r\cos\theta\right)\cos^{2j+1}\theta\,\dd\theta\frac{\dd r}{r^{2p+3}},\hspace{1em}\hspace{1em}0\leq j\leq p,\hspace{1em}p\geq2,\label{eq:S_pj_def}
\end{align}
and, in particular, their derivatives evaluated at $k_{1}=0$ from
the right: $\left.\partial_{k_{1}}^{2n+1}\mathcal{S}_{p,j}\right|_{k_{1}=0^{+}}$,
$0\leq n\leq5$.

Note that the relation between $p$ (\ref{eq:sin_terms_nonvan})--(\ref{eq:S_pj_def})
and $q$ in (\ref{eq:LN_def}) is $q=2p+1$, $p\geq2$. In other words,
in the asymptotic expansion of the field $B_{3}$ at infinity, not
every term contributes to $\mathcal{R}_{2n+1}^{\sin}$, but only the
terms of every second order in $1/A$, i.e., $\mathcal{O}\left(1/A^{8}\right)$,
$\mathcal{O}\left(1/A^{10}\right)$ and so on.\\
Therefore, we can consider $\mathcal{S}_{p,j}$ only for $p\leq\left\lfloor \frac{N-1}{2}\right\rfloor =8$,
where $\left\lfloor x\right\rfloor $ designates the integer part
of $x$.

To sum up, we need to show that, for all $0\leq j\leq p$, $2\leq p\leq8$
and $0\leq n\leq5$, we are able to produce an estimate
\begin{equation}
\left.\partial_{k_{1}}^{2n+1}\mathcal{S}_{p,j}\right|_{k_{1}=0^{+}}=\mathcal{O}\left(\frac{1}{A^{5-2n}}\right).\label{eq:S_pj_estim}
\end{equation}

For $0\leq n\leq p$, we have
\[
\left.\partial_{k_{1}}^{2n+1}\mathcal{S}_{p,j}\right|_{k_{1}=0^{+}}=\left(-1\right)^{n}\left(2\pi\right)^{2n+1}\int_{A}^{\infty}\int_{0}^{2\pi}\cos^{2j+1}\theta\,\dd \theta\frac{\dd r}{r^{2\left(p-n\right)+2}},
\]
and hence
\[
\left|\left.\partial_{k_{1}}^{2n+1}\mathcal{S}_{p,j}\right|_{k_{1}=0^{+}}\right|\leq\left(2\pi\right)^{2\left(n+1\right)}\int_{A}^{\infty}\frac{\dd r}{r^{2\left(p-n\right)+2}}=\mathcal{O}\left(\frac{1}{A^{2\left(p-n\right)+1}}\right),
\]
where the convergence of the last integral is due to $p\geq n$. The
obtained estimate is in agreement with (\ref{eq:S_pj_estim}) since
$p\ge2$.

To treat the case $n\geq p+1$, a more careful estimate is needed.
To this end, it is convenient to make use of the integral representation
of the Bessel function $J_{1}$ given by (\ref{eq:BessJ1_int_repr})
and rewrite (\ref{eq:S_pj_def}) as
\[
\mathcal{S}_{p,j}=\left(-1\right)^{j}2\pi\int_{A}^{\infty}J_{1}^{\left(2j\right)}\left(2\pi k_{1}r\right)\frac{\dd r}{r^{2p+3}},\hspace{1em}\hspace{1em}0\leq j\leq p,\hspace{1em}p\geq2.
\]
We then evaluate
\begin{align}
\partial_{k_{1}}^{2n+1}\mathcal{S}_{p,j} & =\left(-1\right)^{j}\left(2\pi\right)^{2p+3}\partial_{k_{1}}^{2\left(n-p\right)-1}\int_{A}^{\infty}J_{1}^{\left(2j+2p+2\right)}\left(2\pi k_{1}r\right)\frac{\dd r}{r}\label{eq:S_pj_der}\\
 & =\left(-1\right)^{j+1}\left(2\pi\right)^{2p+4}A\partial_{k_{1}}^{2\left(n-p\right)-2}\left.\frac{J_{1}^{\left(2j+2p+2\right)}\left(\rho\right)}{\rho}\right|_{\rho=2\pi k_{1}A}\nonumber \\
 & =\left(-1\right)^{j+1}\left(2\pi\right)^{2\left(n+1\right)}A^{2\left(n-p\right)-1}\left.\frac{\dd^{2\left(n-p\right)-2}}{\dd\rho^{2\left(n-p\right)-2}}\left(\frac{J_{1}^{\left(2j+2p+2\right)}\left(\rho\right)}{\rho}\right)\right|_{\rho=2\pi k_{1}A},\nonumber 
\end{align}
where, in passing the differential operator $\partial_{k_{1}}^{2p+2}$
under the integral sign, we took into account the asymptotic behaviour
at infinity of $J_{1}$ given by (\ref{eq:BessJ_asympt}) and, on
the second line, employed the following identity valid for $k_{1}>0$:
\begin{align*}
\partial_{k_{1}}\int_{A}^{\infty}J_{1}^{\left(2j+2p+2\right)}\left(2\pi k_{1}r\right)\frac{\dd r}{r} & =\partial_{k_{1}}\int_{2\pi k_{1}A}^{\infty}J_{1}^{\left(2j+2p+2\right)}\left(\rho\right)\frac{\dd\rho}{\rho}\\
 & =-2\pi A\left.\frac{J_{1}^{\left(2j+2p+2\right)}\left(\rho\right)}{\rho}\right|_{\rho=2\pi k_{1}A}.
\end{align*}
Now, recalling the analytic character of the function $J_{1}$ (see
beginning of Appendix) and, more precisely, its series representation
given by (\ref{eq:BessJ_series}), it is clear that every derivative
of $J_{1}$ of even order is also analytic and vanishes at zero. This
implies analyticity of the function $J_{1}^{\left(2j+2p+2\right)}\left(\rho\right)/\rho$
and, consequently, a bound on its every derivative at the origin.
Therefore, from (\ref{eq:S_pj_der}), we deduce that 
\[
\left|\left.\partial_{k_{1}}^{2n+1}\mathcal{S}_{p,j}\right|_{k_{1}=0^{+}}\right|\leq CA^{2\left(n-p\right)-1},
\]
 for some constant $C>0$, and hence (\ref{eq:S_pj_estim}) follows
due to the fact that $p\geq2$.

\medskip

\subparagraph*{\uline{\label{par:tang_estim_step3} Step 3: Asymptotic estimates
for the net moment components}}

\medskip

Recalling that $d_{1}=\pi m_{1}$ (according to (\ref{eq:d1_def}))
and using (\ref{eq:Rn_sin_estim}), we obtain from (\ref{eq:d1_ident})
\begin{equation}
m_{1}=2\iint_{D_{A}}x_{1}B_{3}\left(\mathbf{x},h\right)\dd^{2}x+2\pi\left(\frac{c_{1,0}}{A}+\frac{3c_{3,0}+c_{1,2}}{12A^{3}}\right)+\mathcal{O}\left(\frac{1}{A^{5}}\right).\label{eq:m1_estim_base}
\end{equation}
Similarly, using (\ref{eq:Rn_sin_estim}), we rewrite (\ref{eq:d3_ident})--(\ref{eq:d11_ident}),
respectively, as
\begin{align}
-\frac{1}{6A^{2}}\iint_{D_{A}}x_{1}^{3}B_{3}\left(\mathbf{x},h\right)\dd^{2}x+\frac{\pi}{8}\left(\frac{c_{1,0}}{A}-\frac{5c_{3,0}+c_{1,2}}{6A^{3}}\right) & =\frac{d_{3}}{\left(2\pi\right)^{3}A^{2}}+\mathcal{O}\left(\frac{1}{A^{5}}\right)\label{eq:d3_estim}\\
 & =\mathcal{O}\left(\frac{1}{A^{2}}\right),\nonumber 
\end{align}
\begin{align}
\frac{1}{120A^{4}}\iint_{D_{A}}x_{1}^{5}B_{3}\left(\mathbf{x},h\right)\dd^{2}x-\frac{\pi}{192}\left(\frac{c_{1,0}}{3A}+\frac{7c_{3,0}+c_{1,2}}{8A^{3}}\right) & =\frac{d_{5}}{\left(2\pi\right)^{5}A^{4}}+\mathcal{O}\left(\frac{1}{A^{5}}\right)\label{eq:d5_estim}\\
 & =\mathcal{O}\left(\frac{1}{A^{4}}\right),\nonumber 
\end{align}
\begin{align}
-\frac{1}{5040A^{6}}\iint_{D_{A}}x_{1}^{7}B_{3}\left(\mathbf{x},h\right)\dd^{2}x+\frac{\pi}{46080}\left(\frac{c_{1,0}}{A}+\frac{9c_{3,0}+c_{1,2}}{6A^{3}}\right) & =\frac{d_{7}}{\left(2\pi\right)^{7}A^{6}}+\mathcal{O}\left(\frac{1}{A^{5}}\right)\label{eq:d7_estim}\\
 & =\mathcal{O}\left(\frac{1}{A^{5}}\right),\nonumber 
\end{align}
\begin{align}
\frac{1}{362880A^{8}}\iint_{D_{A}}x_{1}^{9}B_{3}\left(\mathbf{x},h\right)\dd^{2}x-\frac{\pi}{737280}\left(\frac{c_{1,0}}{7A}+\frac{11c_{3,0}+c_{1,2}}{60A^{3}}\right) & =\frac{d_{9}}{\left(2\pi\right)^{9}A^{8}}+\mathcal{O}\left(\frac{1}{A^{5}}\right)\label{eq:d9_estim}\\
 & =\mathcal{O}\left(\frac{1}{A^{5}}\right),\nonumber 
\end{align}
\begin{align}
-\frac{1}{39916800A^{10}}\iint_{D_{A}}x_{1}^{11}B_{3}\left(\mathbf{x},h\right)\dd^{2}x+\frac{\pi}{8670412800}\left(\frac{98c_{1,0}}{9A}+\frac{13c_{3,0}+c_{1,2}}{A^{3}}\right) & =\frac{d_{11}}{\left(2\pi\right)^{11}A^{10}}+\mathcal{O}\left(\frac{1}{A^{5}}\right)\label{eq:d11_estim}\\
 & =\mathcal{O}\left(\frac{1}{A^{5}}\right).\nonumber 
\end{align}

While the first-order estimate for $m_{1}$ given in (\ref{eq:fin_estim_m1m2_ord1})
follows immediately from rigorously justified (\ref{eq:m1_estim_base}),
the higher-order estimates require more work. Namely, we wish to combine
(\ref{eq:d3_estim})--(\ref{eq:d11_estim}) in order to eliminate
in (\ref{eq:m1_estim_base}) the terms with 
\begin{equation}
\widetilde{c}_{1,0}:=\frac{c_{1,0}}{A},\hspace{1em}\hspace{1em}\widetilde{c}_{3,0}:=\frac{c_{3,0}}{A^{3}},\hspace{1em}\hspace{1em}\widetilde{c}_{1,2}:=\frac{c_{1,2}}{A^{3}},\label{eq:a1455_coeffs_tild}
\end{equation}
and, at the same time, would not commit a larger error (in order of
$A$) than that of the eliminated term.

Expressing $c_{1,0}/A$ in terms of $\mathcal{O}\left(1/A^{2}\right)$
quantities from (\ref{eq:d3_estim}) and inserting it into (\ref{eq:m1_estim_base}),
we deduce the second-order estimate for $m_{1}$ given by (\ref{eq:fin_estim_m1m2_ord2}).
We note, however, that, for third or higher order estimates, identity
(\ref{eq:d3_estim}) is not useful due to the fact that its right-hand
side has an unknown quantity $d_{3}$ appearing of order $\mathcal{O}\left(1/A^{2}\right)$
which will block any further effort to increase the accuracy of estimates. 

Derivation of the third-order estimate given by (\ref{eq:fin_estim_m1m2_ord3})
is analogous to the previous one with the only difference that $c_{1,0}/A$
is expressed (now in terms of $\mathcal{O}\left(1/A^{3}\right)$ quantities)
from (\ref{eq:d5_estim}) rather than from (\ref{eq:d3_estim}).

To proceed with derivation of estimates (\ref{eq:fin_estim_m1m2_ord4})
and (\ref{eq:fin_estim_m1m2_ord5}), it is convenient first to rewrite
(\ref{eq:d5_estim})--(\ref{eq:d11_estim}), respectively, as
\begin{equation}
7\widetilde{c}_{3,0}+\widetilde{c}_{1,2}=\frac{64}{5\pi A^{4}}\iint_{D_{A}}x_{1}^{5}B_{3}\left(\mathbf{x},h\right)\dd^{2}x-\frac{8}{3}\widetilde{c}_{1,0}+\mathcal{O}\left(\frac{1}{A^{4}}\right)=:\mathcal{T}_{5},\label{eq:T5_def}
\end{equation}
\begin{equation}
9\widetilde{c}_{3,0}+\widetilde{c}_{1,2}=\frac{384}{7\pi A^{6}}\iint_{D_{A}}x_{1}^{7}B_{3}\left(\mathbf{x},h\right)\dd^{2}x-6\widetilde{c}_{1,0}+\mathcal{O}\left(\frac{1}{A^{5}}\right)=:\mathcal{T}_{7},\label{eq:T7_def}
\end{equation}
\begin{equation}
11\widetilde{c}_{3,0}+\widetilde{c}_{1,2}=\frac{2560}{21\pi A^{8}}\iint_{D_{A}}x_{1}^{9}B_{3}\left(\mathbf{x},h\right)\dd^{2}x-\frac{60}{7}\widetilde{c}_{1,0}+\mathcal{O}\left(\frac{1}{A^{5}}\right)=:\mathcal{T}_{9},\label{eq:T9_def}
\end{equation}
\begin{equation}
13\widetilde{c}_{3,0}+\widetilde{c}_{1,2}=\frac{64512}{297\pi A^{10}}\iint_{D_{A}}x_{1}^{11}B_{3}\left(\mathbf{x},h\right)\dd^{2}x-\frac{98}{9}\widetilde{c}_{1,0}+\mathcal{O}\left(\frac{1}{A^{5}}\right)=:\mathcal{T}_{11}.\label{eq:T11_def}
\end{equation}

It is easy to see that
\begin{equation}
\frac{1}{2}\left(\mathcal{T}_{5}+\mathcal{T}_{9}\right)=\mathcal{T}_{7},\hspace{1em}\hspace{1em}\frac{1}{2}\left(\mathcal{T}_{7}+\mathcal{T}_{11}\right)=\mathcal{T}_{9}.\label{eq:T_lin_dep1}
\end{equation}

To obtain the fourth-order estimate for $m_{1}$ given in (\ref{eq:fin_estim_m1m2_ord4}),
we shall use (\ref{eq:T5_def})--(\ref{eq:T9_def}).\\
We start by using the first equation of (\ref{eq:T_lin_dep1}) (together
with definitions (\ref{eq:T5_def})--(\ref{eq:T9_def})) to express
$\widetilde{c}_{1,0}=c_{1,0}/A$ up
to order $\mathcal{O}\left(1/A^{4}\right)$, namely,
\begin{equation}
\frac{c_{1,0}}{A}=\widetilde{c}_{1,0}=-\frac{4}{\pi}\iint_{D_{A}}\left[\frac{21}{5}\left(\frac{x_{1}}{A}\right)^{4}-36\left(\frac{x_{1}}{A}\right)^{6}+40\left(\frac{x_{1}}{A}\right)^{8}\right]x_{1}B_{3}\left(\mathbf{x},h\right)\dd^{2}x+\mathcal{O}\left(\frac{1}{A^{4}}\right).\label{eq:a1_tild_ord4}
\end{equation}
Second, we observe that the quantity $3c_{3,0}+c_{1,2}$
appearing in (\ref{eq:m1_estim_base}) is related to (\ref{eq:T5_def})--(\ref{eq:T9_def})
as follows:
\begin{align}
\frac{3c_{3,0}+c_{1,2}}{A^{3}}= & 3\widetilde{c}_{3,0}+\widetilde{c}_{1,2}=4\left(\mathcal{T}_{5}-\mathcal{T}_{7}\right)+\mathcal{T}_{9}\label{eq:a455_tild_ord4}\\
= & \frac{256}{\pi}\iint_{D_{A}}\left[\frac{1}{5}\left(\frac{x_{1}}{A}\right)^{4}-\frac{6}{7}\left(\frac{x_{1}}{A}\right)^{6}+\frac{10}{21}\left(\frac{x_{1}}{A}\right)^{8}\right]x_{1}B_{3}\left(\mathbf{x},h\right)\dd^{2}x\nonumber \\
 & +\frac{100}{21}\widetilde{c}_{1,0}+\mathcal{O}\left(\frac{1}{A^{4}}\right).\nonumber 
\end{align}
Finally, substitution of (\ref{eq:a455_tild_ord4}) in (\ref{eq:m1_estim_base})
followed by the use of (\ref{eq:a1_tild_ord4}) furnishes (\ref{eq:fin_estim_m1m2_ord4}). 

To arrive at the fifth-order estimate given by (\ref{eq:fin_estim_m1m2_ord5}),
we shall use identities (\ref{eq:T7_def})--(\ref{eq:T11_def}).
The second equation of (\ref{eq:T_lin_dep1}) gives
\begin{equation}
\frac{c_{1,0}}{A}=\widetilde{c}_{1,0}=\frac{24}{\pi}\iint_{D_{A}}\left[-9\left(\frac{x_{1}}{A}\right)^{6}+40\left(\frac{x_{1}}{A}\right)^{8}-\frac{392}{11}\left(\frac{x_{1}}{A}\right)^{10}\right]x_{1}B_{3}\left(\mathbf{x},h\right)\dd^{2}x+\mathcal{O}\left(\frac{1}{A^{5}}\right).\label{eq:a1_tild_ord5}
\end{equation}
We also have an analog of (\ref{eq:a1455_coeffs_tild}), namely, 
\begin{align}
\frac{3c_{3,0}+c_{1,2}}{A^{3}}= & 3\widetilde{c}_{3,0}+\widetilde{c}_{1,2}=5\left(\mathcal{T}_{7}-\mathcal{T}_{9}\right)+\mathcal{T}_{11}\label{eq:a455_tild_ord5}\\
= & \frac{128}{\pi}\iint_{D_{A}}\left[\frac{15}{7}\left(\frac{x_{1}}{A}\right)^{6}-\frac{100}{21}\left(\frac{x_{1}}{A}\right)^{8}+\frac{56}{33}\left(\frac{x_{1}}{A}\right)^{10}\right]x_{1}B_{3}\left(\mathbf{x},h\right)\dd^{2}x\nonumber \\
 & +\frac{124}{63}\widetilde{c}_{1,0}+\mathcal{O}\left(\frac{1}{A^{5}}\right).\nonumber 
\end{align}
Therefore, inserting of (\ref{eq:a455_tild_ord5}) in (\ref{eq:m1_estim_base})
followed by the use of (\ref{eq:a1_tild_ord5}) gives (\ref{eq:fin_estim_m1m2_ord5}). 

\subsubsection{\label{subsec:norm_estim} Normal component of the net moment}

Similarly to the case of tangential net moment components, we expand
the integrand in (\ref{eq:B3_hat_k1}) in power series in a positive
neighborhood of $k_{1}=0$, but, in contrast to that previous situation,
the attention will now be on the real part. This yields
\begin{equation}
\left.\left[\text{Re }\hat{B}_{3}\left(k_{1},0,h\right)\right]\right|_{k_{1}=0^{+}}=0,\label{eq:d0_def}
\end{equation}
\begin{align}
\frac{1}{2}\left.\partial_{k_{1}}^{2}\left[\text{Re }\hat{B}_{3}\left(k_{1},0,h\right)\right]\right|_{k_{1}=0^{+}} & =-2\pi^{2}\left[\left\langle x_{1}M_{1}\right\rangle +\left\langle \left(h-x_{3}\right)M_{3}\right\rangle \right]\label{eq:d2_def}\\
 & =:d_{2},\nonumber 
\end{align}
\begin{align}
\frac{1}{24}\left.\partial_{k_{1}}^{4}\left[\text{Re }\hat{B}_{3}\left(k_{1},0,h\right)\right]\right|_{k_{1}=0^{+}} & =-\frac{4\pi^{4}}{3}\left[3\left\langle \left(h-x_{3}\right)^{2}x_{1}M_{1}\right\rangle -\left\langle x_{1}^{3}M_{1}\right\rangle +\left\langle \left(h-x_{3}\right)^{3}M_{3}\right\rangle -3\left\langle \left(h-x_{3}\right)x_{1}^{2}M_{3}\right\rangle \right]\label{eq:d4_def}\\
 & =:d_{4},\nonumber 
\end{align}
\begin{align}
\frac{1}{6!}\left.\partial_{k_{1}}^{6}\left[\text{Re }\hat{B}_{3}\left(k_{1},0,h\right)\right]\right|_{k_{1}=0^{+}}= & -\frac{4\pi^{6}}{15}\left[5\left\langle \left(h-x_{3}\right)^{4}x_{1}M_{1}\right\rangle -10\left\langle \left(h-x_{3}\right)^{2}x_{1}^{3}M_{1}\right\rangle +\left\langle x_{1}^{5}M_{1}\right\rangle \right.\label{eq:d6_def}\\
 & \left.+\left\langle \left(h-x_{3}\right)^{5}M_{3}\right\rangle -10\left\langle \left(h-x_{3}\right)^{3}x_{1}^{2}M_{3}\right\rangle +5\left\langle \left(h-x_{3}\right)x_{1}^{4}M_{3}\right\rangle \right]\nonumber \\
=: & d_{6},\nonumber 
\end{align}
\begin{align}
\frac{1}{8!}\left.\partial_{k_{1}}^{8}\left[\text{Re }\hat{B}_{3}\left(k_{1},0,h\right)\right]\right|_{k_{1}=0^{+}}= & -\frac{8\pi^{8}}{315}\left[7\left\langle \left(h-x_{3}\right)^{6}x_{1}M_{1}\right\rangle -35\left\langle \left(h-x_{3}\right)^{4}x_{1}^{3}M_{1}\right\rangle +21\left\langle \left(h-x_{3}\right)^{2}x_{1}^{5}M_{1}\right\rangle \right.\label{eq:d8_def}\\
 & -\left\langle x_{1}^{7}M_{1}\right\rangle +\left\langle \left(h-x_{3}\right)^{7}M_{3}\right\rangle -21\left\langle \left(h-x_{3}\right)^{5}x_{1}^{2}M_{3}\right\rangle +35\left\langle \left(h-x_{3}\right)^{3}x_{1}^{4}M_{3}\right\rangle \nonumber \\
 & \left.-7\left\langle \left(h-x_{3}\right)x_{1}^{6}M_{3}\right\rangle \right]\nonumber \\
=: & d_{8},\nonumber 
\end{align}
\begin{align}
\frac{1}{10!}\left.\partial_{k_{1}}^{10}\left[\text{Re }\hat{B}_{3}\left(k_{1},0,h\right)\right]\right|_{k_{1}=0^{+}}= & -\frac{4\pi^{10}}{2385}\left[9\left\langle \left(h-x_{3}\right)^{8}x_{1}M_{1}\right\rangle -84\left\langle \left(h-x_{3}\right)^{6}x_{1}^{3}M_{1}\right\rangle +126\left\langle \left(h-x_{3}\right)^{4}x_{1}^{5}M_{1}\right\rangle \right.\label{eq:d10_def}\\
 & -36\left\langle \left(h-x_{3}\right)^{2}x_{1}^{7}M_{1}\right\rangle +\left\langle x_{1}^{9}M_{1}\right\rangle +\left\langle \left(h-x_{3}\right)^{9}M_{3}\right\rangle -36\left\langle \left(h-x_{3}\right)^{7}x_{1}^{2}M_{3}\right\rangle \nonumber \\
 & \left.+126\left\langle \left(h-x_{3}\right)^{5}x_{1}^{4}M_{3}\right\rangle -84\left\langle \left(h-x_{3}\right)^{3}x_{1}^{6}M_{3}\right\rangle +9\left\langle \left(h-x_{3}\right)x_{1}^{8}M_{3}\right\rangle \right]\nonumber \\
=: & d_{10},\nonumber 
\end{align}

On the other hand, (\ref{eq:B3_decomp}) implies that, for $n\in\mathbb{N}_{0}$,
\begin{align}
\left.\partial_{k_{1}}^{2n}\left[\text{Re }\hat{B}_{3}\left(k_{1},0,h\right)\right]\right|_{k_{1}=0^{+}}= & \left(-1\right)^{n}\left(2\pi\right)^{2n}\iint_{D_{A}}x_{1}^{2n}B_{3}\left(\mathbf{x},h\right)\dd^{2}x\label{eq:ReB3_int}\\
 & +\left.\partial_{k_{1}}^{2n}\left(\iint_{\mathbb{R}^{2}\backslash D_{A}}\cos\left(2\pi k_{1}x_{1}\right)B_{3}^{\text{asympt}}\left(\mathbf{x},h\right)\dd^{2}x\right)\right|_{k_{1}=0^{+}}\nonumber \\
 & +\left.\partial_{k_{1}}^{2n}\left(\iint_{\mathbb{R}^{2}\backslash D_{A}}\cos\left(2\pi k_{1}x_{1}\right)\left[B_{3}\left(\mathbf{x},h\right)-B_{3}^{\text{asympt}}\left(\mathbf{x},h\right)\right]\dd^{2}x\right)\right|_{k_{1}=0^{+}}.\nonumber 
\end{align}

As before, we continue in 3 steps. First, we evaluate explicitly 
\begin{equation}
\mathcal{I}^{\cos}\equiv\mathcal{I}^{\cos}\left(k_{1},A\right):=\iint_{\mathbb{R}^{2}\backslash D_{A}}\cos\left(2\pi k_{1}x_{1}\right)B_{3}^{\text{asympt}}\left(\mathbf{x},h\right)\dd^{2}x,\label{eq:I_cos}
\end{equation}
that allow us to obtain from (\ref{eq:ReB3_int}), a set of useful
identities involving remainder terms. Second, we estimate the remainder
terms, namely, 
\begin{equation}
\mathcal{R}_{2n}^{\cos}\equiv\mathcal{R}_{2n}^{\cos}\left(A\right):=\left.\partial_{k_{1}}^{2n}\left(\iint_{\mathbb{R}^{2}\backslash D_{A}}\cos\left(2\pi k_{1}x_{1}\right)\left[B_{3}\left(\mathbf{x},h\right)-B_{3}^{\text{asympt}}\left(\mathbf{x},h\right)\right]\dd^{2}x\right)\right|_{k_{1}=0^{+}},\hspace{1em}n\in\mathbb{N}_{0},\label{eq:Rn_cos}
\end{equation}
for large $A$. At last, we rigorously derive asymptotic formulas
(\ref{eq:fin_estim_m3_ord2}), (\ref{eq:fin_estim_m3_ord3}), (\ref{eq:fin_estim_m3_ord4})
from the obtained set of identities.

\medskip

\subparagraph*{\uline{\label{par:norm_estim_step1} Step 1: Derivation of the
set of identities}}

\medskip

Using (\ref{eq:B3_asympt_def}) and passing to the polar coordinates
using $x=r\cos\theta$, $y=r\sin\theta$, $\dd^{2}x=r\dd r\dd\theta$, we
obtain from (\ref{eq:I_cos})
\begin{equation}
\mathcal{I}^{\cos}=c_{0,0}\mathcal{I}_{1}^{\cos}+c_{2,0}\mathcal{I}_{2}^{\cos}+c_{0,2}\mathcal{I}_{3}^{\cos},\label{eq:I_cos_decomp}
\end{equation}
where 
\begin{equation}
\mathcal{I}_{1}^{\cos}:=\int_{A}^{\infty}\int_{0}^{2\pi}\cos\left(2\pi k_{1}r\cos\theta\right)d\theta\frac{\dd r}{r^{2}},\label{eq:I_cos1_def}
\end{equation}
\begin{equation}
\mathcal{I}_{2}^{\cos}:=\int_{A}^{\infty}\int_{0}^{2\pi}\cos\left(2\pi k_{1}r\cos\theta\right)\cos^{2}\theta d\theta\frac{\dd r}{r^{4}},\label{eq:I_cos3_def}
\end{equation}
\begin{equation}
\mathcal{I}_{3}^{\cos}:=\int_{A}^{\infty}\int_{0}^{2\pi}\cos\left(2\pi k_{1}r\cos\theta\right)\sin^{2}\theta d\theta\frac{\dd r}{r^{4}}.\label{eq:I_cos4_def}
\end{equation}
Here, we used results (\ref{eq:trig_int_cos1})--(\ref{eq:trig_int_cos2})
of Lemma \ref{lem:trig_ints} multiple times to deduce vanishing of
the integrals associated with the terms which involve $c_{1,0}$,
$c_{0,1}$, $c_{1,1}$, $c_{3,0}$, $c_{0,3}$, $c_{2,1}$, $c_{1,2}$.

Using the integral representation of Bessel functions (due to (\ref{eq:BessJ0_int_repr})),
integration by parts, the asymptotics of $J_{1}$ given by (\ref{eq:BessJ_asympt}),
and the relation $J_{0}^{\prime}\left(x\right)=-J_{1}\left(x\right)$
(see (\ref{eq:J0der_J1_conn})), we can write, for $k_{1}>0$,
\begin{align}
\mathcal{I}_{1}^{\cos} & =2\pi\int_{A}^{\infty}J_{0}\left(2\pi k_{1}r\right)\frac{\dd r}{r^{2}}=\left(2\pi\right)^{2}k_{1}\int_{2\pi k_{1}A}^{\infty}\frac{J_{0}\left(x\right)}{x^{2}}\dd x\label{eq:I_cos1}\\
 & =\left(2\pi\right)^{2}k_{1}\left.\left(\frac{J_{0}\left(\rho\right)}{\rho}-\int_{\rho}^{\infty}\frac{J_{1}\left(x\right)}{x}\dd x\right)\right|_{\rho=2\pi k_{1}A},\nonumber 
\end{align}
\begin{align}
\mathcal{I}_{2}^{\cos} & =-2\pi\int_{A}^{\infty}J_{0}^{\prime\prime}\left(2\pi k_{1}r\right)\frac{\dd r}{r^{4}}=\left(2\pi\right)^{4}k_{1}^{3}\int_{2\pi k_{1}A}^{\infty}\frac{J_{1}^{\prime}\left(x\right)}{x^{4}}\dd x\label{eq:I_cos3}\\
 & =\left(2\pi\right)^{4}k_{1}^{3}\left.\left(\frac{J_{1}\left(\rho\right)}{\rho^{4}}+4\int_{\rho}^{\infty}\frac{J_{1}\left(x\right)}{x^{5}}\dd x\right)\right|_{\rho=2\pi k_{1}A},\nonumber 
\end{align}
\begin{align}
\mathcal{I}_{3}^{\cos} & =2\pi\int_{A}^{\infty}\left[J_{0}\left(2\pi k_{1}r\right)+J_{0}^{\prime\prime}\left(2\pi k_{1}r\right)\right]\frac{\dd r}{r^{4}}=\left(2\pi\right)^{4}k_{1}^{3}\int_{2\pi k_{1}A}^{\infty}\frac{J_{0}\left(x\right)+J_{0}^{\prime\prime}\left(x\right)}{x^{4}}\dd x,\label{eq:I_cos4}\\
 & =\left(2\pi\right)^{4}k_{1}^{3}\left.\left(\frac{J_{0}\left(\rho\right)}{3\rho^{3}}-\frac{J_{1}\left(\rho\right)}{\rho^{4}}-\frac{1}{3}\int_{\rho}^{\infty}\frac{J_{1}\left(x\right)}{x^{3}}\dd x-4\int_{\rho}^{\infty}\frac{J_{1}\left(x\right)}{x^{5}}\dd x\right)\right|_{\rho=2\pi k_{1}A}.\nonumber 
\end{align}
Inserting here the results of Lemmas \ref{lem:int_J1_x3}--\ref{lem:int_J1_var},
namely, (\ref{eq:int_J1_x3}), (\ref{eq:int_J1_x})--(\ref{eq:int_J1_x5}),
we arrive at
\begin{align}
\mathcal{I}_{1}^{\cos} & =\frac{2\pi}{A}\left[2J_{0}\left(\rho\right)+\frac{J_{1}\left(\rho\right)}{\rho}-J_{1}^{\prime}\left(\rho\right)-2J_{1}\left(\rho\right)-\rho-\rho J_{1}\left(\rho\right)\right.\label{eq:I1_cos_fin}\\
 & \left.\left.+\rho^{2}J_{0}\left(\rho\right)-\frac{\pi\rho^{2}}{2}J_{0}\left(\rho\right)H_{1}\left(\rho\right)+\frac{\pi\rho^{2}}{2}J_{1}\left(\rho\right)H_{0}\left(\rho\right)\right]\right|_{\rho=2\pi k_{1}A},\nonumber 
\end{align}
\begin{align}
\mathcal{I}_{2}^{\cos} & =\frac{2\pi}{15A^{3}}\left[-\frac{J_{0}\left(\rho\right)}{3}+4J_{1}^{\prime}\left(\rho\right)-\frac{4\rho^{2}J_{0}\left(\rho\right)}{3}-\frac{4\rho J_{1}\left(\rho\right)}{3}+\frac{4\rho^{3}}{3}+\frac{4\rho^{3}J_{1}\left(\rho\right)}{3}\right.\label{eq:I3_cos_fin}\\
 & \left.\left.-\frac{4\rho^{4}J_{0}\left(\rho\right)}{3}+\frac{2\pi\rho^{4}}{3}J_{0}\left(\rho\right)H_{1}\left(\rho\right)-\frac{2\pi\rho^{4}}{3}J_{1}\left(\rho\right)H_{0}\left(\rho\right)\right]\right|_{\rho=2\pi k_{1}A},\nonumber 
\end{align}
\begin{align}
\mathcal{I}_{3}^{\cos} & =\frac{2\pi}{3A^{3}}\left[J_{0}\left(\rho\right)+\frac{J_{1}\left(\rho\right)}{15\rho}-\frac{4J_{1}^{\prime}\left(\rho\right)}{5}-\frac{\rho^{2}J_{0}\left(\rho\right)}{15}-\frac{\rho J_{1}\left(\rho\right)}{15}+\frac{\rho^{3}}{15}+\frac{\rho^{3}J_{1}\left(\rho\right)}{15}\right.\label{eq:I4_cos_fin}\\
 & \left.\left.-\frac{\rho^{4}J_{0}\left(\rho\right)}{15}+\frac{\pi\rho^{4}}{30}J_{0}\left(\rho\right)H_{1}\left(\rho\right)-\frac{\pi\rho^{4}}{30}J_{1}\left(\rho\right)H_{0}\left(\rho\right)\right]\right|_{\rho=2\pi k_{1}A}.\nonumber 
\end{align}
Substitution of (\ref{eq:I1_cos_fin})--(\ref{eq:I4_cos_fin}) into
(\ref{eq:I_cos_decomp}), we employ (\ref{eq:BessJ_series}) and (\ref{eq:StruvH})
to compute
\begin{equation}
\left.\mathcal{I}^{\cos}\right|_{k_{1}=0^{+}}=2\pi\left(\frac{c_{0,0}}{A}+\frac{11c_{2,0}+19c_{0,2}}{90A^{3}}\right),\label{eq:I_cos_d0}
\end{equation}
\begin{equation}
\left.\partial_{k_{1}}^{2}\mathcal{I}^{\cos}\right|_{k_{1}=0^{+}}=\frac{1}{4}\left(2\pi\right)^{3}A^{2}\left(\frac{c_{0,0}}{A}-\frac{131c_{2,0}+49c_{0,2}}{180A^{3}}\right),\label{eq:I_cos_d2}
\end{equation}
\begin{equation}
\left.\partial_{k_{1}}^{4}\mathcal{I}^{\cos}\right|_{k_{1}=0^{+}}=-\frac{1}{8}\left(2\pi\right)^{5}A^{4}\left(\frac{c_{0,0}}{A}+\frac{229c_{2,0}+41c_{0,2}}{90A^{3}}\right),\label{eq:I_cos_d4}
\end{equation}
\begin{equation}
\left.\partial_{k_{1}}^{6}\mathcal{I}^{\cos}\right|_{k_{1}=0^{+}}=\frac{1}{16}\left(2\pi\right)^{7}A^{6}\left(\frac{c_{0,0}}{A}+\frac{109c_{2,0}+11c_{0,2}}{72A^{3}}\right),\label{eq:I_cos_d6}
\end{equation}
\begin{equation}
\left.\partial_{k_{1}}^{8}\mathcal{I}^{\cos}\right|_{k_{1}=0^{+}}=-3136\left(2\pi\right)^{9}A^{8}\left(\frac{c_{0,0}}{7A}+\frac{17c_{2,0}+c_{0,2}}{90A^{3}}\right),\label{eq:I_cos_d8}
\end{equation}
\begin{equation}
\left.\partial_{k_{1}}^{10}\mathcal{I}^{\cos}\right|_{k_{1}=0^{+}}=\frac{7}{256}\left(2\pi\right)^{11}A^{10}\left(\frac{c_{0,0}}{A}+\frac{523c_{2,0}+17c_{0,2}}{420A^{3}}\right).\label{eq:I_cos_d10}
\end{equation}
Taking into account (\ref{eq:I_cos}), (\ref{eq:Rn_cos}), we use
(\ref{eq:d0_def})--(\ref{eq:d10_def}) and (\ref{eq:I_cos_d0})--(\ref{eq:I_cos_d10})
in (\ref{eq:ReB3_int}) with $n=0,\ldots,5$, respectively, and thus
arrive at the following set of identities:
\begin{equation}
\iint_{D_{A}}B_{3}\left(\mathbf{x},h\right)\dd^{2}x+2\pi\left(\frac{c_{0,0}}{A}+\frac{11c_{2,0}+19c_{0,2}}{90A^{3}}\right)+\mathcal{R}_{0}^{\cos}=0,\label{eq:d0_ident}
\end{equation}
\begin{equation}
-\frac{\left(2\pi\right)^{2}}{2}\iint_{D_{A}}x_{1}^{2}B_{3}\left(\mathbf{x},h\right)\dd^{2}x+\frac{\left(2\pi\right)^{3}}{4}A^{2}\left(\frac{c_{0,0}}{A}-\frac{131c_{2,0}+49c_{0,2}}{180A^{3}}\right)+\frac{1}{2}\mathcal{R}_{2}^{\cos}=d_{2},\label{eq:d2_ident}
\end{equation}
\begin{equation}
\frac{\left(2\pi\right)^{4}}{24}\iint_{D_{A}}x_{1}^{4}B_{3}\left(\mathbf{x},h\right)\dd^{2}x-\frac{1}{192}\left(2\pi\right)^{5}A^{4}\left(\frac{c_{0,0}}{A}+\frac{229c_{2,0}+41c_{0,2}}{90A^{3}}\right)+\frac{1}{4!}\mathcal{R}_{4}^{\cos}=d_{4},\label{eq:d4_ident}
\end{equation}
\begin{equation}
-\frac{\left(2\pi\right)^{6}}{720}\iint_{D_{A}}x_{1}^{6}B_{3}\left(\mathbf{x},h\right)\dd^{2}x+\frac{1}{11520}\left(2\pi\right)^{7}A^{6}\left(\frac{c_{0,0}}{A}+\frac{109c_{2,0}+11c_{0,2}}{72A^{3}}\right)+\frac{1}{6!}\mathcal{R}_{6}^{\cos}=d_{6},\label{eq:d6_ident}
\end{equation}
\begin{equation}
\frac{\left(2\pi\right)^{8}}{40320}\iint_{D_{A}}x_{1}^{8}B_{3}\left(\mathbf{x},h\right)\dd^{2}x-\frac{7}{90}\left(2\pi\right)^{9}A^{8}\left(\frac{c_{0,0}}{7A}+\frac{17c_{2,0}+c_{0,2}}{90A^{3}}\right)+\frac{1}{8!}\mathcal{R}_{8}^{\cos}=d_{8},\label{eq:d8_ident}
\end{equation}
\begin{equation}
-\frac{\left(2\pi\right)^{10}}{3628800}\iint_{D_{A}}x_{1}^{10}B_{3}\left(\mathbf{x},h\right)\dd^{2}x+\frac{1}{132710400}\left(2\pi\right)^{11}A^{10}\left(\frac{c_{0,0}}{A}+\frac{523c_{2,0}+17c_{0,2}}{420A^{3}}\right)+\frac{1}{10!}\mathcal{R}_{10}^{\cos}=d_{10}.\label{eq:d10_ident}
\end{equation}

\medskip

\subparagraph*{\uline{\label{par:norm_estim_step2} Step 2: Analysis of the remainder
terms \mbox{$\mathcal{R}_{2n}^{\cos}$} for \mbox{$0\protect\leq n\protect\leq5$}}}

\medskip

We shall show that
\begin{equation}
\mathcal{R}_{2n}^{\cos}=\mathcal{O}\left(\frac{1}{A^{5-2n}}\right),\hspace{1em}\hspace{1em}0\leq n\leq5.\label{eq:Rn_cos_estim}
\end{equation}
The reasoning will we identical to that of Step 2 of Subsection \ref{subsec:tang_estim},
therefore, we omit repetition of some details.

Using previously introduced notation $L_{N}$ (see (\ref{eq:LN_def})),
we can write

\begin{align}
\mathcal{R}_{2n}^{\cos}= & \left.\partial_{k_{1}}^{2n}\left(\iint_{\mathbb{R}^{2}\backslash D_{A}}\cos\left(2\pi k_{1}x_{1}\right)L_{N}\left(\mathbf{x}\right)\dd^{2}x\right)\right|_{k_{1}=0^{+}}\label{eq:Rn_cos_decomp}\\
 & +\left.\partial_{k_{1}}^{2n}\left(\iint_{\mathbb{R}^{2}\backslash D_{A}}\cos\left(2\pi k_{1}x_{1}\right)\left[B_{3}\left(\mathbf{x},h\right)-B_{3}^{\text{asympt}}\left(\mathbf{x},h\right)-L_{N}\left(\mathbf{x}\right)\right]\dd^{2}x\right)\right|_{k_{1}=0^{+}}.\nonumber 
\end{align}

Assuming $N>9$, we estimate
\[
\left|\partial_{k_{1}}^{2n}\left(\iint_{\mathbb{R}^{2}\backslash D_{A}}\cos\left(2\pi k_{1}x_{1}\right)\left[B_{3}\left(\mathbf{x},h\right)-B_{3}^{\text{asympt}}\left(\mathbf{x},h\right)-L_{N}\left(\mathbf{x}\right)\right]\dd^{2}x\right)\Bigr|_{k_{1}=0^{+}}\right|
\]
\begin{align*}
= & \left|\left(\iint_{\mathbb{R}^{2}\backslash D_{A}}\frac{\left(2\pi x_{1}\right)^{2n}}{\left|\mathbf{x}\right|^{N+4}}\left|\mathbf{x}\right|^{N+4}\left[B_{3}\left(\mathbf{x},h\right)-B_{3}^{\text{asympt}}\left(\mathbf{x},h\right)-L_{N}\left(\mathbf{x}\right)\right]\dd^{2}x\right)\right|\\
\leq & \left(2\pi\right)^{2n+1}\widetilde{C}_{N}\int_{A}^{\infty}\frac{dr}{r^{N-2n+2}}=\frac{\left(2\pi\right)^{2\left(n+1\right)}\widetilde{C}_{N}}{\left(N-2n+1\right)}\frac{1}{A^{N-2n+1}},
\end{align*}
for some constant $\widetilde{C}_{N}>0$, and note that this complies
with (\ref{eq:Rn_cos_estim}). 

We now fix $N=10$ and proceed with estimating the term on the first
line of (\ref{eq:Rn_cos_estim}). Identities (\ref{eq:trig_int_cos1})--(\ref{eq:trig_int_cos2})
of Lemma \ref{lem:trig_ints} entail that the only non-vanishing terms
are those of the form
\begin{equation}
\int_{A}^{\infty}\int_{0}^{2\pi}\cos\left(2\pi k_{1}r\cos\theta\right)\cos^{2\left(p-l\right)}\theta\,\,\cos^{2l}\theta \dd\theta\frac{\dd r}{r^{2\left(p+1\right)}},\hspace{1em}\hspace{1em}0\leq l\leq p,\hspace{1em}p\geq2,\label{eq:cos_terms_nonvan}
\end{equation}
and hence it is sufficient to only deal with the quantities
\begin{align}
\mathcal{C}_{p,j}:= & \int_{A}^{\infty}\int_{0}^{2\pi}\cos\left(2\pi k_{1}r\cos\theta\right)\cos^{2j}\theta\,\dd\theta\frac{\dd r}{r^{2\left(p+1\right)}},\hspace{1em}\hspace{1em}0\leq j\leq p,\hspace{1em}p\geq2.\label{eq:C_pj_def}
\end{align}
As before, $p$ in (\ref{eq:cos_terms_nonvan})--(\ref{eq:C_pj_def})
is related to $q$ from (\ref{eq:LN_def}) as $q=2p+1$, $p\geq2$:
only the terms $\mathcal{O}\left(1/A^{7}\right)$, $\mathcal{O}\left(1/A^{9}\right)$,
..., in (\ref{eq:LN_def}) contribute to $\mathcal{R}_{2n}^{\cos}$.
Hence, we consider $\mathcal{C}_{p,j}$ only for $p\leq\left\lfloor \frac{N-1}{2}\right\rfloor =8$.

We are intending to show that the following estimate holds for all
$0\leq n\leq5$, $0\leq j\leq p$, $2\leq p\leq8$: 
\begin{equation}
\left.\partial_{k_{1}}^{2n}\mathcal{C}_{p,j}\right|_{k_{1}=0^{+}}=\mathcal{O}\left(\frac{1}{A^{5-2n}}\right).\label{eq:C_pj_estim}
\end{equation}

For $0\leq n\leq p$, we have
\[
\left.\partial_{k_{1}}^{2n}\mathcal{C}_{p,j}\right|_{k_{1}=0^{+}}=\left(-1\right)^{n}\left(2\pi\right)^{2n}\int_{A}^{\infty}\int_{0}^{2\pi}\cos^{2j}\theta\,\dd\theta\frac{\dd r}{r^{2\left(p-n\right)+2}},
\]
and hence
\[
\left|\left.\partial_{k_{1}}^{2n+1}\mathcal{C}_{p,j}\right|_{k_{1}=0^{+}}\right|\leq\left(2\pi\right)^{2n+1}\int_{A}^{\infty}\frac{\dd r}{r^{2\left(p-n\right)+2}}=\mathcal{O}\left(\frac{1}{A^{2\left(p-n\right)+1}}\right),
\]
where the convergence of the last integral is due to $p\geq n$. The
obtained estimate satisfies (\ref{eq:C_pj_estim}) due to $p\ge2$.

For $n\geq p+1$, we use (\ref{eq:BessJ0_int_repr}) to rewrite (\ref{eq:C_pj_def})
as
\[
\mathcal{C}_{p,j}=\left(-1\right)^{j}2\pi\int_{A}^{\infty}J_{0}^{\left(2j\right)}\left(2\pi k_{1}r\right)\frac{\dd r}{r^{2\left(p+1\right)}},\hspace{1em}\hspace{1em}0\leq j\leq p,\hspace{1em}p\geq2.
\]
We then evaluate
\begin{align}
\partial_{k_{1}}^{2n}\mathcal{C}_{p,j} & =\left(-1\right)^{j}\left(2\pi\right)^{2\left(p+1\right)}\partial_{k_{1}}^{2\left(n-p\right)-1}\int_{A}^{\infty}J_{0}^{\left(2j+2p+1\right)}\left(2\pi k_{1}r\right)\frac{\dd r}{r}\label{eq:C_pj_der}\\
 & =\left(-1\right)^{j+1}\left(2\pi\right)^{2p+3}A\partial_{k_{1}}^{2\left(n-p\right)-2}\left.\frac{J_{0}^{\left(2j+2p+1\right)}\left(\rho\right)}{\rho}\right|_{\rho=2\pi k_{1}A}\nonumber \\
 & =\left(-1\right)^{j+1}\left(2\pi\right)^{2n+1}A^{2\left(n-p\right)-1}\left.\frac{\dd^{2\left(n-p\right)-2}}{\dd\rho^{2\left(n-p\right)-2}}\left(\frac{J_{0}^{\left(2j+2p+1\right)}\left(\rho\right)}{\rho}\right)\right|_{\rho=2\pi k_{1}A},\nonumber 
\end{align}
where the following identity, for $k_{1}>0$, was used:
\begin{align*}
\partial_{k_{1}}\int_{A}^{\infty}J_{0}^{\left(2j+2p+1\right)}\left(2\pi k_{1}r\right)\frac{\dd r}{r} & =\partial_{k_{1}}\int_{2\pi k_{1}A}^{\infty}J_{0}^{\left(2j+2p+1\right)}\left(\rho\right)\frac{\dd \rho}{\rho}\\
 & =-2\pi A\left.\frac{J_{0}^{\left(2j+2p+1\right)}\left(\rho\right)}{\rho}\right|_{\rho=2\pi k_{1}A}.
\end{align*}
Due to analyticity $J_{0}$ (see beginning of Appendix), we see from
(\ref{eq:BessJ_series}) that its every derivative of odd order is
also analytic and vanishes at zero. Hence, $J_{0}^{\left(2j+2p+1\right)}\left(\rho\right)/\rho$
is analytic as well, and, in particular, has bounded derivatives at
the origin. Therefore, (\ref{eq:C_pj_der}) entails that 
\[
\left|\left.\partial_{k_{1}}^{2n}\mathcal{C}_{p,j}\right|_{k_{1}=0^{+}}\right|\leq\widetilde{C}A^{2\left(n-p\right)-1},
\]
 for some constant $\widetilde{C}>0$, and hence (\ref{eq:C_pj_estim})
follows due to the fact that $p\geq2$.

\medskip

\subparagraph*{\uline{\label{par:norm_estim_step3} Step 3: Asymptotic estimates
for the net moment component}}

\medskip

Since, according to (\ref{eq:a0_def}), we have $c_{0,0}=-\frac{m_{3}}{4\pi}$,
and hence using (\ref{eq:Rn_cos_estim}), it follows from (\ref{eq:d0_ident})
that
\begin{equation}
m_{3}=2\iint_{D_{A}}B_{3}\left(\mathbf{x},h\right)\dd^{2}x+\frac{\pi}{45A^{2}}\left(11c_{2,0}+19c_{0,2}\right)+\mathcal{O}\left(\frac{1}{A^{4}}\right).\label{eq:m3_estim_base}
\end{equation}
Employing (\ref{eq:Rn_sin_estim}), we rewrite (\ref{eq:d3_ident})--(\ref{eq:d11_ident}),
respectively, as
\begin{align}
-\frac{1}{2A^{2}}\iint_{D_{A}}x_{1}^{2}B_{3}\left(\mathbf{x},h\right)\dd^{2}x-\frac{\pi}{8}\left(\frac{m_{3}}{\pi A}+\frac{131c_{2,0}+49c_{0,2}}{45A^{3}}\right) & =\frac{d_{2}}{\left(2\pi A\right)^{2}}+\mathcal{O}\left(\frac{1}{A^{5}}\right)\label{eq:d2_estim}\\
 & =\mathcal{O}\left(\frac{1}{A^{2}}\right),\nonumber 
\end{align}
\begin{align}
\frac{1}{24A^{4}}\iint_{D_{A}}x_{1}^{4}B_{3}\left(\mathbf{x},h\right)\dd^{2}x+\frac{\pi}{192}\left(\frac{m_{3}}{2\pi A}-\frac{229c_{2,0}+41c_{0,2}}{45A^{3}}\right) & =\frac{d_{4}}{\left(2\pi A\right)^{4}}+\mathcal{O}\left(\frac{1}{A^{5}}\right)\label{eq:d4_estim}\\
 & =\mathcal{O}\left(\frac{1}{A^{4}}\right),\nonumber 
\end{align}
\begin{align}
-\frac{1}{720A^{6}}\iint_{D_{A}}x_{1}^{6}B_{3}\left(\mathbf{x},h\right)\dd^{2}x-\frac{\pi}{2}\left(\frac{m_{3}}{\pi A}-\frac{109c_{2,0}+11c_{0,2}}{18A^{3}}\right) & =\frac{d_{6}}{\left(2\pi A\right)^{6}}+\mathcal{O}\left(\frac{1}{A^{5}}\right)\label{eq:d6_estim}\\
 & =\mathcal{O}\left(\frac{1}{A^{5}}\right),\nonumber 
\end{align}
\begin{align}
\frac{1}{40320A^{8}}\iint_{D_{A}}x_{1}^{8}B_{3}\left(\mathbf{x},h\right)\dd^{2}x+\frac{\pi}{147456}\left(\frac{m_{3}}{14\pi A}-\frac{17c_{2,0}+c_{0,2}}{45A^{3}}\right) & =\frac{d_{8}}{\left(2\pi A\right)^{8}}+\mathcal{O}\left(\frac{1}{A^{5}}\right)\label{eq:d8_estim}\\
 & =\mathcal{O}\left(\frac{1}{A^{5}}\right),\nonumber 
\end{align}
\begin{align}
-\frac{1}{3628800A^{10}}\iint_{D_{A}}x_{1}^{10}B_{3}\left(\mathbf{x},h\right)\dd^{2}x-\frac{\pi}{66355200}\left(\frac{m_{3}}{\pi A}-\frac{523c_{2,0}+17c_{0,2}}{105A^{3}}\right) & =\frac{d_{10}}{\left(2\pi A\right)^{10}}+\mathcal{O}\left(\frac{1}{A^{5}}\right)\label{eq:d10_estim}\\
 & =\mathcal{O}\left(\frac{1}{A^{5}}\right).\nonumber 
\end{align}

We see that (\ref{eq:m3_estim_base}) already provides the second-order
estimate of $m_{3}$ given in (\ref{eq:fin_estim_m3_ord2}). The higher-order
estimates given by (\ref{eq:fin_estim_m3_ord3}) and (\ref{eq:fin_estim_m3_ord4})
can be obtained by combining (\ref{eq:d4_estim})--(\ref{eq:d8_estim}).
Note that elimination of $m_{3}$ using either (\ref{eq:d2_estim})
would incur a rather large error $\mathcal{O}\left(1/A\right)$, reducing
the final estimate to the first order. Similarly, any use of (\ref{eq:d10_estim})
would result in an estimate of order $\mathcal{O}\left(1/A^{4}\right)$,
but such an estimate could already be deduced using the other listed
above relations. 

For the sake of simplification, let us set 
\begin{equation}
\widetilde{c}_{2,0}:=\frac{c_{2,0}}{A^{3}},\hspace{1em}\hspace{1em}\widetilde{c}_{0,2}:=\frac{c_{0,2}}{A^{3}},\label{eq:a233_coeffs_tild}
\end{equation}
and rewrite (\ref{eq:m3_estim_base}) and (\ref{eq:d4_estim})--(\ref{eq:d8_estim}),
respectively, as 
\begin{equation}
11\widetilde{c}_{2,0}+19\widetilde{c}_{0,2}=-\frac{45}{\pi}\iint_{D_{A}}B_{3}\left(\mathbf{x},h\right)\dd^{2}x+\frac{45m_{3}}{2\pi A}+\mathcal{O}\left(\frac{1}{A^{5}}\right)=:\mathcal{T}_{0},\label{eq:T0_def}
\end{equation}
\begin{equation}
131\widetilde{c}_{2,0}+49\widetilde{c}_{0,2}=-\frac{180}{\pi A^{2}}\iint_{D_{A}}x_{1}^{2}B_{3}\left(\mathbf{x},h\right)\dd^{2}x-\frac{45m_{3}}{\pi A}+\mathcal{O}\left(\frac{1}{A^{2}}\right)=:\mathcal{T}_{2},\label{eq:T2_def}
\end{equation}
\begin{equation}
229\widetilde{c}_{2,0}+41\widetilde{c}_{0,2}=\frac{360}{\pi A^{4}}\iint_{D_{A}}x_{1}^{4}B_{3}\left(\mathbf{x},h\right)\dd^{2}x+\frac{45m_{3}}{2\pi A}+\mathcal{O}\left(\frac{1}{A^{4}}\right)=:\mathcal{T}_{4},\label{eq:T4_def}
\end{equation}
\begin{equation}
109\widetilde{c}_{2,0}+11\widetilde{c}_{0,2}=\frac{576}{\pi A^{6}}\iint_{D_{A}}x_{1}^{6}B_{3}\left(\mathbf{x},h\right)\dd^{2}x+\frac{18m_{3}}{\pi A}+\mathcal{O}\left(\frac{1}{A^{5}}\right)=:\mathcal{T}_{6},\label{eq:T6_def}
\end{equation}
\begin{equation}
17\widetilde{c}_{2,0}+\widetilde{c}_{0,2}=\frac{1152}{7\pi A^{8}}\iint_{D_{A}}x_{1}^{8}B_{3}\left(\mathbf{x},h\right)\dd^{2}x+\frac{45m_{3}}{14\pi A}+\mathcal{O}\left(\frac{1}{A^{5}}\right)=:\mathcal{T}_{8}.\label{eq:T8_def}
\end{equation}

There are multiple ways to combine equations (\ref{eq:T0_def})--(\ref{eq:T8_def}) to eliminate $\widetilde{c}_{2,0}$ and $\widetilde{c}_{0,2}$.
In particular, we should use the following two relations which are
straightforward to verify:
\begin{equation}
\mathcal{T}_{0}=\mathcal{T}_{4}-2\mathcal{T}_{6},\hspace{1em}\hspace{1em}\mathcal{T}_{0}=4\mathcal{T}_{6}-25\mathcal{T}_{8}.\label{eq:T_lin_dep2}
\end{equation}

Substitution of (\ref{eq:T0_def}), (\ref{eq:T4_def}) and (\ref{eq:T6_def})
in the first relation of (\ref{eq:T_lin_dep2}) allows us to solve
for $m_{3}/A$. This leads to the third-order estimate of $m_{3}$
given by (\ref{eq:fin_estim_m3_ord3}).

Similarly, plugging (\ref{eq:T0_def}), (\ref{eq:T6_def}) and (\ref{eq:T8_def})
into the second relation of (\ref{eq:T_lin_dep2}) gives (\ref{eq:fin_estim_m3_ord4}),
a fourth-order estimate of $m_{3}$. 

\section{Numerical validation and practical considerations\label{sec:numerics}}

We demonstrate results by performing a numerical simulation on a synthetic
example. In this example, we choose magnetisation distribution to
consist of 4 magnetic dipoles: $\vec{\mathcal{M}}\left(\vec{x}\right)=\sum_{j=1}^{4} \vec{m}_{j} \delta\left(\vec{x}-\vec{x_j}\right)$ with $\delta$ denoting the Dirac delta function. The positions and the components of dipolar moments of each dipole are given in Table \ref{tabl:dipoles}.
%

\begin{table}
\begin{centering}
\begin{tabular}{|c|c|c|c|c|}
\hline 
$j$ & $1$ & $2$ & $3$ & $4$\tabularnewline
\hline 
\hline 
$x_{1}^{\left(j\right)}$ & $3.5\cdot10^{-5}$ $\text{m}$ & $0$ $\text{m}$ & $4.0\cdot10^{-5}$ $\text{m}$ & $-4.0\cdot10^{-5}$ $\text{m}$\tabularnewline
\hline 
$x_{2}^{\left(j\right)}$ & $3.0\cdot10^{-5}$ $\text{m}$ & $0$ $\text{m}$ & $-5.5\cdot10^{-5}$ $\text{m}$ & $5.5\cdot10^{-5}$ $\text{m}$\tabularnewline
\hline 
$x_{3}^{\left(j\right)}$ & $1.0\cdot10^{-5}$ $\text{m}$ & $7.0\cdot10^{-5}$ $\text{m}$ & $11.5\cdot10^{-5}$ $\text{m}$ & $2.5\cdot10^{-5}$ $\text{m}$\tabularnewline
\hline 
$m_{1}^{\left(j\right)}$ & $4.5\cdot10^{-12}$ $\text{A}\cdot\text{m}^{2}$ & $2.5\cdot10^{-12}$ $\text{A}\cdot\text{m}^{2}$ & $-3.0\cdot10^{-12}$ $\text{A}\cdot\text{m}^{2}$ & $-1.0\cdot10^{-12}$ $\text{A}\cdot\text{m}^{2}$\tabularnewline
\hline 
$m_{2}^{\left(j\right)}$ & $3.5\cdot10^{-12}$ $\text{A}\cdot\text{m}^{2}$ & $4.5\cdot10^{-12}$ $\text{A}\cdot\text{m}^{2}$ & $2.0\cdot10^{-12}$ $\text{A}\cdot\text{m}^{2}$ & $2.0\cdot10^{-12}$ $\text{A}\cdot\text{m}^{2}$\tabularnewline
\hline 
$m_{3}^{\left(j\right)}$ & $1.0\cdot10^{-12}$ $\text{A}\cdot\text{m}^{2}$ & $0.5\cdot10^{-12}$ $\text{A}\cdot\text{m}^{2}$ & $2.5\cdot10^{-12}$ $\text{A}\cdot\text{m}^{2}$ & $1.5\cdot10^{-12}$ $\text{A}\cdot\text{m}^{2}$\tabularnewline
\hline 
\end{tabular}
\par\end{centering}
\medskip
\caption{Positions and dipolar moments of the synthetic magnetisation distribution
(4 dipoles)}
\label{tabl:dipoles}
\end{table}

The net moment of this magnetisation distribution is equal to 
\begin{equation}
\vec{m}^{\text{true}}=\sum_{j=1}^{4}{\vec{m}}^{\left(j\right)}=\left(3.0,12.0,5.5\right)^{T}10^{-12}\,\,\text{A}\cdot\text{m}^{2}.\label{eq:net_mom_dipoles}
\end{equation}
The produced magnetic field $B_{3}$ is given by
\begin{equation}
B_{3}\left(\mathbf{x},h\right)=\frac{\mu_{0}}{4\pi}\sum_{j=1}^{4}\dfrac{3\left(h-x_{3}^{\left(j\right)}\right)\left[\left(x_{1}-x_{1}^{\left(j\right)}\right)m_{1}^{\left(j\right)}+\left(x_{2}-x_{2}^{\left(j\right)}\right)m_{2}^{\left(j\right)}\right]+\left(2\left(h-x_{3}^{\left(j\right)}\right)^{2}-\left|\mathbf{x}-\mathbf{x}^{\left(j\right)}\right|^{2}\right)m_{3}^{\left(j\right)}}{\left(\left|\mathbf{x}-\mathbf{x}^{\left(j\right)}\right|^{2}+\left(h-x_{3}^{\left(j\right)}\right)^{2}\right)^{5/2}},\label{eq:B3_dipoles}
\end{equation}
and is measured on the disk $D_{A}=\left\{ \mathbf{x}\in\mathbb{R}^{2}:\,\left|\mathbf{x}\right|<A\right\} $
at the height $x_{3}=h=2.5\cdot10^{-4}$ m. Since we now work in Si
units, we should recall Remark \ref{rem:about_units} and take into
account the previously omitted factor $\mu_{0}=4\pi\cdot10^{-7}$
N / A$^{2}$.

In order to check robustness of the moment estimates obtained in Theorem
\ref{thm:main}, we also perform simulations on data with a synthetic
noise. Namely, we modify $B_{3}$ using additive Gaussian white noise
with the amplitude $\sqrt{10^{-\text{SNR}/10}\cdot\text{Var}\left(B_{3}\right)}$,
where $\text{SNR}$ is the signal-to-noise ratio (in decibels) and
$\text{Var}\left(B_{3}\right)$ is the variance of $B_{3}$ on $D_{A}$.
For our simulations, we choose $\text{SNR}=20$ dB which corresponds
to the $10\%$ noise level.

In Figure \ref{fig:B3map}, we illustrate the field (\ref{eq:B3_dipoles})
and its noise component on the disk $D_{A}$ of radius $A=7.5\cdot10^{-4}$
m.

We shall now compute the integrals on the right-hand sides of (\ref{eq:fin_estim_m1m2_ord1})--(\ref{eq:fin_estim_m1m2_ord5})
for different values of $A$. According to our asymptotic result for
large $A$, we expect to see that, as $A$ grows, the values of each
of these integrals converge, with a different rate, to the value of
a component of the net moment given by (\ref{eq:net_mom_dipoles}).
Figures \ref{fig:m1m2}--\ref{fig:m3} show exactly that for the
tangential and normal net moment components, respectively. We note
that in Figure \ref{fig:m3} and further figures involving the normal
net moment component $m_{3}$, a pair of the estimates are used: one
with $x_{j}=x_{1}$ in (\ref{eq:fin_estim_m3_ord3}), (\ref{eq:fin_estim_m3_ord4})
and the other with $x_{j}=x_{2}$.

To illustrate better the variability of the convergence rates for different asymptotic estimates, in Figures \ref{fig:m1m2_log}--\ref{fig:m3_log},
we plot the differences $\left|m_{j}^{\texttt{true}}-m_{j}\right|$,
$j\in\left\{ 1,2,3\right\}$, against $A$ in logarithmic scale.

Finally, in Figures \ref{fig:m1m2_nsy}--\ref{fig:m3_nsy}, we directly test the estimates of the net moment components when the magnetic field is contaminated by noise (with the noise model described above). We observe the persistence of the lower-order estimates for the net moment components whereas the higher-order estimates clearly perform significantly worse.

Let us now briefly comment on the fact that we chose to illustrate
the results on a magnetisation distribution with a singular support.
Besides its simplicity, this choice is physically motivated as any
magnetisation can be thought of a combination of dipole sources. From
mathematical (numerical) viewpoint, a magnetic field produced by continuous
magnetisation distribution is given by the integral whose numerical
approximation (quadrature rule) is nothing but a weighted sum of dipoles.
Consequently, we do not expect results for continuous magnetisation
distributions to be of any drastic difference. On the other hand,
this highlights the applicability of our methodology to the magnetisations
that could be much more singular than smooth or square-integrable functions.

\begin{figure}
\includegraphics[scale=0.5]{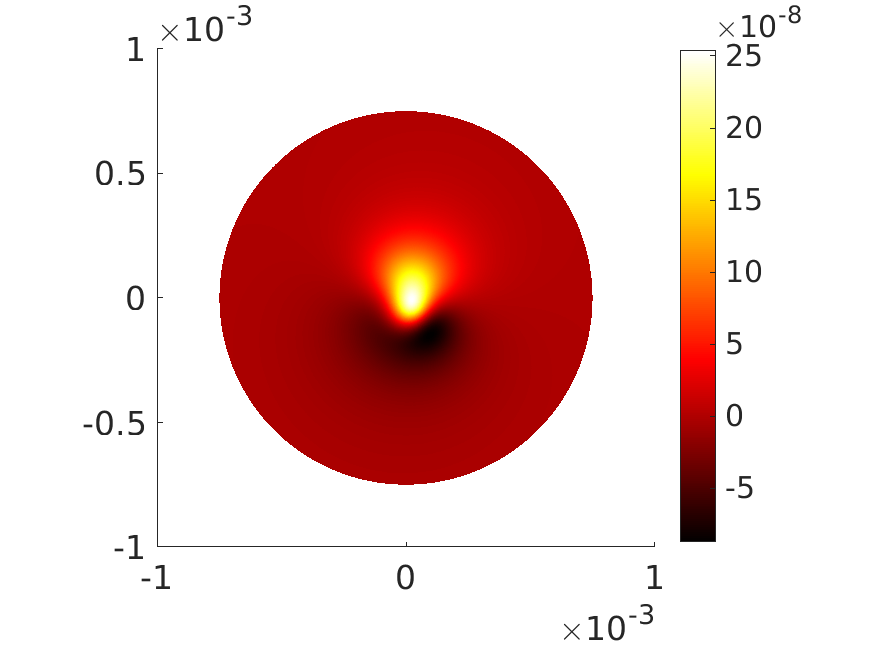}\hspace{15pt}\includegraphics[scale=0.5]{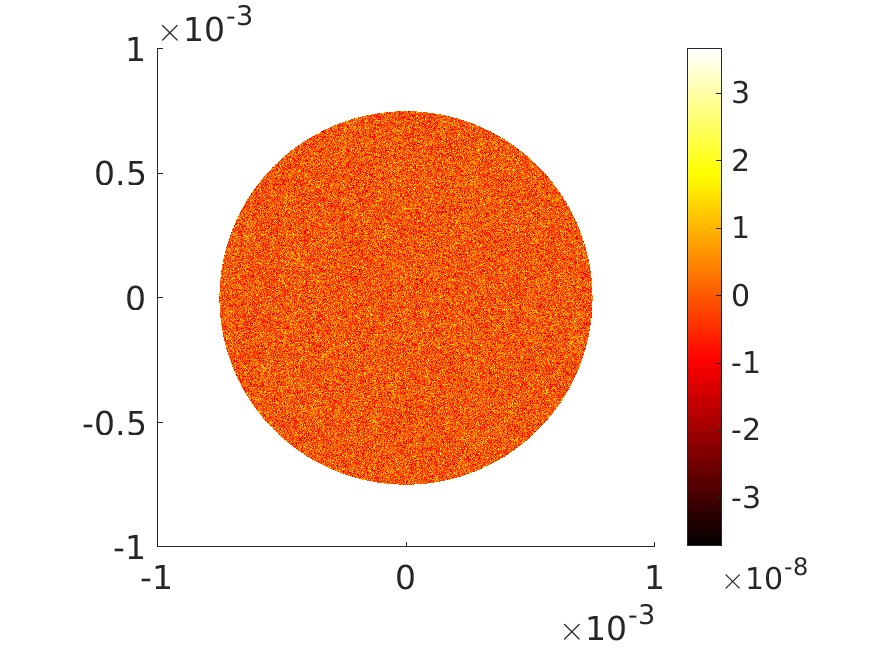}

\caption{\label{fig:B3map} Magnetic field $B_{3}\left(\mathbf{x},h\right)$
(left) and added noise (right) on $D_{A}$ for $A=7.5\cdot10^{-4}$
m.}
\end{figure}
\begin{figure}
\includegraphics[scale=0.5]{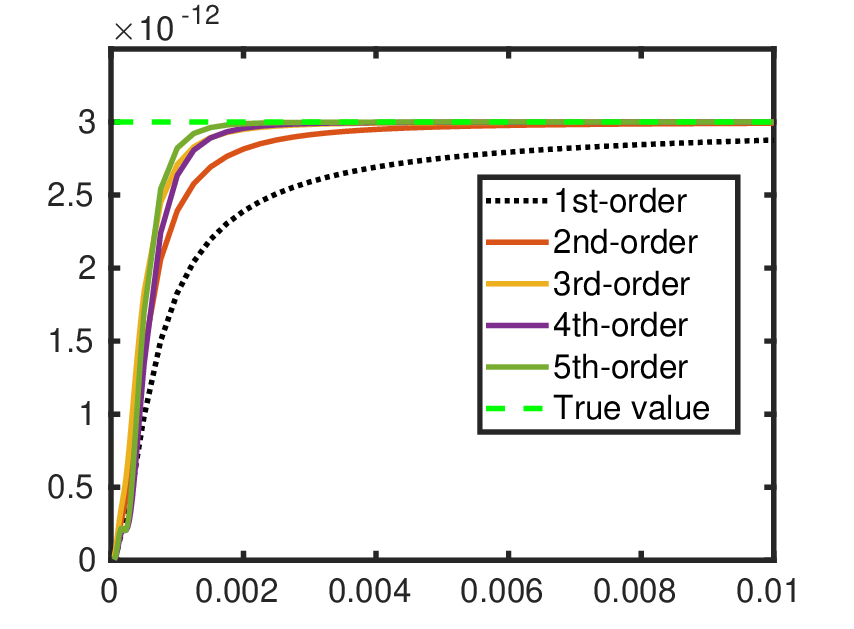}\hspace{15pt}\includegraphics[scale=0.5]{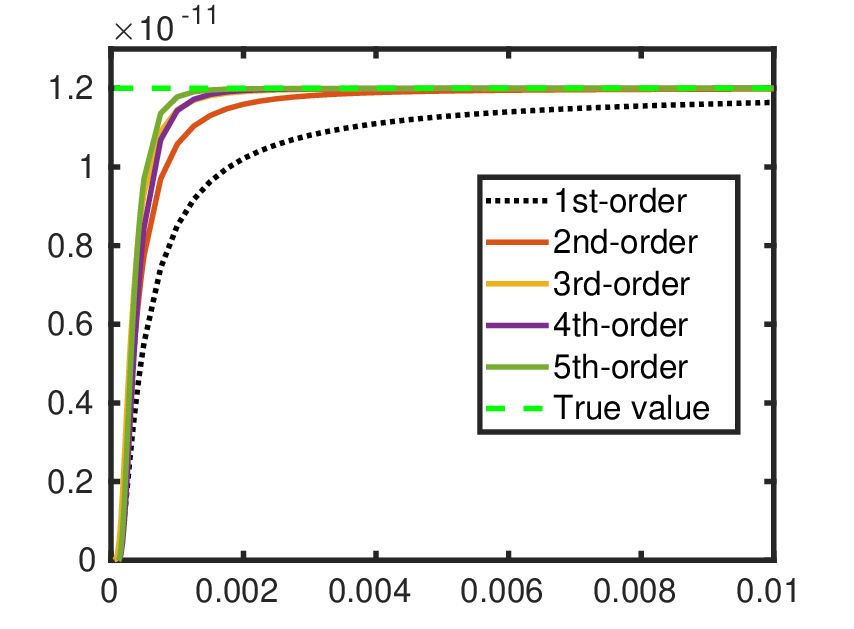}

\caption{\label{fig:m1m2} Estimates of the tangential net moment components
$m_{1}$ (left) and $m_{2}$ (right) versus $A$. }
\end{figure}

\begin{figure}
\includegraphics[scale=0.5]{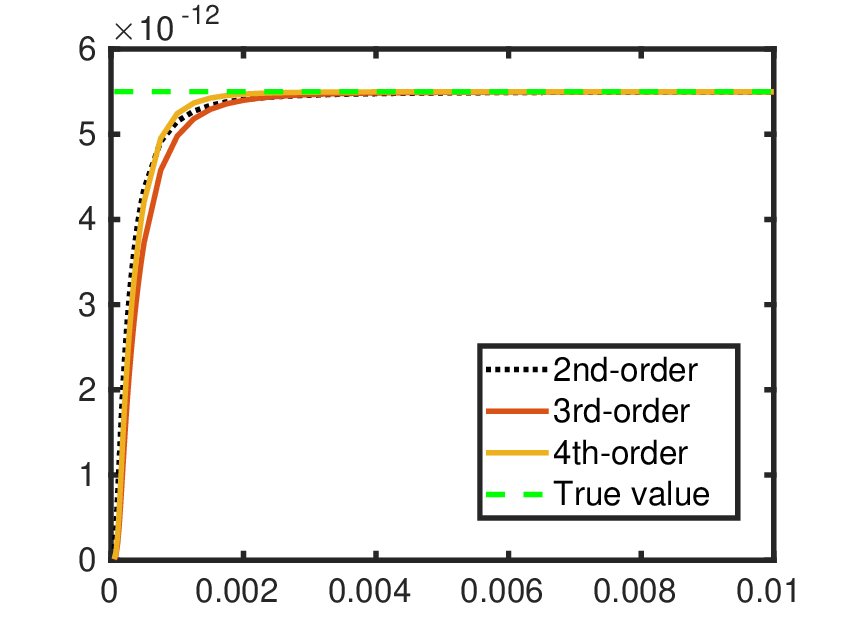}\hspace{15pt}\includegraphics[scale=0.5]{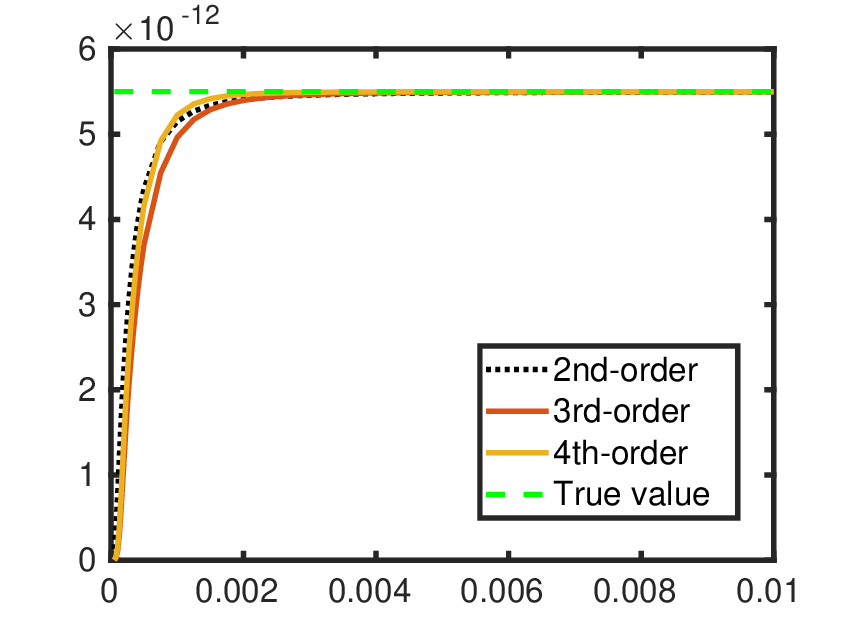}

\caption{\label{fig:m3} Estimates of the normal net moment component $m_{3}$
versus $A$: using $x_{1}$ (left) and $x_{2}$ (right) formulas.}
\end{figure}
\begin{figure}
\includegraphics[scale=0.5]{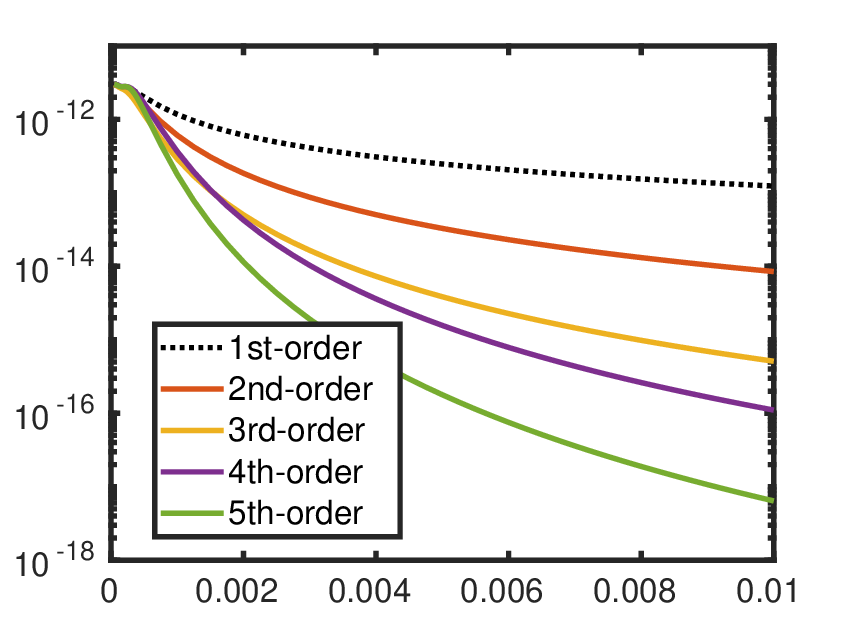}\hspace{15pt}\includegraphics[scale=0.5]{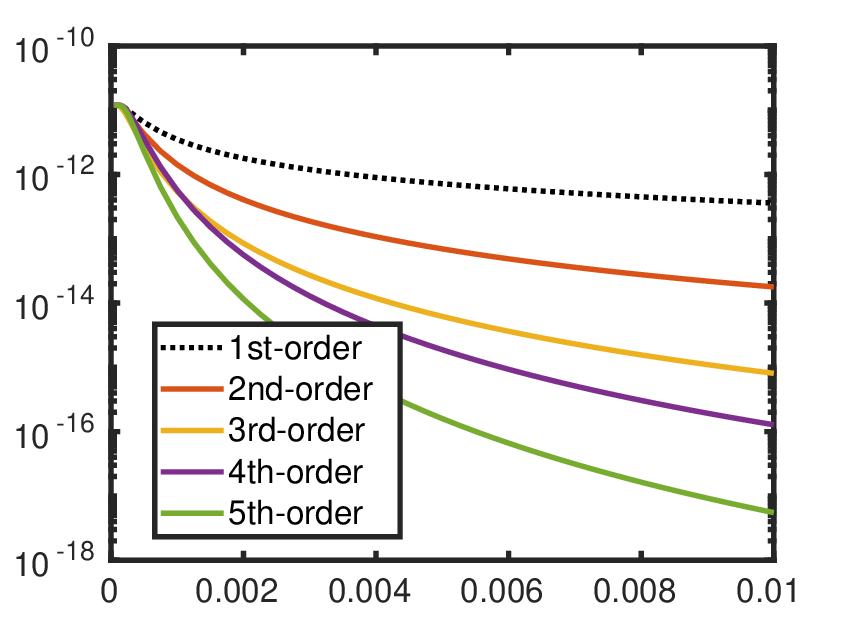}

\caption{\label{fig:m1m2_log} Convergence for the
estimates of $m_{1}$ (left) and $m_{2}$ (right).}
\end{figure}
\begin{figure}
\includegraphics[scale=0.5]{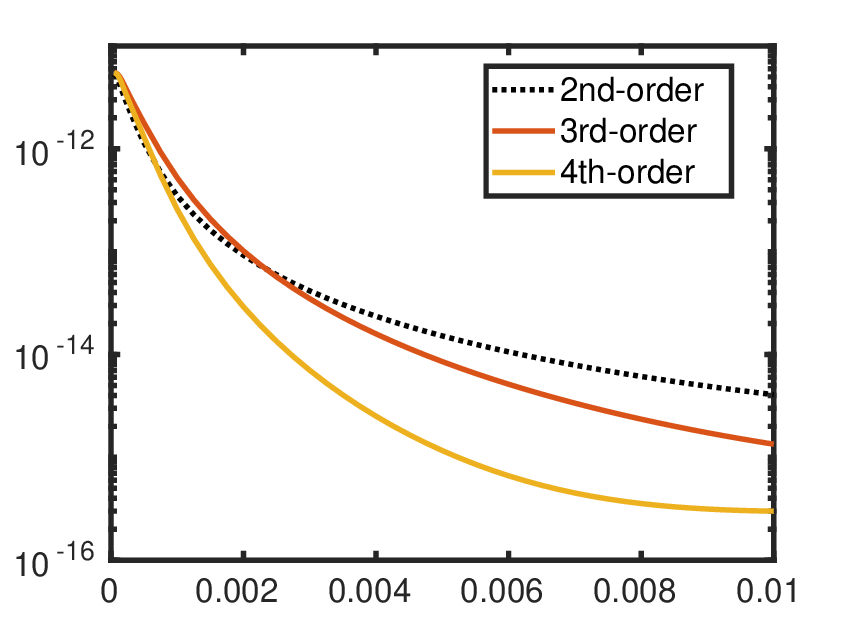}\hspace{15pt}\includegraphics[scale=0.5]{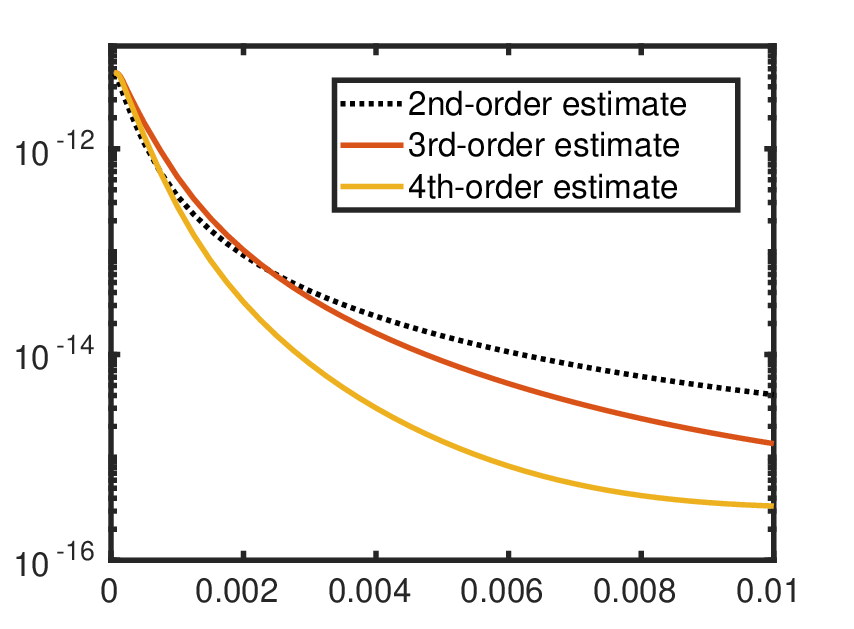}

\caption{\label{fig:m3_log} Convergence for the $m_{3}$
estimates: using $x_{1}$ (left) and $x_{2}$ (right) formulas.}
\end{figure}
\begin{figure}
\includegraphics[scale=0.5]{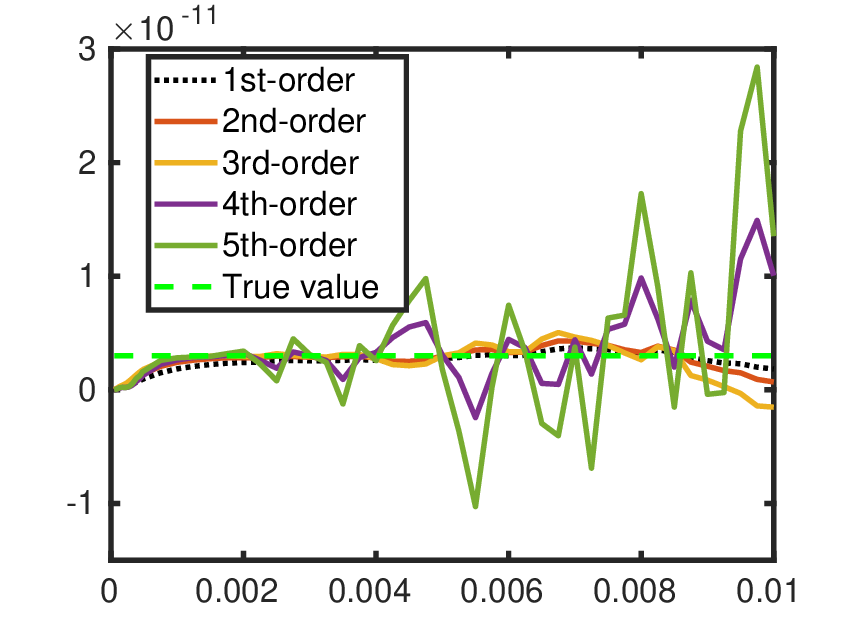}\hspace{15pt}\includegraphics[scale=0.5]{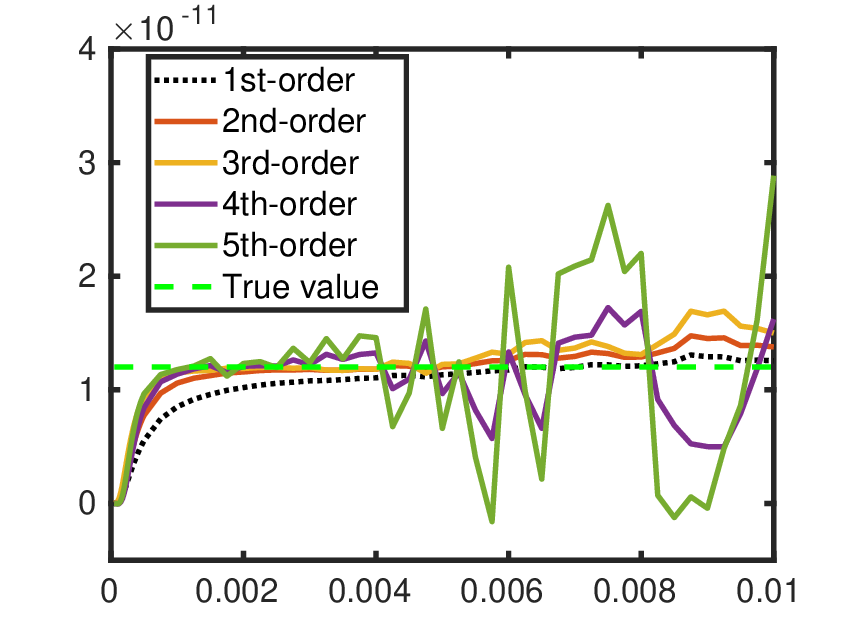}

\caption{\label{fig:m1m2_nsy} Estimates of the tangential net moment components
$m_{1}$ (left) and $m_{2}$ (right) versus $A$. Noisy data.}
\end{figure}

\begin{figure}
\includegraphics[scale=0.5]{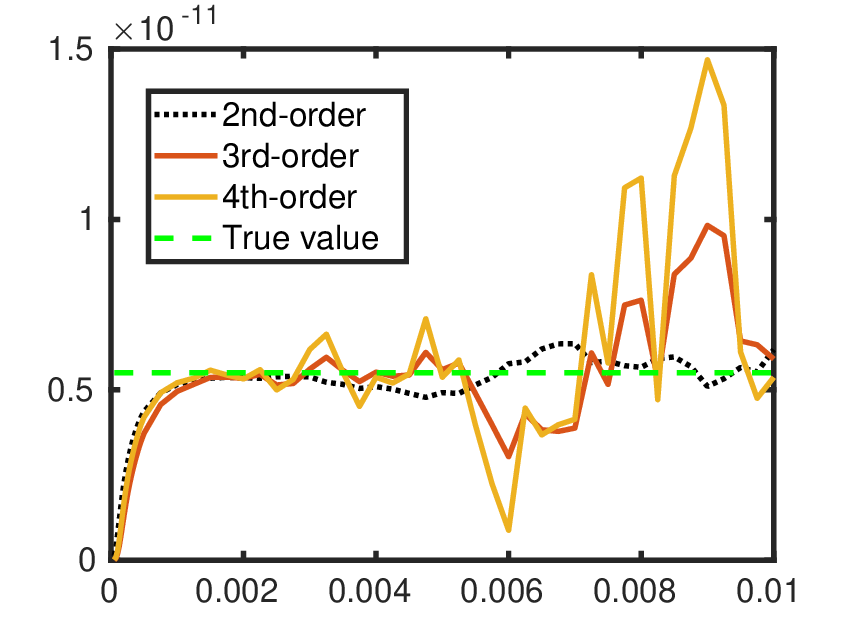}\hspace{15pt}\includegraphics[scale=0.5]{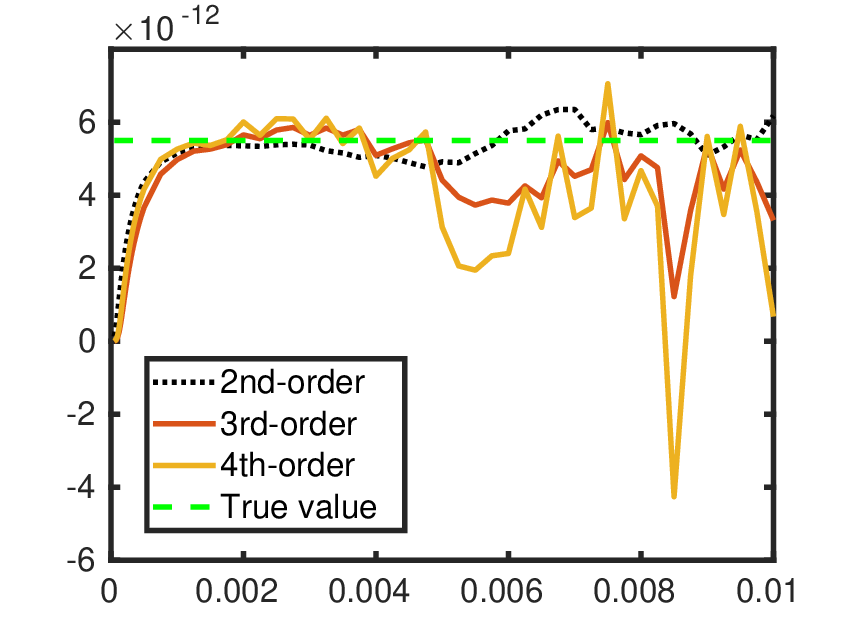}

\caption{\label{fig:m3_nsy} Estimates of the normal net moment component $m_{3}$
versus $A$: using $x_{1}$ (left) and $x_{2}$ (right) formulas.
Noisy data.}
\end{figure}
%

\section{Discussion and conclusion\label{sec:concl}}

Motivated by a concrete experimental set-up, we considered a problem
of estimating net magnetisation of a sample from one component
of the magnetic field available in the limited measurement area in
the plane above the sample. We approached this problem asymptotically,
assuming the size of the measurement area to be large. We derived
a set of explicit formulas for the asymptotic estimates of all three
components of the net moment. For simplicity, we considered only the
case of the circular geometry, i.e., where the measurement area is
a disk. Analogous results to those in Theorem
\ref{thm:main} could be deduced by our method also for the rectangular geometry which
technically is even closer to what is used in Paleomagnetism lab at
EAPS department of MIT, USA. The main difference
in obtaining such results with the present approach would be
a technique of the asymptotic estimation of Fourier integrals (such as the second term in the right-hand side of (\ref{eq:B3_decomp})) which involve both small and large parameters (and hence must likely rely on a partial availability of explicit integration formulas). 

In this paper, we have obtained and proved asymptotic estimates up
to order $5$ for the tangential net moment components $m_{1}$,
$m_{2}$ and up to order $4$ for the normal component $m_{3}$.
The main purpose was, however, to introduce a machinery that can generate
asymptotic estimates of an arbitrary order. It is clear from the proof
of Theorem \ref{thm:main} and auxiliary computations in Appendix
that the asymptotic order of the estimates can be upgraded by proceeding
in the established manner. This would require integration of the field
against polynomials of a higher order. Practical advantages of it,
however, are not yet obvious. First, such estimates are going to
be extremely sensitive to the presence of noise in $B_{3}$: it was demonstrated in Section
\ref{sec:numerics} that while the lower- and mid-order estimates lead to the expected results, the estimates of higher orders are much more prone to the instability. This is, of course, not surprising since, as it was mentioned, the moment recovery is a severely ill-posed problem (see Section \ref{sec:intro} and \cite{BarChevHarLebLimMar2019}), and an order of the asymptotic estimate plays the role of a regularisation parameter here. Second, when pursuing an estimate of a higher order,
one should not forget that the obtained result is of only asymptotic
nature: while a high-order estimate would be advantageous for very
large values of $A$, it may not be so for a smaller $A$ due to a
potentially large value of a multiplicative constant (in $A$) in
the remainder term. In particular, while Figures \ref{fig:m1m2} and \ref{fig:m1m2_log} show uniform improvement over the whole
range of $A$ when using estimates of higher order for $m_1$ and $m_3$, Figures \ref{fig:m3} and \ref{fig:m3_log} demonstrate that for some small pre-asymptotic range of $A$, an estimate for $m_3$ of a lower (second) order can be better than that of a higher (third) order. 
A relevant issue to bear in mind is that, apart from the basic asymptoticness
condition given by (\ref{eq:cond_asympt_thm}) (see also (\ref{eq:cond_asympt_thm_alt_smpl})), proceeding to higher-order
estimates assumes implicitly (but it is evident from the form of
the remainder terms) that the magnetisation is sufficiently localised
so that its higher algebraic moments $L_{j_{1},j_{2},j_{3}}^{\left(n\right)}:=\left\langle x_{1}^{j_{1}}x_{2}^{j_{2}}x_{3}^{j_{3}}M_{n}\right\rangle $,
$n\in\left\{ 1,2,3\right\}$, (see (\ref{eq:alg_moments_3D})) do
not grow too fast with respect to their order $j:=j_{1}+j_{2}+j_{3}$,
namely, that the quantities $L_{j_{1},j_{2},j_{3}}^{\left(n\right)}/A^{j}$
are not large for the value of $A$ in question. A conclusion to draw
from these observations is that the mid-order estimates (e.g., those
of orders $2$ and $3$) are perhaps the best: both from practical prospective
of proximity to the true value of the net moment and from the viewpoint
of the robustness to imperfect measurements. One can also consider
a possibility of choosing an estimate which is best possible for a
given magnitude of $A$. This resembles a problem of finding
an optimal truncation of asymptotic series. 

In future works, we should study relation of our asymptotic results
to stable estimates of the net moment using constrained optimisation
approaches, in particular those obtained in \cite{BarChevHarLebLimMar2019}
for planar $L^2$ magnetisations without the asymptotic assumption.
The shape of auxiliary functions to be integrated against the measured field to produce the net moment estimates looks similar to ours. For example, for estimates of $m_1$ of second order and higher in Theorem \ref{thm:main}, this auxiliary function is seen to consist of the main trend (the $2 x_1$ term) and a correction term (a polynomial depending in the "stretched" variable $x_1/A$) which is the largest towards the boundary of the disk. 
It would be curious to derive an explicit asymptotic from an integro-differential equation of \cite{BarChevHarLebLimMar2019} governing the optimal solution in a particular class. 
Studying other ways of stabilisation of the estimates and regularisation
is also important. In this respect, it is natural to ask a question:
what is the best way of using the redundancy of the set of asymptotic
formulas for the net moment components to arrive at an estimate with
the minimal effect of noise or to have the fastest (non-asymptotic)
convergence to the true value of a net moment? For instance, in the
present approach, we have a natural redundancy for all estimates of
$m_{3}$ of order $3$ and higher due to the freedom of choice $x_{j}=x_{1}$
and $x_{j}=x_{2}$ in formulas (\ref{eq:fin_estim_m3_ord3}), (\ref{eq:fin_estim_m3_ord4})
and so on. In Figure \ref{fig:m3}, these estimates give almost indistinguishable
results (it can be checked that the difference between the two is
not zero but very small), and hence may not look directly useful.
However, since having more relations than unknown quantities is always better when dealing with noise, this is still an advantage.
Also, it is clear from Step 3 of the proof of Theorem \ref{thm:main}
that, upon involving higher order polynomials, a richer set of lower-order
estimates could be obtained.

The form of asymptotic expressions suggests the direct use of polynomial (in a stretched variables $x_j/A$, $j\in\left\{1,2\right\}$) ansatz for the derivation of estimates of higher order, with a potential application of machine learning for identification of their coefficients.

Another path for a possible future work is to explore the possibility of obtaining asymptotic expansions analogous to those given in Theorem \ref{thm:main} but relying on the smallness of a slightly different asymptotic parameter. Namely, if, instead of $\left(x_1^2+x_2^2\right)^{-3/2}$ and $\left(x_1^2+x_2^2\right)^{-5/2}$ in (\ref{eq:expans_alg32}) and (\ref{eq:expans_alg52}), respectively, one factored out $\left(x_1^2+x_2^2+h^2\right)^{-3/2}$ and $\left(x_1^2+x_2^2+h^2\right)^{-5/2}$ and proceeded with appropriate modifications, the final asymptotic results would have a larger area of validity than that described by (\ref{eq:cond_asympt_thm}). In particular, this would cover a reduction to a dipolar case in a situation when the measurement area $D_A$ is not necessarily large but the value of $h$ is. However, since in the mentioned experimental set-up, the height $h$ is small, this modification would yield a little practical benefit but would complicate analysis of Fourier integrals. 

Along the same lines, it should be better understood why the height parameter $h$ does not enter any of the final estimates in the asymptotic regime. If this remains true for estimates of any order, can this be used, for example, as a shortcut to generate the higher-order estimates?   

The results of this work naturally connect to the issue of
the asymptotic field extension. Indeed, the asymptotic field expansion
at infinity (\ref{eq:B3_asympt}) is seen to feature $m_{3}$ at the
leading order and quantities $c_{1,0}$, $c_{0,1}$
at the next order. While the estimates of $m_{3}$ are given, to a
different order, in (\ref{eq:fin_estim_m3_ord2}), (\ref{eq:fin_estim_m3_ord3})
and (\ref{eq:fin_estim_m3_ord4}), it is evident from the proof of
Theorem \ref{thm:main}, that the quantities $c_{1,0}$,
$c_{0,1}$ can also be asymptotically estimated, see, 
e.g., (\ref{eq:a1_tild_ord4}) and (\ref{eq:a1_tild_ord5}) (and their
versions with $x_{1}$ replaced by $x_{2}$, as well as their lower-order
analogs). Thus, this furnishes an explicit $3$-term expansion of
$B_{3}$ at infinity which, in general, when solving inverse magnetisation
problems, should serve as a better alternative to a simple prolongation
of the field by zero outside of the measurement area. Such a strategy
can also potentially be used in setting up an iterative scheme. Of
course, all of this is meaningful only when the actual measurement
area is already large enough so that the asymptotic estimates
for $m_{3}$, $c_{1,0}$, $c_{0,1}$
are sufficiently accurate. 

Finally, on the practical side, it is essential to test the obtained results on a concrete application to see the actual improvement in the existing results and, more importantly, to address the cases which were not amenable for treatment with previous methodologies. Most natural context would certainly be that of estimation of the remanent magnetisation in (terrestrial or extraterrestrial) rock samples using scanning microscopy data, a setting which was the main motivation of the present work. However, due to close proximity of the setting to the active area of magnetic prospection, see, e.g., a very recent work \cite{Car2025}, it will be great to test the already derived formulas in that context. But also, since in that case all three components of the field are available instead of only one, it is worthwhile to investigate improvements of the presented results combining integrals of $B_1$, $B_2$ and $B_3$ giving higher-order estimates of the net moment $\vec{m}$.

\section*{Acknowledgements}

The author acknowledges inspiring discussions on the topic of inverse
magnetisation problems with L. Baratchart, S. Chevillard, J. Leblond,
C. Villalobos-Guill{\'e}n (Centre Inria d'Universit{\'e} C{\^o}te d'Azur, France), D. Hardin (Vanderbilt University,
USA) and E. A. Lima (MIT, USA).

\vspace{10pt}\appendix
\section{}
\renewcommand{\thesection}{\Alph{section}}

We collect here several useful results about cylindrical functions
and some relevant integrals. In what follows we will use the notation
$\mathbb{N}_{+}$ to denote natural numbers, $\mathbb{N}_{0}:=\mathbb{N}_{+}\cup\left\{ 0\right\} $,
and the notation $\mathbb{Z}$ for integer numbers. 

\subsection*{Basic facts about Bessel, Neumann and Struve functions}

The Bessel function $J_{n}$ of order $n\in\mathbb{Z}$ is an entire
function satisfying the differential equation \cite[(10.2.1)]{NIST}
\begin{equation}
z^{2}J_{n}^{\prime\prime}\left(z\right)+zJ_{n}^{\prime}\left(z\right)+\left(z^{2}-n^{2}\right)J_{n}\left(z\right)=0,\hspace{1em}z\in\mathbb{C}.\label{eq:BessJ_ODE}
\end{equation}
We have the following series expansion \cite[(10.2.2)]{NIST}
\begin{equation}
J_{n}\left(z\right)=\left(\frac{z}{2}\right)^{n}\sum_{k=0}^{\infty}\frac{\left(-1\right)^{k}}{k!\,\Gamma\left(n+k+1\right)}\left(\frac{z}{2}\right)^{2k},\label{eq:BessJ_series}
\end{equation}
which, due to the entire character of $J_{n}$, is absolutely convergent
for every $z\in\mathbb{C}$. Here, $\Gamma$ denotes the Euler gamma
function, for which we have, in particular, $\Gamma\left(k\right)=\left(k-1\right)!$
for $k\in\mathbb{N}_{+}$. Moreover, it is worth noting that $1/\Gamma\left(k\right)=0$
for $k\in\mathbb{Z}\backslash\mathbb{N}_{+}$, which implies vanishing
of negative powers of $z$ in expansion (\ref{eq:BessJ_series}) even
for negative orders $n$. \\
The following integral representation holds \cite[(10.9.1)]{NIST}
\begin{equation}
J_{n}\left(x\right)=\dfrac{1}{\pi}\int_{0}^{\pi}\cos\left(nt-x\sin t\right)\dd t,\hspace{1em}x\in\mathbb{R}.\label{eq:BessJ_int_repr}
\end{equation}
In particular,
\begin{equation}
J_{0}\left(x\right)=\dfrac{1}{\pi}\int_{0}^{\pi}\cos\left(x\sin t\right)\dd t=\dfrac{1}{\pi}\int_{0}^{\pi}\cos\left(x\cos t\right)\dd t=\frac{1}{2\pi}\int_{0}^{2\pi}\cos\left(x\cos t\right)\dd t,\hspace{1em}x\in\mathbb{R},\label{eq:BessJ0_int_repr}
\end{equation}
\begin{equation}
J_{1}\left(x\right)=\dfrac{1}{\pi}\int_{0}^{\pi}\cos\left(t-x\sin t\right)\dd t=\frac{1}{2\pi}\int_{0}^{2\pi}\sin\left(x\cos t\right)\cos t\dd t,\hspace{1em}x\in\mathbb{R}.\label{eq:BessJ1_int_repr}
\end{equation}
For $x\gg1$, $n\in\mathbb{N}_{0}$, the leading order asymptotics
reads \cite[(10.17.2)]{NIST}
\begin{equation}
J_{n}\left(x\right)=\left(\frac{2}{\pi x}\right)^{1/2}\cos\left(x-\frac{n\pi}{2}-\frac{\pi}{4}\right)+\mathcal{O}\left(\frac{1}{x^{3/2}}\right),\label{eq:BessJ_asympt}
\end{equation}
where the estimate of the remainder term is due to the discussion
in \cite[Sect. 10.17(iii)]{NIST}.\\
The Bessel functions $J_{n}$, $n\in\mathbb{Z},$ satisfy the connection
formula \cite[(10.4.1)]{NIST}
\begin{equation}
J_{-n}\left(x\right)=\left(-1\right)^{n}J_{n}\left(x\right),\label{eq:BessJ_conn}
\end{equation}
as well as simple recurrence relations \cite[(10.6.1)]{NIST}
\begin{equation}
\dfrac{1}{x}J_{n}\left(x\right)=\dfrac{1}{2n}\left(J_{n-1}\left(x\right)+J_{n+1}\left(x\right)\right),\hspace{1em}n\neq0,\label{eq:BessJ_recurr}
\end{equation}
\begin{equation}
J_{n}^{\prime}\left(x\right)=\dfrac{1}{2}\left(J_{n-1}\left(x\right)-J_{n+1}\left(x\right)\right).\label{eq:BessJ_recurr_diff}
\end{equation}
In particular, (\ref{eq:BessJ_conn}) and (\ref{eq:BessJ_recurr_diff})
entail that
\begin{equation}
J_{0}^{\prime}\left(x\right)=-J_{1}\left(x\right).\label{eq:J0der_J1_conn}
\end{equation}

The Struve function $H_{n}$ of order $n\in\mathbb{N}_{0}\cup\left\{ -1\right\} $
is an entire function defined by the absolutely convergent power series
\cite[(11.2.1)]{NIST}
\begin{equation}
H_{n}\left(z\right)=\left(\dfrac{z}{2}\right)^{n+1}\sum_{k=0}^{\infty}\dfrac{\left(-1\right)^{k}}{\Gamma\left(k+\dfrac{3}{2}\right)\Gamma\left(k+n+\dfrac{3}{2}\right)}\left(\dfrac{z}{2}\right)^{2k}.\label{eq:StruvH}
\end{equation}
The companion Struve function $K_{n}$ of order $n\in\mathbb{Z}$
is defined \cite[(11.2.5)]{NIST} as
\begin{equation}
K_{n}\left(z\right)=H_{n}\left(z\right)-Y_{n}\left(z\right),\label{eq:StruvK}
\end{equation}
where $Y_{n}$ is the Neumann function.\\
For $x\gg1$, $n\in\mathbb{N}_{0}$, the following asymptotics hold
true \cite[(11.6.1)]{NIST}
\begin{equation}
K_{n}\left(x\right)=\frac{1}{\pi}\frac{2^{n+1}n!}{\left(2n\right)!}x^{n-1}+\mathcal{O}\left(x^{n-3}\right),\label{eq:StruvK_asympt}
\end{equation}
\begin{equation}
Y_{n}\left(x\right)=\left(\frac{2}{\pi x}\right)^{1/2}\sin\left(x-\frac{n\pi}{2}-\frac{\pi}{4}\right)+\mathcal{O}\left(\frac{1}{x^{3/2}}\right),\label{eq:BessY_asympt}
\end{equation}
and the remainder terms are discussed in \cite[Sect. 11.6(i)]{NIST}
and \cite[Sect. 10.17(iii)]{NIST}, respectively. The asymptotic behaviour
of $H_{n}\left(x\right)$ for $x\gg1$, $n\in\mathbb{N}_{0}$, hence
follows from (\ref{eq:StruvK})--(\ref{eq:BessY_asympt}). In particular,
for $x\gg1$, 
\begin{equation}
H_{0}\left(x\right)=\left(\frac{2}{\pi x}\right)^{1/2}\sin\left(x-\frac{\pi}{4}\right)+\mathcal{O}\left(\frac{1}{x}\right),\label{eq:StruvH0_asympt}
\end{equation}
\begin{equation}
H_{1}\left(x\right)=\frac{2}{\pi}+\left(\frac{2}{\pi x}\right)^{1/2}\sin\left(x-\frac{3\pi}{4}\right)+\mathcal{O}\left(\frac{1}{x^{3/2}}\right).\label{eq:StruvH1_asympt}
\end{equation}
Moreover, we have the following connection formula
\begin{equation}
H_{-1}\left(x\right)=\dfrac{2}{\pi}-H_{1}\left(x\right).\label{eq:StruvK_conn}
\end{equation}

\subsection*{Some useful integrals pertinent to the cylindrical functions}

The following lemmas establish several integral relations which are
also crucial for the proof in Section \ref{sec:proof}. 
\begin{lem}
\label{lem:int_J1_x_odd} For $n\in\mathbb{N}_{+}$, $\rho>0$, the
following identity holds
\begin{equation}
\int_{\rho}^{\infty}\frac{J_{1}\left(x\right)}{x^{2n+1}}\dd x=\frac{1}{4n^{2}-1}\left[2n\frac{J_{1}\left(\rho\right)}{\rho^{2n}}+\frac{J_{1}^{\prime}\left(\rho\right)}{\rho^{2n-1}}-\int_{\rho}^{\infty}\frac{J_{1}\left(x\right)}{x^{2n-1}}\dd x\right].\label{eq:int_J1_x_odd}
\end{equation}
\end{lem}
\begin{proof}
First, upon integration by parts (using the asymptotic behaviour at
infinity of $J_{1}$ given by (\ref{eq:BessJ_asympt})), we have
\begin{equation}
\int_{\rho}^{\infty}\frac{J_{1}\left(x\right)}{x^{2n+1}}\dd x=\frac{1}{2n}\frac{J_{1}\left(\rho\right)}{\rho^{2n}}+\frac{1}{2n}\int_{\rho}^{\infty}\frac{J_{1}^{\prime}\left(x\right)}{x^{2n}}\dd x.\label{eq:int_J1_x_odd_prelim1}
\end{equation}
Then, if we employ (\ref{eq:BessJ_ODE}) to express $J_{1}^{\prime}$
in terms of $J_{1}$ and $J_{1}^{\prime\prime}$, we obtain
\[
\int_{\rho}^{\infty}\frac{J_{1}\left(x\right)}{x^{2n+1}}\dd x=\frac{1}{2n}\frac{J_{1}\left(\rho\right)}{\rho^{2n}}+\frac{1}{2n}\int_{\rho}^{\infty}\frac{J_{1}\left(x\right)}{x^{2n+1}}\dd x-\frac{1}{2n}\int_{\rho}^{\infty}\frac{J_{1}\left(x\right)}{x^{2n-1}}\dd x-\frac{1}{2n}\int_{\rho}^{\infty}\frac{J_{1}^{\prime\prime}\left(x\right)}{x^{2n-1}}\dd x,
\]
and hence
\begin{equation}
\int_{\rho}^{\infty}\frac{J_{1}^{\prime\prime}\left(x\right)}{x^{2n-1}}\dd x=\frac{J_{1}\left(\rho\right)}{\rho^{2n}}-\left(2n-1\right)\int_{\rho}^{\infty}\frac{J_{1}\left(x\right)}{x^{2n+1}}\dd x-\int_{\rho}^{\infty}\frac{J_{1}\left(x\right)}{x^{2n-1}}\dd x.\label{eq:int_J1_x_odd_prelim2}
\end{equation}
On the other hand, returning to (\ref{eq:int_J1_x_odd_prelim1}) and
integrating it by parts again, we arrive at
\begin{align}
\int_{\rho}^{\infty}\frac{J_{1}\left(x\right)}{x^{2n+1}}\dd x & =\frac{1}{2n}\frac{J_{1}\left(\rho\right)}{\rho^{2n}}+\frac{1}{2n\left(2n-1\right)}\frac{J_{1}^{\prime}\left(\rho\right)}{\rho^{2n-1}}+\frac{1}{2n\left(2n-1\right)}\int_{\rho}^{\infty}\frac{J_{1}^{\prime\prime}\left(x\right)}{x^{2n-1}}\dd x\label{eq:int_J1_x_odd_prelim3}\\
 & =\frac{1}{2n-1}\frac{J_{1}\left(\rho\right)}{\rho^{2n}}+\frac{1}{2n\left(2n-1\right)}\frac{J_{1}^{\prime}\left(\rho\right)}{\rho^{2n-1}}-\frac{1}{2n}\int_{\rho}^{\infty}\frac{J_{1}\left(x\right)}{x^{2n+1}}\dd x-\frac{1}{2n\left(2n-1\right)}\int_{\rho}^{\infty}\frac{J_{1}\left(x\right)}{x^{2n-1}}\dd x.\nonumber 
\end{align}
Here, in the second line, we eliminated the integral term involving
$J_{1}^{\prime\prime}$ using (\ref{eq:int_J1_x_odd_prelim2}).

Rearranging the terms in (\ref{eq:int_J1_x_odd_prelim3}) (solving
for the quantity on the left-hand side), we deduce (\ref{eq:int_J1_x_odd}). 
\end{proof}
\begin{lem}
\label{lem:int_J1_x3} For $\rho>0$, we have 
\begin{equation}
\int_{\rho}^{\infty}\frac{J_{1}\left(x\right)}{x^{3}}\dd x=\frac{J_{0}\left(\rho\right)}{3\rho}+\frac{J_{1}\left(\rho\right)}{3\rho^{2}}-\frac{1}{3}-\frac{J_{1}\left(\rho\right)}{3}+\frac{\rho J_{0}\left(\rho\right)}{3}-\frac{\pi\rho}{6}\left[J_{0}\left(\rho\right)H_{1}\left(\rho\right)-J_{1}\left(\rho\right)H_{0}\left(\rho\right)\right].\label{eq:int_J1_x3}
\end{equation}
\end{lem}
\begin{proof}
Note that using (\ref{eq:BessJ_recurr}), we have
\begin{equation}
\int_{\rho}^{\infty}\dfrac{J_{1}\left(x\right)}{x^{3}}\dd x=\dfrac{1}{2}\int_{\rho}^{\infty}\dfrac{J_{0}\left(x\right)}{x^{2}}\dd x+\dfrac{1}{2}\int_{\rho}^{\infty}\dfrac{J_{2}\left(x\right)}{x^{2}}\dd x.\label{eq:int_J1_x3_prelim1}
\end{equation}

Let us start by transforming the first term in (\ref{eq:int_J1_x3_prelim1}),
namely,
\begin{align}
\int_{\rho}^{\infty}\dfrac{J_{0}\left(x\right)}{x^{2}}\dd x & =\frac{J_{0}\left(\rho\right)}{\rho}+\int_{\rho}^{\infty}\dfrac{J_{0}^{\prime}\left(x\right)}{x}\dd x\label{eq:int_J1_x3_prelim2}\\
 & =\frac{J_{0}\left(\rho\right)}{\rho}+J_{0}^{\prime}\left(\rho\right)-\int_{\rho}^{\infty}J_{0}\left(x\right)\dd x.\nonumber 
\end{align}
Here, in the first line, we employed integration by parts (together
with the asymptotic behaviour at infinity of $J_{0}$ given by (\ref{eq:BessJ_asympt})).
To arrive at the second line, we used the identity
\[
\frac{J_{0}^{\prime}\left(x\right)}{x}=-J_{0}\left(x\right)-J_{0}^{\prime\prime}\left(x\right),\hspace{1em}x\neq0,
\]
implied by (\ref{eq:BessJ_ODE}). 

Performing the same procedure with the second term in (\ref{eq:int_J1_x3_prelim1}),
we have
\begin{align*}
\int_{\rho}^{\infty}\dfrac{J_{2}\left(x\right)}{x^{2}}\dd x & =\frac{J_{2}\left(\rho\right)}{\rho}+\int_{\rho}^{\infty}\dfrac{J_{2}^{\prime}\left(x\right)}{x}\dd x\\
 & =\frac{J_{2}\left(\rho\right)}{\rho}+J_{2}^{\prime}\left(\rho\right)-\int_{\rho}^{\infty}J_{2}\left(x\right)\dd x+4\int_{\rho}^{\infty}\dfrac{J_{2}\left(x\right)}{x^{2}}\dd x,
\end{align*}
and, hence,
\begin{equation}
\int_{\rho}^{\infty}\dfrac{J_{2}\left(x\right)}{x^{2}}\dd x=-\frac{1}{3}\left[\frac{J_{2}\left(\rho\right)}{\rho}+J_{2}^{\prime}\left(\rho\right)-\int_{\rho}^{\infty}J_{2}\left(x\right)\dd x\right].\label{eq:int_J1_x3_prelim3}
\end{equation}

Expressions (\ref{eq:int_J1_x3_prelim2}) and (\ref{eq:int_J1_x3_prelim3})
imply that the integral on the left-hand side of (\ref{eq:int_J1_x3_prelim1})
is expressible in terms of two integral quantities: $\int_{\rho}^{\infty}J_{0}\left(x\right)dx$
and $\int_{\rho}^{\infty}J_{2}\left(x\right)dx$. Let us now show
that these quantities are simply related:
\begin{equation}
\int_{\rho}^{\infty}J_{2}\left(x\right)dx=\int_{\rho}^{\infty}J_{0}\left(x\right)dx+2J_{1}\left(\rho\right),\label{eq:int_J2J1_conn}
\end{equation}
and, moreover,
\begin{equation}
\int_{\rho}^{\infty}J_{0}\left(x\right)dx=1-\rho J_{0}\left(\rho\right)-\frac{\pi}{2}\rho\left[J_{1}\left(\rho\right)H_{0}\left(\rho\right)-J_{0}\left(\rho\right)H_{1}\left(\rho\right)\right].\label{eq:int_J0_Struv_conn}
\end{equation}

Recalling asymptotics (\ref{eq:BessJ_asympt}), we note that the integrals
on the left-hand sides of (\ref{eq:int_J2J1_conn}), (\ref{eq:int_J0_Struv_conn})
are not absolutely convergent, and hence an additional care with technical
manipulations is needed. Namely, we shall first replace the integration
range $\left(\rho,\infty\right)$ with $\left(\rho,R\right)$ for
arbitrary finite $R>\rho$, and then pass to the limit as $R\rightarrow+\infty$.
We shall proceed in several steps.

\medskip

\uline{Step 1: Establishing \mbox{(\ref{eq:int_J2J1_conn})}} 

\medskip

We start with (\ref{eq:int_J2J1_conn}) and use integral representation
(\ref{eq:BessJ_int_repr}):
\[
J_{2}\left(x\right)=\frac{1}{\pi}\int_{0}^{\pi}\cos\left(2t-x\sin t\right)\dd t=\frac{1}{\pi}\int_{0}^{\pi}\cos\left(x\sin t-2t\right)\dd t.
\]
Consequently, exchanging the order of integration (permissible due
to the regularity of the integrand and finiteness of the integration
limits), we obtain
\begin{align}
\int_{\rho}^{R}J_{2}\left(x\right)\dd x= & \frac{1}{\pi}\int_{0}^{\pi}\frac{\sin\left(R\sin t-2t\right)-\sin\left(\text{\ensuremath{\rho\sin t-2t}}\right)}{\sin t}\dd t\label{eq:int_J2J1_conn_prelim}\\
= & \frac{1}{\pi}\int_{0}^{\pi}\frac{\sin\left(R\sin t\right)-\sin\left(\text{\ensuremath{\rho\sin t}}\right)}{\sin t}\dd t-\frac{2}{\pi}\int_{0}^{\pi}\left[\sin\left(R\sin t\right)\sin t+\cos\left(R\sin t\right)\cos t\right]\dd t\nonumber \\
 & +\frac{2}{\pi}\int_{0}^{\pi}\left[\sin\left(\rho\sin t\right)\sin t+\cos\left(\rho\sin t\right)\cos t\right]\dd t\nonumber \\
= & \int_{\rho}^{R}J_{0}\left(x\right)dx-2\left(J_{1}\left(R\right)-J_{1}\left(\rho\right)\right),\nonumber 
\end{align}
where we used the identities
\[
\sin\left(R\sin t-2t\right)=\left(1-2\sin^{2}t\right)\sin\left(R\sin t\right)-2\sin t\cos t\cos\left(R\sin t\right),
\]

\[
\sin\left(R\sin t\right)\sin t+\cos\left(R\sin t\right)\cos t=\cos\left(R\sin t-t\right),
\]
\[
\frac{1}{\pi}\int_{0}^{\pi}\cos\left(R\sin t-t\right)\dd t=J_{1}\left(R\right),
\]
\[
\frac{1}{\pi}\int_{0}^{\pi}\frac{\sin\left(R\sin t\right)-\sin\left(\rho\sin t\right)}{\sin t}\dd t=\int_{\rho}^{R}J_{0}\left(x\right)\dd x.
\]
and, except for the last one, also their analogs with $\rho$ instead
of $R$. The first two of these identities are purely trigonometrical
whereas the last two are due to (\ref{eq:BessJ0_int_repr})--(\ref{eq:BessJ1_int_repr}).\\
Passing to the limit $R\rightarrow+\infty$ in (\ref{eq:int_J2J1_conn_prelim})
using asymptotics (\ref{eq:BessJ_asympt}), we thus conclude with
(\ref{eq:int_J2J1_conn}).

\medskip

\uline{Step 2: Establishing \mbox{(\ref{eq:int_J0_Struv_conn})}}

\medskip

To deduce relation (\ref{eq:int_J0_Struv_conn}), we use \cite[(10.22.2)]{NIST}
to write
\begin{align}
\int_{\rho}^{R}J_{0}\left(x\right)dx & =\frac{\pi}{2}R\left[J_{0}\left(R\right)H_{-1}\left(R\right)-J_{-1}\left(R\right)H_{0}\left(R\right)\right]-\frac{\pi}{2}\rho\left[J_{0}\left(\rho\right)H_{-1}\left(\rho\right)-J_{-1}\left(\rho\right)H_{0}\left(\rho\right)\right]\label{eq:int_J0_Struv_prelim1}\\
 & =RJ_{0}\left(R\right)-\rho J_{0}\left(\rho\right)+\frac{\pi}{2}R\left[J_{1}\left(R\right)H_{0}\left(R\right)-J_{0}\left(R\right)H_{1}\left(R\right)\right]-\frac{\pi}{2}\rho\left[J_{1}\left(\rho\right)H_{0}\left(\rho\right)-J_{0}\left(\rho\right)H_{1}\left(\rho\right)\right],\nonumber 
\end{align}
where we used (\ref{eq:BessJ_conn}), (\ref{eq:StruvK_conn}) in passing
to the second line.\\
Note that, employing asymptotics (\ref{eq:BessJ_asympt}), (\ref{eq:StruvH0_asympt})--(\ref{eq:StruvH1_asympt}),
we have, for $R\gg1$,
\begin{align*}
RJ_{0}\left(R\right)+\frac{\pi}{2}R\left[J_{1}\left(R\right)H_{0}\left(R\right)-J_{0}\left(R\right)H_{1}\left(R\right)\right] & =\cos\left(R-\frac{3\pi}{4}\right)\sin\left(R-\frac{\pi}{4}\right)-\cos\left(R-\frac{\pi}{4}\right)\sin\left(R-\frac{3\pi}{4}\right)+\mathcal{O}\left(\frac{1}{R^{1/2}}\right)\\
 & =1+\mathcal{O}\left(\frac{1}{R^{1/2}}\right).
\end{align*}
Therefore, by passing to the limit as $R\rightarrow+\infty$ in (\ref{eq:int_J0_Struv_prelim1}),
we obtain (\ref{eq:int_J0_Struv_conn}).

\medskip

\uline{Step 3: Conclusion of the proof}

\medskip

Plugging (\ref{eq:int_J0_Struv_conn}) into (\ref{eq:int_J2J1_conn})
and (\ref{eq:int_J1_x3_prelim2}), we obtain
\begin{equation}
\int_{\rho}^{\infty}J_{2}\left(x\right)\dd x=1+2J_{1}\left(\rho\right)-\rho J_{0}\left(\rho\right)-\frac{\pi}{2}\rho\left[J_{1}\left(\rho\right)H_{0}\left(\rho\right)-J_{0}\left(\rho\right)H_{1}\left(\rho\right)\right],\label{eq:int_J2}
\end{equation}
\begin{equation}
\int_{\rho}^{\infty}\frac{J_{0}\left(x\right)}{x^{2}}\dd x=\frac{J_{0}\left(\rho\right)}{\rho}+J_{0}^{\prime}\left(\rho\right)-1+\rho J_{0}\left(\rho\right)+\frac{\pi}{2}\rho\left[J_{1}\left(\rho\right)H_{0}\left(\rho\right)-J_{0}\left(\rho\right)H_{1}\left(\rho\right)\right],\label{eq:int_J0_x2}
\end{equation}
respectively.\\
Substitution of (\ref{eq:int_J2}) into (\ref{eq:int_J1_x3_prelim3})
gives
\begin{equation}
\int_{\rho}^{\infty}\dfrac{J_{2}\left(x\right)}{x^{2}}\dd x=-\frac{J_{2}\left(\rho\right)}{3\rho}-\frac{1}{3}J_{2}^{\prime}\left(\rho\right)+\frac{1}{3}+\frac{2}{3}J_{1}\left(\rho\right)-\frac{1}{3}\rho J_{0}\left(\rho\right)-\frac{\pi}{6}\rho\left[J_{1}\left(\rho\right)H_{0}\left(\rho\right)-J_{0}\left(\rho\right)H_{1}\left(\rho\right)\right].\label{eq:int_J2_x2}
\end{equation}
Finally, using (\ref{eq:int_J0_x2})--(\ref{eq:int_J2_x2}) in (\ref{eq:int_J1_x3_prelim1}),
we arrive at
\begin{equation}
\int_{\rho}^{\infty}\dfrac{J_{1}\left(x\right)}{x^{3}}\dd x=\frac{J_{0}\left(\rho\right)}{2\rho}-\frac{J_{2}\left(\rho\right)}{6\rho}+\frac{1}{2}J_{0}^{\prime}\left(\rho\right)-\frac{1}{6}J_{2}^{\prime}\left(x\right)-\frac{1}{3}+\frac{\rho}{3}J_{0}\left(\rho\right)+\frac{1}{3}J_{1}\left(\rho\right)+\frac{\pi}{6}\rho\left[J_{1}\left(\rho\right)H_{0}\left(\rho\right)-J_{0}\left(\rho\right)H_{1}\left(\rho\right)\right].\label{eq:intJ_1_x3_prefin}
\end{equation}
Using (\ref{eq:J0der_J1_conn}) and the recursive identities (due
to (\ref{eq:BessJ_recurr})--(\ref{eq:BessJ_recurr_diff}))
\[
J_{2}\left(\rho\right)=\frac{2}{\rho}J_{1}\left(\rho\right)-J_{0}\left(\rho\right),\hspace{1em}\hspace{1em}J_{3}\left(\rho\right)=\frac{4}{\rho}J_{2}\left(\rho\right)-J_{1}\left(\rho\right)=\frac{8}{\rho^{2}}J_{1}\left(\rho\right)-\frac{4}{\rho}J_{0}\left(\rho\right)-J_{1}\left(\rho\right),
\]
\[
J_{2}^{\prime}\left(\rho\right)=\frac{1}{2}\left[J_{1}\left(\rho\right)-J_{3}\left(\rho\right)\right]=-\frac{4}{\rho^{2}}J_{1}\left(\rho\right)+\frac{2}{\rho}J_{0}\left(\rho\right)+J_{1}\left(\rho\right),
\]
we transform (\ref{eq:intJ_1_x3_prefin}) into the desired relation
(\ref{eq:int_J1_x3}).
\end{proof}
\begin{lem}
\label{lem:int_J1_var} For $\rho>0$, in addition to (\ref{eq:int_J1_x3}),
we have the following identities
\begin{align}
\int_{\rho}^{\infty}\dfrac{J_{1}\left(x\right)}{x}\dd x= & \frac{J_{1}\left(\rho\right)}{\rho^{2}}+\frac{J_{1}^{\prime}\left(\rho\right)}{\rho}-\frac{1}{\rho}J_{0}\left(\rho\right)+1+J_{1}\left(\rho\right)-\rho J_{0}\left(\rho\right)\label{eq:int_J1_x}\\
 & +\frac{\pi\rho}{2}\left[J_{0}\left(\rho\right)H_{1}\left(\rho\right)-J_{1}\left(\rho\right)H_{0}\left(\rho\right)\right],\nonumber 
\end{align}
\begin{align}
\int_{\rho}^{\infty}\dfrac{J_{1}\left(x\right)}{x^{5}}\dd x= & \frac{4J_{1}\left(\rho\right)}{15\rho^{4}}+\frac{J_{1}^{\prime}\left(\rho\right)}{15\rho^{3}}-\frac{J_{0}\left(\rho\right)}{45\rho}-\frac{J_{1}\left(\rho\right)}{45\rho^{2}}+\frac{1}{45}+\frac{J_{1}\left(\rho\right)}{45}\label{eq:int_J1_x5}\\
 & -\frac{\rho J_{0}\left(\rho\right)}{45}+\frac{\pi\rho}{90}\left[J_{0}\left(\rho\right)H_{1}\left(\rho\right)-J_{1}\left(\rho\right)H_{0}\left(\rho\right)\right],\nonumber 
\end{align}
\begin{align}
\int_{\rho}^{\infty}\dfrac{J_{1}\left(x\right)}{x^{7}}\dd x= & \frac{6}{35}\frac{J_{1}\left(\rho\right)}{\rho^{6}}+\frac{1}{35}\frac{J_{1}^{\prime}\left(\rho\right)}{\rho^{5}}-\frac{4J_{1}\left(\rho\right)}{525\rho^{4}}-\frac{J_{1}^{\prime}\left(\rho\right)}{525\rho^{3}}+\frac{J_{0}\left(\rho\right)}{1575\rho}+\frac{J_{1}\left(\rho\right)}{1575\rho^{2}}\label{eq:int_J1_x7}\\
 & -\frac{1}{1575}-\frac{J_{1}\left(\rho\right)}{1575}+\frac{\rho J_{0}\left(\rho\right)}{1575}-\frac{\pi\rho}{3150}\left[J_{0}\left(\rho\right)H_{1}\left(\rho\right)-J_{1}\left(\rho\right)H_{0}\left(\rho\right)\right].\nonumber 
\end{align}
\end{lem}
\begin{proof}
We shall prove identities (\ref{eq:int_J1_x})--(\ref{eq:int_J1_x7})
sequentially.

Applying Lemma \ref{lem:int_J1_x_odd} with $n=1$, we have
\[
\int_{\rho}^{\infty}\dfrac{J_{1}\left(x\right)}{x^{3}}\dd x=\frac{2}{3}\frac{J_{1}\left(\rho\right)}{\rho^{2}}+\frac{1}{3}\frac{J_{1}^{\prime}\left(\rho\right)}{\rho}-\frac{1}{3}\int_{\rho}^{\infty}\dfrac{J_{1}\left(x\right)}{x}\dd x,
\]
and, hence,
\begin{equation}
\int_{\rho}^{\infty}\dfrac{J_{1}\left(x\right)}{x}\dd x=2\frac{J_{1}\left(\rho\right)}{\rho^{2}}+\frac{J_{1}^{\prime}\left(\rho\right)}{\rho}-3\int_{\rho}^{\infty}\dfrac{J_{1}\left(x\right)}{x^{3}}\dd x.\label{eq:int_J1_x_prelim}
\end{equation}
Using Lemma \ref{lem:int_J1_x3}, (\ref{eq:int_J1_x_prelim}) becomes
(\ref{eq:int_J1_x}).

Application of Lemma \ref{lem:int_J1_x_odd} with $n=2$ yields
\begin{equation}
\int_{\rho}^{\infty}\dfrac{J_{1}\left(x\right)}{x^{5}}\dd x=\frac{4}{15}\frac{J_{1}\left(\rho\right)}{\rho^{4}}+\frac{1}{15}\frac{J_{1}^{\prime}\left(\rho\right)}{\rho^{3}}-\frac{1}{15}\int_{\rho}^{\infty}\dfrac{J_{1}\left(x\right)}{x^{3}}\dd x.\label{eq:int_J1_x5_prelim}
\end{equation}
With the help of Lemma \ref{lem:int_J1_x3}, (\ref{eq:int_J1_x5_prelim})
transforms into (\ref{eq:int_J1_x5}).

Finally, we use Lemma \ref{lem:int_J1_x_odd} with $n=3$ to obtain
\begin{equation}
\int_{\rho}^{\infty}\dfrac{J_{1}\left(x\right)}{x^{7}}\dd x=\frac{6}{35}\frac{J_{1}\left(\rho\right)}{\rho^{6}}+\frac{1}{35}\frac{J_{1}^{\prime}\left(\rho\right)}{\rho^{5}}-\frac{1}{35}\int_{\rho}^{\infty}\dfrac{J_{1}\left(x\right)}{x^{5}}\dd x.\label{eq:int_J1_x7_prelim}
\end{equation}
Substitution of the already proven identity (\ref{eq:int_J1_x5})
into (\ref{eq:int_J1_x7_prelim}) furnishes (\ref{eq:int_J1_x7}). 
\end{proof}
\begin{lem}
\label{lem:trig_ints}For $\alpha\in\mathbb{R}$, $m$, $n\in\mathbb{N}_{0}$,
we have the following identities
\begin{equation}
\int_{0}^{2\pi}\cos\left(\alpha\cos\theta\right)\cos^{2m+1}\theta\,\,\sin^{n}\theta \dd\theta=0,\label{eq:trig_int_cos1}
\end{equation}
\begin{equation}
\int_{0}^{2\pi}\cos\left(\alpha\cos\theta\right)\cos^{m}\theta\,\,\sin^{2n+1}\theta \dd\theta=0,\label{eq:trig_int_cos2}
\end{equation}
\begin{equation}
\int_{0}^{2\pi}\sin\left(\alpha\cos\theta\right)\cos^{m}\theta\,\,\sin^{2n+1}\theta \dd\theta=0,\label{eq:trig_int_sin1}
\end{equation}
\begin{equation}
\int_{0}^{2\pi}\sin\left(\alpha\cos\theta\right)\cos^{2m}\theta\,\,\sin^{n}\theta \dd\theta=0.\label{eq:trig_int_sin2}
\end{equation}
\end{lem}
\begin{proof}
We shall prove identities (\ref{eq:trig_int_cos1})--(\ref{eq:trig_int_sin2})
one after another.

Using periodicity of the integrand, we can shift the interval from
$\left(0,2\pi\right)$ to $\left(-\pi/2,3\pi/2\right)$ and further
split it in two: 
\begin{align}
\int_{0}^{2\pi}\cos\left(\alpha\cos\theta\right)\cos^{2m+1}\theta\,\,\sin^{n}\theta \dd\theta & =\int_{-\pi/2}^{3\pi/2}\cos\left(\alpha\cos\theta\right)\cos^{2m+1}\theta\,\,\sin^{n}\theta \dd\theta\label{eq:trig_int_cos1_prelim}\\
 & =\int_{-\pi/2}^{\pi/2}\ldots+\int_{\pi/2}^{3\pi/2}\ldots.\nonumber 
\end{align}
Performing the change of variable $y=\sin\theta$, $\dd y=\cos\theta \dd\theta$
in each of the integrals on the second line of (\ref{eq:trig_int_cos1_prelim}),
we have
\[
\int_{-\pi/2}^{\pi/2}\cos\left(\alpha\cos\theta\right)\cos^{2m+1}\theta\,\,\sin^{n}\theta \dd\theta=\int_{-1}^{1}\cos\left(\alpha\sqrt{1-y^{2}}\right)\left(1-y^{2}\right)^{m}y^{n}\dd y,
\]
\begin{align*}
\int_{\pi/2}^{3\pi/2}\cos\left(\alpha\cos\theta\right)\cos^{2m+1}\theta\,\,\sin^{n}\theta \dd\theta & =\int_{1}^{-1}\cos\left(-\alpha\sqrt{1-y^{2}}\right)\left(-\sqrt{1-y^{2}}\right)^{2m}y^{n}\dd y\\
 & =-\int_{-1}^{1}\cos\left(\alpha\sqrt{1-y^{2}}\right)\left(1-y^{2}\right)^{m}y^{n}\dd y.
\end{align*}
Here, we used the fact that $\cos\theta=\sqrt{1-\sin^{2}\theta}>0$
for $\theta\in\left(-\pi/2,\pi/2\right)$ and $\cos\theta=-\sqrt{1-\sin^{2}\theta}<0$
for $\theta\in\left(\pi/2,3\pi/2\right)$ which makes 
both integral quantities on the second line of (\ref{eq:trig_int_cos1_prelim})
opposite to each other in sign and thus entails (\ref{eq:trig_int_cos1}).

To show (\ref{eq:trig_int_cos2}), a useful change of variable is
$y=\cos\theta$, $\dd y=-\sin\theta \dd\theta$. This, upon splitting the integration range to into $\left(0,\pi\right)$ and $\left(\pi,2\pi\right)$, leads to
\begin{align*}
\int_{0}^{2\pi}\cos\left(\alpha\cos\theta\right)\cos^{m}\theta\,\,\sin^{2n+1}\theta \dd\theta&=-\int_{1}^{-1}\cos\left(\alpha y\right)y^{m}\left(1-y^{2}\right)^{n}\dd y - \int_{-1}^{1}\cos\left(\alpha y\right)y^{m}\left(1-y^{2}\right)^{n}\dd y\\
&=\int_{-1}^{1}\cos\left(\alpha y\right)y^{m}\left(1-y^{2}\right)^{n}\dd y - \int_{-1}^{1}\cos\left(\alpha y\right)y^{m}\left(1-y^{2}\right)^{n}\dd y\\
&=0.
\end{align*}

Similarly, we have
\begin{align*}
\int_{0}^{2\pi}\sin\left(\alpha\cos\theta\right)\cos^{m}\theta\,\,\sin^{2n+1}\theta \dd\theta&=-\int_{1}^{-1}\sin\left(\alpha y\right)y^{m}\left(1-y^{2}\right)^{n}\dd y - \int_{-1}^{1}\sin\left(\alpha y\right)y^{m}\left(1-y^{2}\right)^{n}\dd y\\
&=0,
\end{align*}
which proves (\ref{eq:trig_int_sin1}).

Finally, to show (\ref{eq:trig_int_sin2}), we use the fact that identity
(\ref{eq:trig_int_cos1}) holds, in particular, for any $\alpha\in\left[0,\beta\right]$
with arbitrary $\beta\geq0$. Integrating it in $\alpha$ over this
interval and interchanging the order of integration (permissible by
the regularity of the integrand and the finite integration range),
we obtain
\[
0=\int_{0}^{\beta}\int_{0}^{2\pi}\cos\left(\alpha\cos\theta\right)\cos^{2m+1}\theta\,\,\sin^{n}\theta \dd\theta \dd\alpha=\int_{0}^{2\pi}\sin\left(\beta\cos\theta\right)\cos^{2m}\theta\,\,\sin^{n}\theta \dd\theta.
\]
Since the same reasoning also works $\beta\leq0$ by working with
an interval $\left[\beta,0\right]$, identity (\ref{eq:trig_int_sin2})
is thus proved up to a change of the notation $\beta$ to $\alpha$.
\end{proof}

\end{document}